%% file: main.tex
\documentclass[letter,11pt]{article}
  \newcommand{\fullversion}[1]{#1}
\newcommand{\submversion}[1]{}
\fullversion{
\usepackage[colorlinks]{hyperref}
  \hypersetup{linkcolor=blue,filecolor=blue,citecolor=blue,urlcolor=blue}
}
  
\usepackage{xspace}
\fullversion{
\usepackage{fullpage}
}

\usepackage{verbatim}
\usepackage{todonotes}
\usepackage{mathtools}

\usepackage[symbol,perpage]{footmisc}

\newcommand{\rolo}[1]{{\color{blue} Rolo:#1}}
\newcommand{\prab}[1]{{\color{red} Prab:#1}}

\newcommand{\bra}[1]{\langle #1|}
\newcommand{\ket}[1]{|#1\rangle}

\newcommand{\sfC}{\mathsf{C}}
\newcommand{\sfb}{\mathsf{B}}
\newcommand{\sfprovsecstg}{\mathsf{P}}
\newcommand{\sfversecstg}{\mathsf{V}}

\newcommand{\aux}{{\sf Aux}}
\newcommand{\msg}{\mathsf{msg}}
\newcommand{\na}{\mathsf{N}/\mathsf{A}}

\input{headers}
\usepackage{cleveref}

\newcommand{\prover}{\prvr}

\newcommand{\smprot}{\mathsf{prot}}
\newcommand{\smslot}{\mathsf{slot}}
\newcommand{\final}{\mathsf{final}}
\newcommand{\smsecmsg}{4}

\ignore{
\newcommand{\prvr}{\mathsf{Prover}}
\newcommand{\vrfr}{\mathsf{Verifier}}

\newcommand{\rel}{\mathcal{R}}
}
\newcommand{\np}{\mathsf{NP}}
\newcommand{\qma}{\mathsf{QMA}}
\newcommand{\sch}{\mathcal{S}}
\newcommand{\rv}{\mathsf{X}}
\newcommand{\expct}{\mathbb{E}}
\newcommand{\accept}{\mathsf{Accept}}

\newcommand{\relwi}{\rel_{\mathsf{WI}}}

\newcommand{\noslot}{\mathsf{noslot}}
\newcommand{\yes}{\mathsf{yes}}
\newcommand{\no}{\mathsf{no}}
\newcommand{\slot}{\mathsf{slot}}
\newcommand{\Good}{\mathsf{Good}}
\newcommand{\Bad}{\mathsf{Bad}}
\newcommand{\randreg}{\mathbf{R}}
\newcommand{\simreg}{\mathbf{Sim}}
\newcommand{\verreg}{\mathbf{Ver}}
\newcommand{\decreg}{\mathbf{Dec}}
\newcommand{\auxreg}{\mathbf{Aux}}
\newcommand{\slotreg}{\mathbf{M}}
\newcommand{\bitreg}{\mathbf{B}}
\newcommand{\witreg}{\mathbf{W}}
\newcommand{\regX}{\mathbf{X}}

\newcommand{\transreg}{\mathbf{T}}

\newcommand{\regB}{\mathbf{B}}
\newcommand{\ot}{\mathsf{OT}}
\newcommand{\smot}{\mathsf{ot}}
\newcommand{\advantage}{\mathsf{Adv}}
\newcommand{\game}{\mathsf{G}}
\newcommand{\quant}{\mathcal{Q}}

\newcommand{\advgen}{\mathsf{AdviceGen}}

\usepackage{breakcites}

\submversion{
\setcounter{secnumdepth}{3}
}

\fullversion{
\title{On the Concurrent Composition of Quantum Zero-Knowledge}
}
\submversion{
\title{On the Concurrent Composition of\\ Quantum Zero-Knowledge}
}
\fullversion{
\author{Prabhanjan Ananth\thanks{UC Santa Barbara. Email: \texttt{prabhanjan@cs.ucsb.edu}} \and Kai-Min Chung\thanks{Academia Sinica, Taiwan. Email: \texttt{kmchung@iis.sinica.edu.tw}} \and Rolando L. La Placa\thanks{MIT. Email: \texttt{rlaplaca@mit.edu}}}
}
\submversion{

\author{Prabhanjan Ananth\inst{1} \and Kai-Min Chung\inst{2} \and Rolando L. La Placa\inst{3}} 

\institute{UCSB, Santa Barbara, USA\\ \texttt{prabhanjan@cs.ucsb.edu} \and Academia Sinica, Taiwan\\ \texttt{kmchung@iis.sinica.edu.tw}  \and MIT, Cambridge, USA\\ \texttt{rlaplaca@mit.edu}} }
\date{}
\submversion{
\pagestyle{plain}
}

\begin{document}

\maketitle

\submversion{
\renewcommand{\inst}{{\bf y}}
}
\fullversion{
\newcommand{\inst}{{\bf y}}
}

\begin{abstract}
\noindent We study the notion of zero-knowledge secure against quantum polynomial-time verifiers (referred to as quantum zero-knowledge) in the concurrent composition setting. %In this setting, a prover can simultaneously interact with multiple verifiers and we require that zero-knowledge should hold even if all the verifiers are controlled by a malicious quantum polynomial-time adversary. 
Despite being extensively studied in the classical setting, concurrent composition in the quantum setting has hardly been studied. \par We initiate a formal study of concurrent quantum zero-knowledge. Our results are as follows:
\begin{itemize}
    \item {\bf Bounded Concurrent QZK for NP and QMA}: Assuming post-quantum one-way functions, there exists a quantum zero-knowledge proof system for NP in the bounded concurrent setting. In this setting, we fix a priori the number of verifiers that can simultaneously interact with the prover. Under the same assumption, we also show that there exists a quantum zero-knowledge proof system for QMA in the bounded concurrency setting.
    \item {\bf Quantum Proofs of Knowledge}: Assuming quantum hardness of learning with errors (QLWE), there exists a bounded concurrent zero-knowledge proof system for NP satisfying quantum proof of knowledge property. 
    \par Our extraction mechanism simultaneously allows for extraction probability to be negligibly close to acceptance probability ({\em extractability}) and also ensures that the prover's state after extraction is statistically close to the prover's state after interacting with the verifier ({\em simulatability}).\\ Even in the standalone setting, the seminal work of [Unruh EUROCRYPT'12], and all its followups, satisfied a weaker version of extractability property and moreover, did not achieve simulatability. Our result yields a proof of {\em quantum knowledge} system for QMA with better parameters than prior works.  
  
\end{itemize}
\end{abstract}

\fullversion{
\newpage
\setcounter{tocdepth}{2}
\tableofcontents
\newpage
}
\input{Intro/intro}

\fullversion{
\input{Intro/prelims}
}
\input{Definitions/def}

\fullversion{
\input{Concurrent_QZK/construction}

\input{QPoK/postot}

%\input{standalonepok}

\input{QPoK/pok}
\input{QMA/newQMA}
}

\submversion{
\input{Submission/submvertechnsections}
}

%\input{coinflip}
%\input{fullczk}

%\fullversion{
\section*{Acknowledgements}
We thank Abhishek Jain for many enlightening discussions, Zhengzhong Jin for patiently answering questions regarding~\cite{GJJM20}, Dakshita Khurana for suggestions on constructing oblivious transfer, Ran Canetti for giving an overview of existing classical concurrent ZK techniques, Aram Harrow and Takashi Yamakawa for discussions on the assumption of cloning security (included in a previous version of this paper) and Andrea Coladangelo for clarifications regarding~\cite{CVZ20}. RL was funded by NSF grant CCF-1729369. MIT-CTP/5289.
%}

\bibliographystyle{alpha}
\bibliography{crypto}

\submversion{

% \noindent {\bf {\Large Supplementary Material}}
% \ \\
% \input{Intro/prelims}
% \section{Bounded Concurrent QZK: Proofs}\label{sec:bcqzkproof}
% \input{Concurrent QZK/bcqzkproof}

% \input{QPoK/postot}
% \input{QPoK/pok}

}

%\appendix
%\input{appendix}

\end{document}

%% file: headers.tex
\fullversion{
\usepackage{amsthm}
\newtheorem{theorem}{Theorem}

\newtheorem{definition}[theorem]{Definition}
\newtheorem{lemma}[theorem]{Lemma}
\newtheorem{claim}[theorem]{Claim}
\newtheorem{remark}[theorem]{Remark}

}

\usepackage{amsmath,amsfonts,amssymb}

\usepackage[utf8]{inputenc}
\usepackage{tabularx}
\usepackage[normalem]{ulem}
\usepackage{setspace}
\usepackage{fancyhdr}
\usepackage{lastpage}
\usepackage{extramarks}
\usepackage{chngpage}
\usepackage{soul}
\usepackage{xspace}
\usepackage{bm}
\usepackage{nicefrac}
\usepackage{paralist}
\usepackage{enumitem}
\usepackage{tikz}

\newcommand{\bfs}{\mathbf{s}}
\newcommand{\bfA}{\mathbf{A}}
\newcommand{\bfe}{\mathbf{e}}

\newcommand{\bfu}{\mathbf{u}}

\newcommand{\srot}{\mathsf{SROT}}
\newcommand{\ssot}{\mathsf{SSOT}}

\newcommand{\secparam}{\lambda}
\newcommand{\negl}{\mathsf{negl}}

\newcommand{\ignore}[1]{}

\newcommand{\poly}{\mathrm{poly}}

\newcommand{\distr}{\mathcal{D}}

\newcommand{\size}{\mathsf{size}}
\newcommand{\amplifier}{\mathsf{Amplifier}}

\newcommand{\gldec}{\mathsf{GLDec}}
\newcommand{\predictor}{\mathcal{P}}

\newcommand{\cP}{\mathcal{P}}

\newcommand{\E}{\mathsf{E}}
\newcommand{\cS}{\mathcal{S}}

\newcommand{\bfM}{\mathbf{M}}

\newcommand{\prvr}{P}
\newcommand{\vrfr}{V}

\newcommand{\rel}{\mathcal{R}}
\newcommand{\simr}{\mathsf{Sim}}

\newcommand{\view}{\mathsf{View}}

\newcommand{\extractor}{\mathsf{Ext}}

\newcommand{\state}{\mathsf{st}}

\newcommand{\bfr}{\mathbf{r}}

\newcommand{\expt}{\mathsf{Expt}}
\newcommand{\adversary}{\mathcal{A}}  % Adversary
\newcommand{\hybrid}{\mathsf{Hyb}} % Hybrid
\newcommand{\prob}{\mathsf{Pr}} % Probability
\newcommand{\dec}{\mathsf{Dec}}

\newcommand{\sender}{\mathsf{S}}
\newcommand{\receiver}{\mathsf{R}}

\newcommand{\protwi}{\Pi_{\mathsf{WI}}}

\newcommand{\eval}{\mathsf{Eval}}

\newcommand{\bfc}{\mathbf{c}}
\newcommand{\comm}{\mathsf{Comm}}

\newcommand{\rnd}{t}

\newcommand{\lang}{\mathcal{L}}
\newcommand{\Simu}{\mathsf{Sim}}

\newcommand{\tr}{\mathsf{Tr}}
\newcommand{\cE}{\mathcal{E}}
\newcommand{\cD}{\mathcal{D}}
\newcommand{\cF}{\mathcal{F}}
\newcommand{\cH}{\mathcal{H}}

\newcommand{\RegP}{\mathrm{R}_\prvr}
\newcommand{\RegV}{\mathrm{R}_\vrfr}
\newcommand{\RegM}{\mathrm{R}_{\mathsf{M}}}

\newcommand{\cA}{\mathcal{A}}

\newcommand{\Id}{\mathsf{Id}}

\newcommand{\bfN}{\mathsf{N}}
\newcommand{\decision}{\mathsf{decision}}
\newcommand{\zk}{\mathsf{zk}}

\newcommand{\viewR}{\mathsf{View}_{R^*}}
\newcommand{\cB}{\mathcal{B}}
\newcommand{\bfX}{\mathbf{X}}

%% file: Intro/intro.tex
\section{Introduction}
\noindent Zero-knowledge~\cite{GMR} is one of the foundational concepts in  cryptography. A zero-knowledge system for NP is an interactive protocol between a prover $P$, who receives as input an instance $x$ and a witness $w$, and a verifier $V$ who receives as input an instance $x$. The (classical) zero-knowledge property roughly states that the view of the malicious probabilistic polynomial-time verifier $V^*$ generated after interacting with the prover $P$ can be simulated by a PPT simulator, who doesn't know the witness $w$.

\paragraph{Protocol Composition in the Quantum Setting.} Typical zero-knowledge proof systems only focus on the case when the malicious verifier is classical. The potential threat of quantum computers  forces us to revisit this definition. There are already many works~\cite{ARU14,BJSW16,BG19,BS20,ALP20a,VZ20,ABGKM20}, starting with the work of Watrous~\cite{Wat09}, that consider the definition of zero-knowledge against verifiers modeled as quantum polynomial-time (QPT) algorithms; henceforth this definition will be referred to as quantum zero-knowledge. However, most of these works study quantum zero-knowledge only in the standalone setting. These constructions work under the assumption that the designed protocols work in isolation. That is, a standalone protocol is one that only guarantees security if the parties participating in an execution of this protocol do not partake in any other protocol execution. This is an unrealistic assumption. Indeed, the standalone setting has been questioned in the classical cryptography literature by a large number of works~\cite{DS98,DO99,Can01,Can02,CF01,RK99,BS05,DNS04,PRS02,Lin03,Pass04,PV08,PTV14,GJRV13,CLP15,FKP19} that have focussed on designing cryptographic protocols that still guarantee security even when composed with the other protocols. 
\par A natural question to ask is whether there exist {\em quantum} zero-knowledge protocols (without any setup) that still guarantee security under composition. Barring a few works~\cite{Unruh10,JKMR06,ABGKM20}, this direction has largely been unaddressed. The couple of works~\cite{JKMR06,ABGKM20} that do address composition only focus on parallel composition; in this setting, all the verifiers interacting with the prover should send the $i^{th}$ round messages before the $(i+1)^{th}$ round begins. The setting of parallel composition is quite restrictive; it disallows the adversarial verifiers from arbitrarily interleaving their messages with the prover. A more reasonable scenario, also referred to as {\em concurrent composition}, would be to allow the adversarial verifiers to choose the scheduling of their messages in any order they desire.  
%Jain et al.~\cite{JKMR06} proved the impossibility of three-round zero-knowledge protocols for NP under parallel composition. Agrawal et al.~\cite{ABGKM20} provided a positive result on constant round parallel zero-knowledge for NP with non-black box simulation. 
So far, there has been no work that addresses concurrent composition in the quantum setting. 

\paragraph{Concurrent Quantum Zero-Knowledge.} In the concurrent setting, quantum zero-knowledge is defined as follows: there is a single prover, who on input instance-witness pair $(x,w)$, can simultaneously interact with
multiple verifiers, where all these verifiers are controlled by a single malicious quantum polynomial-time adversary. All the verifiers can potentially share an entangled state. Moreover, they can arbitrarily interleave their messages when they interact with the prover. For example, suppose the prover sends a message to the first verifier, instead of responding, it could let the second verifier send a message, after which the third verifier interacts with the prover and so on.
\par We say that zero-knowledge in this setting holds if there exists a quantum polynomial-time simulator (with access to the initial quantum state of all the verifiers) that can simultaneously simulate the interaction between the prover and all the verifiers.

We ask the following question in this work:
\begin{quote}
\begin{center}
    {\em Do there exist quantum zero-knowledge proof systems that are secure under concurrent composition?}
\end{center}
\end{quote}

\subsection{Our Contributions}

\paragraph{Bounded Concurrent QZK for NP.} We initiate a formal study of concurrent composition in the quantum setting. We work in the bounded concurrent setting: where the prover interacts only with a bounded number of verifiers where this bound is fixed at the time of protocol specification. This setting has been well studied in the classical concurrency literature~\cite{Lin03,PR03,Pass04,PTW09}. Moreover, we note that the only other existing work that constructs quantum zero-knowledge against multiple verifiers albeit in the parallel composition setting, namely~\cite{ABGKM20}\footnote{They achieve bounded parallel ZK under the assumption of quantum learning with errors and circular security assumption in constant rounds. While the notion they consider is sufficient for achieving MPC, the parallel QZK constructed by~\cite{ABGKM20} has the drawback that the simulator aborts even if one of the verifiers abort. Whereas the notion of bounded concurrent QZK we consider allows for the simulation to proceed even if one of the sessions abort. On the downside, our protocol runs in polynomially many rounds.}, also works in the bounded setting. We prove the following. 

\begin{theorem}[Informal]
\label{thm:cqzk}
Assuming the existence of post-quantum one-way functions\footnote{That is, one-way functions secure against (non-uniform) quantum polynomial-time algorithms.}, there exists a bounded concurrent quantum zero-knowledge proof system for NP. Additionally, our protocol is a public coin proof system. 
\end{theorem}

\noindent Our construction satisfies quantum black-box zero-knowledge\footnote{The simulator has oracle access to the unitary $V$ and $V^{\dagger}$, where $V$ is the verifier.}. We note that achieving public-coin {\em unbounded} concurrrent ZK is impossible~\cite{PTW09} even in the classical setting. %Hence, placing an a priori upper bound on the number of verifiers is the best we could hope for in the public coin setting. 
%\par Our construction of bounded concurrent quantum zero-knowledge is identical to the construction of Pass, Tseng and Wikstr\"om~\cite{PTW09} (modulo the setting of parameters). However, it turns out that the proof of~\cite{PTW09} is not quantum-friendly. So we first propose a different classical bounded concurrent simulation technique that is quantum-friendly. We then show how to port it to the quantum setting.  

\paragraph{Quantum Proofs of Knowledge.} Our construction, described above, only satisfies the standard soundness guarantee. A more desirable property is quantum proof of knowledge. Roughly speaking, proof of knowledge states the following: suppose a malicious (computationally unbounded) prover can convince a verifier to accept an instance $x$ with probability $\varepsilon$. Let the state of the prover at the end of interaction with the verifier be $\ket{\Psi}$\footnote{We work in the purified picture and thus we can assume that the output of the prover is a pure state.}. Then there exists an efficient extractor, with black-box access to the prover, that can output a witness $w$ for $x$ with probability $\delta$. Additionally, it also outputs a quantum state $\ket{\Phi}$. Ideally, we require the following two conditions to hold: (i) $|\varepsilon-\delta|$ is negligible and, (ii) the states $\ket{\Psi}$ and $\ket{\Phi}$ are close in trace distance; this property is also referred to as simulatability property. Unruh~\cite{Unruh12} presented a construction of quantum proofs of knowledge; their construction satisfies (i) but not (ii). Indeed, the prover's state, after it interacts with the extractor, could be completely destroyed. Condition (ii) is especially important if we were to use quantum proofs of knowledge protocols as a sub-routine inside larger protocols, for instance in secure multiparty computation protocols. 
\par Since Unruh's work, there have been other works that present constructions that satisfy both the above conditions but they demonstrate extraction only against {\em computationally bounded} adversaries~\cite{HSS11,BS20,ALP20a}. Thus, it has been an important open problem to design quantum proofs of knowledge satisfying both of the above conditions.
\par We show the following.

\begin{theorem}[Informal]
Assuming that learning with errors is hard against QPT algorithms (QLWE), there exists a bounded concurrent quantum zero-knowledge proof system for NP satisfying quantum proofs of knowledge property. 
\end{theorem}

\noindent Unlike all of the previous quantum proof of knowledge protocols which make use of Unruh's rewinding technique, we make black-box use of Watrous rewinding lemma in conjunction with novel cryptographic tools to prove the above theorem. On the downside, our protocol runs in polynomially many rounds, while Unruh's technique works for the existing 3-message $\Sigma$ protocols.
 
\newcommand{\fld}{\mathbf{q}}

\paragraph{Bounded Concurrent QZK for QMA.} We also show how to extend our result to achieve bounded concurrrent zero-knowledge proof system for QMA~\cite{KSV02} (a quantum-analogue of MA). % For any language $L \in QMA$, every instance $x \in L$ is associated with a poly-sized quantum witness state $\ket{\psi}$ such that a verification circuit (associated with $L$) outputs 1 with probability $\geq \frac{3}{4}$ and for any $x \notin L$, there does exist any poly-sized quantum witness state $\ket{\psi}$ that passes the verification with probability $> \frac{1}{4}$. 
\par We show the following.

\begin{theorem}[Informal]
Assuming post-quantum one-way functions, there exists a bounded concurrent quantum zero-knowledge proof system for QMA. 
\end{theorem}

\noindent This improves upon the existing QZK protocols for QMA~\cite{BJSW16,BG19,CVZ20,BS20} which only guarantee security in the standalone setting. 
\par Our construction considers a simplified version of the framewok of~\cite{BJSW16}\footnote{For the reader familiar with~\cite{BJSW16}, we consider a coin-flipping protocol secure against explainable adversaries as against malicious adversaries as considered in~\cite{BJSW16}.} and instantiates the underlying primitives in their protocol with bounded concurrent secure constructions. %Specifically, we use our bounded concurrent QZK construction for NP for the above result.  %We also develop a bounded concurrent post-quantum coin-flipping protocol.
\par We could combine the recent work of Coladangelo et al.~\cite{CVZ20} with our quantum proof of knowledge system for NP to obtain a proof of {\em quantum} knowledge system for QMA. This result yields better parameters than the one guaranteed in prior works~\cite{CVZ20,BG19}. Specifically, if the malicious prover convinces the verifier with probability negligibly close to 1 then the extractor (in our result) can extract a state that is negligibly close to the witness state whereas the previous works did not have this guarantee.

\submversion{
\subsection{Guide to the Reader}
\noindent We present the overview of our results in the technical sections, just before presenting a formal description of the results. 

\begin{itemize}\setlength\itemsep{1em}
    \item In Section~\ref{sec:cqzk}, we present the definitions of concurrent QZK proof systems for NP and QMA. In the same section, we present definitions of quantum proof of knowledge.
    \item {\bf Bounded Concurrent QZK}: In Section~\ref{sec:bcqzknp}, we present the construction of bounded concurrent QZK for NP. We first begin with an overview of the construction and then present the formal construction in the same section. The proofs are presented in the Appendix (see the relevant references at the end of Section~\ref{sec:bcqzknp}). 
    \item {\bf QZK Proof of Knowledge}: In Section~\ref{sec:qpok}, we present the construction of bounded concurrent QZK proof of knowledge for NP. We first begin with an overview of the construction and then present the formal construction in the same section. This construction involves the tool of oblivious transfer; we present the definition and the construction of oblivious transfer in the Appendix.
    \item {\bf Bounded Concurrent QZK for QMA}: Finally, we present a construction of bounded concurrent QZK for QMA in Section~\ref{sec:bcqzk:qma}. 
\end{itemize}

}

\fullversion{
\input{Overview/overview}
}

%% file: Overview/overview.tex
\subsection{Technical Overview}
\label{sec:overview}
\noindent We highlight the main ideas behind the constructions of the bounded concurrent QZK and quantum proof of knowledge. We omit the overview for the construction of bounded concurrent QZK for QMA. 

\input{Overview/bcqzk-overview}
\input{Overview/stdpok-overview}

\input{Overview/concpok-overview}

\subsection{Organization}
\begin{itemize}
    \item We present the necessary preliminaries -- including the notation used in the paper, basics on quantum computing, definitions of commitments and Watrous rewinding lemma -- in Section~\ref{sec:prelims}. 
    \item We present the definitions of concurrent quantum zero-knowledge for both NP and QMA in Section~\ref{sec:cqzk}. We also present the definition of quantum proof of knowledge in the same section. 
    \item Then, we provide the construction of bounded concurrent QZK for NP in Section~\ref{sec:BCQZK}. 
    \item We then focus on construct a quantum proof of knowledge in the bounded concurrent QZK setting. The main tool used in this construction is an oblivious transfer protocol; we present the definition and the construction of this oblivious transfer protocol in Section~\ref{sec:statOT}. 
    \par In Section~\ref{sec:pok}, we present the construction of quantum proof of knowledge in two steps. First we present a construction in the standalone setting. We then extend this construction to the bounded concurrent setting. 
    \item Finally, we present a construction of bounded concurent QZK for QMA in Section~\ref{sec:bcqzk:qma}. 
    
\end{itemize}

%% file: Overview/bcqzk-overview.tex
\fullversion{
\subsubsection{Bounded Concurrent QZK for NP}
\label{sec:bcqzk}
}
\submversion{
\subsection{Bounded Concurrent QZK for NP}
\label{sec:bcqzk}
}
\paragraph{Black Box QZK via Watrous Rewinding.} The traditional rewinding technique that has been used to prove powerful results on classical zero-knowledge cannot be easily ported to the quantum setting. The fundamental reason behind this difficulty is the fact that to carry out  rewinding, it is necessary to clone the state of the verifier. While cloning comes for free in the classical setting, the no-cloning theorem of quantum mechanics prevents us from being able to clone arbitrary states. Nonetheless, the seminal work of Watrous~\cite{Wat09} demonstrates that there are rewinding techniques that are amenable to the quantum setting. Watrous used this technique to present the first construction of quantum zero-knowledge for NP. This technique is so powerful that all quantum zero-knowledge protocols known so far (including the ones with non-black box simulation~\cite{BS20,ABGKM20}!) either implicitly or explicitly use this technique. 
\par We can abstractly think of Watrous technique as follows: to prove that a classical protocol is quantum zero-knowledge, first come up with a (classical) PPT simulator that  simulates a (classical) malicious PPT verifier. The classical simulator needs to satisfy the following two conditions: 
\begin{itemize} 

\item {\bf Oblivious Rewinding}: There is a distribution induced on the decision bits of the simulator to rewind in any given round $i$. This distribution could potentially depend on the randomness of the simulator and also the state of the verifier. \par The oblivious rewinding condition requires that this distribution should be independent of the state of the verifier. That is, this distribution should remain the same irrespective of the state of the verifier\footnote{A slightly weaker property where the distribution is ``{\em approximately"} independent of the state of the verifier also suffices. }.  

\item {\bf No-recording}: Before rewinding any round, the simulator could record (or remember) the transcript generated so far. This recorded transcript along with the rewound transcript will be used for simulation. For instance, in Goldreich and Kahan~\cite{GK96a}, the simulator first commits to garbage values and then waits for the verifier to decommit its challenges. The simulator then records the decommitments before rewinding and then changing its own commitments based on the decommitted values. \par The no-recording condition requires the following to hold: in order for the simulator to rewind from point $i$ to point $j$ ($i > j$), the simulator needs to forget the transcript generated from $j^{th}$ round to the $i^{th}$ round. Note that the simulator of~\cite{GK96a} does not satisfy the no-recording condition.

\end{itemize}
Once such a classical simulator is identified, we can then simulate quantum verifiers as follows: run the classical simulator and the quantum verifier\footnote{Without loss of generality, we can consider verifiers whose next message functions are implemented as unitaries and they perform all the measurements in the end.} in superposition and then at the end of each round, measure the appropriate register to figure out whether to rewind or not. The fact that the distribution associated with the decision bits are independent of the verifier's state is used to argue that the state, after measuring the decision register, is essentially not disturbed. Using this fact, we can then reverse the computation and go back to an earlier round. Once the computation is reversed (or rewound to an earlier round), the simulator forgets all the messages exchanged from the point -- to which its being rewound to -- until the current round. %And thus, the no-recording condition is required for successful simulation to hold. Watrous observed that there are such classical simulators for the graph isomorphism and 3-coloring protocols of Goldwasser, Micali and Widgerson~\cite{GMW86}.

\paragraph{Incompatibility of Existing Concurrent ZK Techniques.} To realize our goal of building bounded concurrent QZK, a natural direction to pursue is to look for classical concurrent ZK protocols with the guarantee that the classical simulator satisfies both the oblivious rewinding and no-recording conditions. However, most known classical concurrent ZK techniques are such that they satisfy one of these two conditions but not both. For example, the seminal work of~\cite{PRS02} proposes a concurrent ZK protocol and the simulator they describe satisfies the oblivious rewinding condition but not the no-recording condition. More relevant to our work is the work of Pass et al.~\cite{PTW09}, who construct a bounded concurrent ZK protocol whose simulator satisfies the no-recording condition but not the oblivious rewinding condition. 
\par In more detail, at every round, the simulator (as described in~\cite{PTW09}) makes a decision to rewind based on which session verifier sends a message in that round. This means that the probability of whether the simulator rewinds any given round depends on the scheduling of the messages of the verifiers. Unfortunately, the scheduling itself could be a function of the state of the verifier. The malicious verifier could look at the first bit of its auxiliary state. If it is 0, it will ask the first session verifier to send a message and if it is 1, it will ask the second session verifier to send a message and so on. This means that a simulator's decision to rewind could depend on the state of the verifier.

\paragraph{Bounded Concurrent QZK.} We now discuss our construction of bounded concurrent QZK and how we overcome the aforementioned difficulties. Our construction is identical to the bounded concurrent (classical) ZK construction of Pass et al.~\cite{PTW09}, modulo the setting of parameters. We recall their construction below. 
\par The protocol is divided into two phases. In the first phase, a sub-protocol, referred to as {\em slot}, is executed many times. We will fix the number of executions later when we do the analysis. In the second phase, the prover and the verifier execute a witness-indistinguishable proof system. 
\par In more detail, one execution of a slot is defined as follows: 
\begin{itemize} 

\item Prover sends a commitment of a random bit $b$ to the verifier. This commitment is generated using a statistically binding commitment scheme that guarantees hiding property against quantum polynomial-time adversaries (also referred to as quantum concealing). 

\item The verifier then sends a uniformly random bit $b'$ to the prover. 

\end{itemize}
We say that a slot is {\em matched} if $b=b'$. 
\par In the second phase, the prover convinces the verifier that either the instance is in the language or there is a large fraction, denoted by $\tau$, of matched slots. This is done using a proof system satisfying witness-indistinguishability property against efficient quantum verifiers. Of course, $\tau$ needs to be carefully set such that the simulator will be able to satisfy this constraint while a malicious prover cannot. Before we discuss the precise parameters, we first outline the simulator's strategy to prove zero-knowledge. As remarked earlier, the classical simulation strategy described in Pass et al.~\cite{PTW09} is incompatible with Watrous rewinding. We first discuss a new classical simulation strategy, that we call {\em block rewinding}, for this protocol and then we discuss how to combine this strategy along with Watrous rewinding to prove quantum zero-knowledge property of the above protocol. 

\paragraph{Block Rewinding.} Suppose $Q$ be the number of sessions the malicious verifier initiates with the simulator. Since this is a bounded concurrent setting, $Q$ is known even before the protocol is designed. Let $\ell_{\smprot}$ be the number of messages in the protocol. Note that the total number of messages exchanged in all the sessions is at most $\ell_{\smprot} \cdot Q$. We assume for a moment that the malicious verifier never aborts. Thus, the number of messages exchanged between the prover and the verifier is exactly $\ell_{\smprot} \cdot Q$. 
\par The simulator partitions the $\ell_{\smprot} \cdot Q$ messages into many blocks with each block being of a fixed size (we discuss the parameters later). The simulator then runs the verifier till the end of first block. At this point, it checks if this block contains a slot. Note that the verifier can stagger the messages of a particular session across the different blocks such that the first message of a slot is in one block but the second message of this slot could be in a different block. The simulator only considers those slots such that both the messages of these slots are contained inside the first block. Let the set of all the slots in the first block be denoted by $\mu(B_1)$, where $B_1$ denotes the first block. Now, the simulator picks a random slot from the set $\mu(B_1)$. It then checks if this slot is matched or not. That is, it checks if the bit committed in the slot equals the bit sent by the verifier. If indeed they are equal, it continues to the next block, else it rewinds to the beginning of the first block and then executes the first block again. Before rewinding, it forgets the transcript collected in the first block. It repeats this process until the slot it picked is matched. The simulator then moves on to the second block and repeats the entire process. When the simulator needs to compute a witness-indistinguishable proof for a session, it first checks if the fraction of matched slots for that particular session is at least $\tau$. If so, it uses this information to complete the proof. Otherwise, it aborts. 
\par It is easy to see why the no-recording condition is satisfied: the simulator never stores the messages sent in a block. Let us now analyze why the oblivious rewinding condition is satisfied. Suppose we are guaranteed that in every block there is at least one slot. Then, we claim that the probability that the simulator rewinds is $\frac{1}{2} \pm \negl(\secparam)$, where $\negl$ is a negligible function and $\secparam$ is the security parameter. This is because the simulator rewinds only if the slot is not matched and the probability that a slot is not matched is precisely $\frac{1}{2} \pm \negl(\secparam)$, from the hiding property of the commitment scheme. If we can show that every block contains a slot, then the oblivious rewinding condition would also be satisfied.\\

\noindent \textsc{Absence of Slots and Aborting Issues}: We glossed over a couple of issues in the above description. Firstly, the malicious verifier could abort all the sessions in some block. Moreover, it can also stagger the messages across blocks such that there are blocks that contain no slots. In either of the above two cases, the simulator will not rewind these blocks and this violates the oblivious rewinding condition: the decision to rewind would be based on whether the verifier aborted or whether there were any slots within a block. In turn, these two conditions could depend on the state of the verifier.
\par To overcome these two issues, we fix the simulator as follows: at the end of every block, it checks if there are any slots inside this block. If there are slots available, then the simulator continues as detailed above. Otherwise, it performs a dummy rewind: it picks a bit uniformly at random and rewinds only if the bit is 0. If the bit is 1, it continues its execution. This ensures that the simulator will rewind with probability $\frac{1}{2} \pm \negl(\secparam)$ irrespective of whether there are any slots inside a block. Thus, with this fix, the oblivious rewinding condition is satisfied as well. \\

\noindent \textsc{Parameters and Analysis}: We now discuss the parameters associated with the system. We set the number of slots in the system to be $120Q^7 \secparam$. We set $\tau$ to be $\lfloor \frac{60Q^7 \secparam + Q^4 \secparam}{120Q^7 \secparam} \rfloor$. We set the number of blocks to be $24Q^6 \secparam$. Thus, the size of each block is $\lfloor \frac{120Q^7 \secparam}{24Q^6 \secparam} \rfloor$. Recall that the reason why we need to set these parameters carefully is to ensure that the malicious prover cannot match more than $\tau$ slots with better than negligible probability whereas the simulator can beat this threshold with overwhelming probability.
\par We now argue that the classical simulator can successfully simulate all the $Q$ sessions. To simulate any given session, say the $i^{th}$ session, the number of matched slots needs to be at least $60Q^7 \secparam + Q^4 \secparam$. Note that the number of blocks is $24Q^6 \secparam$; the best case scenario is that each of these blocks contain at least one slot of the $i^{th}$ session and the simulator picks this slot every time. Even in this best case scenario, the simulator can match at most $24Q^6 \secparam$ slots and thus, there still would remain $60Q^7 \secparam + Q^4 \secparam - 24Q^6 \secparam$ number of slots to be matched. Moreover, even the likelihood of this best case scenario is quite low. 
\par Instead, we argue the following: 
\begin{itemize}
    \item The simulator only needs to match $3Q^{4} \secparam$ number of slots for the $i^{th}$ session. We argue that with overwhelming probability, there are $3Q^{4} \secparam$ blocks such that (i) there is at least one slot from the $i^{th}$ session and, (ii) the simulator happens to choose a slot belonging to this session in each of these blocks.   
    \item Roughly, $\frac{120Q^7 \secparam - 3Q^{4} \secparam}{2} \gg 60Q^7 \secparam - 2Q^4 \secparam$ number of slots are matched by luck, even without the simulator picking these slots and trying to match. This follows from the fact that with probability $\frac{1}{2}$, a slot is matched and the number of remaining slots that need to be matched are $120Q^7 \secparam - 3Q^{4} \secparam$. 
\end{itemize}
\noindent From the above two bullet points, it follows that with overwhelming probability, the total number of slots matched is at least $60Q^7 \secparam + Q^4 \secparam$.  
\par We note that although the simulation strategy of Pass et al.~\cite{PTW09} is quite different, their analysis follows the same template as above. \\
\ \\
\noindent \textsc{Simulation of Quantum Verifiers}: So far we have demonstrated a simulator that can simulate classical verifiers. We describe, at a high level, how to simulate quantum verifiers. The quantum simulator runs the classical simulator in superposition. At the end of every block, it measures a single-qubit register, denoted by $\decreg$, which indicates whether the simulator needs to rewind this block or not. If this register has 0, the simulator does not rewind, otherwise it rewinds. We can show that, no matter what the auxiliary state of the malicious verifier is, at the end of a block, the quantum state is of the following form:
$$\sqrt{p} \ket{0}_{\decreg} \ket{\Psi_{\mathsf{Good}}} + \sqrt{1-p} \ket{1}_{\decreg} \ket{\Psi_{\mathsf{Bad}}},$$
where $\ket{\Psi_{\mathsf{Good}}}$ is a superposition of all the transcripts where the chosen slot is matched and on the other hand, $\ket{\Psi_{\mathsf{Bad}}}$ is a superposition of all the transcripts where the chosen slot is not matched. Moreover, using the hiding property of the commitment scheme, we can argue that $|p-\frac{1}{2}| \leq \negl(\secparam)$. 
Then we can apply the Watrous rewinding lemma, to obtain a state that is close to $\ket{\Psi_{\mathsf{Good}}}$. This process is repeated for every block. At the end of the protocol, the simulator measures the registers containing the transcript of the protocol and outputs this along with the private state of the verifier.

%\noindent \textsc{Invoking WI.} In the proof of bounded concurrent quantum zero-knowledge, we need to invoke quantum witness-indistinguishability property of the underlying system. Coming up with a reduction that uses the bounded concurrent adversary to break quantum witness-indistinguishability turns out to be tricky. To see why, lets say $i$ be the session is such that we are invoking the witness-indistinguishability property of the $i^{th}$ session. Now, the reduction, while running the bounded concurrent verifier in superposition, needs to determine which of the messages to forward to the external WI verifier. But since the verifier is being run in superposition, the reduction does not know which of the messages correspond to the $i^{th}$ session messages. % To circumvent this issue, we consider an alternate definition of quantum WI that allows the reduction to forward the messages even without knowing which of them correspond to the WI messages.

%% file: Overview/stdpok-overview.tex
\fullversion{
\subsubsection{Standalone Quantum Proofs of Knowledge}
\label{sec:qpok}
}
\submversion{
\subsection{Standalone Quantum Proofs of Knowledge}
\label{sec:qpok}
}
\noindent Towards building a bounded-concurrent QZK system satisfying quantum proof of knowledge property, we first focus on the standalone QZK setting. The quantum proof of knowledge property roughly says the following: for every unbounded prover convincing a verifier to accept an instance $x$ with probability $p$, there exists an  extractor that outputs a witness $w$ with probability negligibly close to $p$ and it also outputs a state $\ket{\Phi}$ that is close (in trace distance) to the output state of the real prover. %In this overview, we only focus on the case where we demonstrate quantum proof of knowledge in the standalone QZK setting; the techniques can be appropriately generalized to the bounded concurrent setting.   
\par Our approach is to design a novel extraction mechanism that uses oblivious transfer to extract a bit from a quantum adversary. 

\paragraph{Main Tool: Statistical Receiver-Private Oblivious Transfer.} Our starting point is an oblivious transfer (OT) protocol~\cite{Rabin05}. This protocol is defined between two entities: a sender and a receiver. The sender has two bits $(m_0,m_1)$ and the receiver has a single bit $b$. At the end of the protocol, the receiver receives the bit $m_b$.
\par The security against malicious senders (receiver privacy) states that the sender should not be able to distinguish (with non-negligible probability) whether the receiver's bit is 0 or 1. The security against malicious receivers (also called sender privacy) states that there is a bit $b'$ such that the receiver cannot distinguish (with non-negligible probability) the case when the sender's input is $(m_0,m_1)$ versus the setting when the sender's input is $(m_{b'},m_{b'})$. %This bit $b'$ is a function of the sender's as well as the adversary's randomness. 
\par We require receiver privacy to hold against unbounded senders while we require sender privacy to hold against quantum polynomial-time receivers. The reason we require receiver privacy against unbounded senders is because our goal is to design extraction mechanism against computationally unbounded provers.
\par We postpone discussing the construction of statistical receiver-private oblivious transfer\submversion{ to the Appendix}. We will now see how to use this to achieve extraction.

\paragraph{One-bit Extraction with $\left(\frac{1}{2} \pm \negl \right)$-error.}  
\par We begin with a naive attempt to design the extraction mechanism for extracting a single secret bit, say $s$\footnote{For instance, $s$ could be the firt bit of the witness.}. The prover and the verifier execute the OT protocol; prover takes on the role of the OT sender and the verifier takes on the receiver's role. The prover picks bits $b$ and $\alpha$ uniformly at random and then sets the OT sender's input to be $(s,\alpha)$ if $b=0$, otherwise if $b=1$, it sets the OT sender's input to be $(\alpha,s)$. The verifier sets the receiver's bit to be 0. After the OT protocol ends, the prover sends the bit $b$. Note that if the bit $b$ picked by the prover was 0 then the verifier can successfully recover $s$, else it recovers $\alpha$. 
\par We first discuss the classical extraction process. The quantum extractor runs the classical extractor in superposition as we did in the case of quantum zero-knowledge. The extraction process proceeds as follows: the extractor picks a bit $\widetilde{b}$ uniformly at random and sets $\widetilde{b}$ to be the receiver's bit in the OT protocol. By the statistical receiver privacy property of OT, it follows that the probability that the extractor succeeds in recovering $s$ is negligibly close to $\frac{1}{2}$. Moreover, the success probability is independent of the initial state of the prover. This means that we can apply the Watrous rewinding lemma and amplify the success probability. \\

\noindent \textsc{Malicious Provers}: However, we missed a subtle issue: the malicious prover could misbehave. For instance, the prover can set the OT sender's input to be $(r,r)$ and thus, not use the secret bit $s$ at all.  
\par We resolve this issue by additionally requiring the prover to prove to the verifier that one of its inputs in the OT protocol is the secret bit\footnote{For now, assume that there exists a predicate that can check if $s$ is a valid secret bit.} $s$. This is realized by using a quantum zero-knowledge protocol, denoted by $\Pi$. 

\paragraph{Error amplification.} A malicious verifier can successfully recover the secret $s$ with probability $\frac{1}{2}$. To reduce the verifier's success probability, we execute the above process (i.e., first executing the OT protocol and then executing the ZK protocol) $\secparam$ number of times, where $\secparam$ is the security parameter. First, the prover will additively secret share the bit $s$ into secret shares $sh_1,\ldots,sh_{\secparam}$. It also samples the bits $b_1,\ldots,b_{\secparam}$ uniformly at random. In the $i^{th}$ execution, it sets the OT sender's input to be $(sh_i,\alpha_i)$ if $b_i=0$, otherwise it sets the OT sender's input to be $(\alpha_i,sh_i)$, where $\alpha_i$ is sampled uniformly at random. After all the OT protocols are executed, the prover is going to prove using a QZK protocol $\Pi$, as considered above, that the messages in the OT protocols were correctly computed. 
\par We first argue that even in this protocol, the extraction still succeeds with overwhelming probability. In each OT execution, the extractor applies Watrous rewinding, as before, to extract all the shares $sh_1,\ldots,sh_{\secparam}$. From this, it can recover $s$. All is left is to argue that this template satisfies quantum zero-knowledge property. It turns out that arguing this is challenging\footnote{We would like to point out that we are designing the standalone PoK protocol as a stepping stone towards the bounded concurrent PoK protocol. If one were to be interested in just the standalone setting, then it might be possible to avoid the subtelties described above by making use of a simulation-secure OT rather than an indistinguishable-secure OT. The reason why we use an indistinguishable-secure OT in the concurrent PoK setting instead of a simulation-secure OT is because we want to avoid using more than one simulator in the analysis; otherwise, we would have multiple simulators trying to rewind the verifier, making the analysis significantly complicated.}.

\paragraph{Challenges in Proving QZK and Distinguisher-Dependent Hybrids.} We first define the simulator as follows:
\begin{itemize} 

\item The simulator uses $(\alpha_i,\alpha_i)$ as the sender's input in the $i^{th}$ OT execution, where $\alpha_i$ is sampled uniformly at random. 

\item It then simulates the protocol $\Pi$. 

\end{itemize}
To prove that the output distribution of the simulated world is computationally indistinguishable from the real world, we adopt a hybrid argument. The first hybrid, $\hybrid_1$, corresponds to the real world. In the second hybrid, $\hybrid_2$, simulate the protocol $\Pi$. The indistinguishability of $\hybrid_1$ and $\hybrid_2$ follows from the QZK property of $\Pi$. Next, we define the third hybrid, $\hybrid_3$, that executes the simulator. To prove the indistinguishability of $\hybrid_2$ and $\hybrid_3$, we consider a sequence of intermediate hybrids, denoted by  $\{\hybrid_{2.j}\}_{j \in [\secparam]}$. Using this sequence of hybrids, we change the inputs in all the $\secparam$ OT executions one at a time. Finally, we define the third hybrid, $\hybrid_3$, that corresponds to the ideal world. Proving the indistinguishability of the consecutive hybrids, $\hybrid_{2.j}$ and $\hybrid_{2.j+1}$, in this sequence turns out to be challenging. 
\par The main issue is the following: suppose we are in the $j^{th}$ intermediate hybrid $\hybrid_{2.j}$, for $j \leq \secparam$. At this point, we have changed the inputs to the first $j$ OT executions and we are about to change the input to the $(j+1)^{th}$ OT. But what exactly are the inputs we are using for the first $j$ OT executions? It is unclear whether we use the input $(sh_i,sh_i)$ or the input $(\alpha_i,\alpha_i)$, for $i \leq j$, in the $i^{th}$ OT execution. Note that the OT security states that we can either switch the real sender's inputs to either $(sh_i,sh_i)$ or $(\alpha_i,\alpha_i)$, based on the sender's and the distinguisher's randomness. And hence, we define an {\em inefficient} intermediate hybrid, which is a function (not necessarily computable), that determines for every $i$, where $i \leq j$, whether to use $(sh_i,sh_i)$ or $(\alpha_i,\alpha_i)$. Moreover, {\em this hybrid depends on the distinguisher}, that distinguishes the two intermediate hybrids.  
\par The indistinguishability of the consecutive pair of inefficient hybrids, say $\hybrid_{2.j}$ and $\hybrid_{2.j+1}$, is proven by a non-uniform reduction that receives as input the advice corresponding to the first $j$ executions of OT, where the sender's inputs are correctly switched to either $(sh_i,sh_i)$ or $(\alpha_i,\alpha_i)$, for $i \leq j$. This in turn depends on the distinguisher distinguishing these two hybrids. Then, the reduction uses the $(j+1)^{th}$ OT execution in the protocol to break the sender privacy property of OT. If the two hybrids can be distinguished with non-negligible probability then the reduction can succeed with the same probability.  
\par In the hybrid  $\hybrid_{2.\secparam-1}$, we additionally include an abort condition: if the inputs in the first $\secparam-1$ OT executions are all switched to $(sh_i,sh_i)$ then we abort. We show that the probability that $\hybrid_{2.\secparam-1}$ aborts is negligible. This is necessary to argue that the verifier does not receive all the shares of the secret. 
\par Note that only the intermediate hybrids, namely $\{\hybrid_{2.j}\}_{j \in [\secparam]}$, are inefficient, and in particular, the final hybrid $\hybrid_3$ is still efficient.

\paragraph{Extraction of Multiple Bits.} To design a quantum proof of knowledge protocol, we need to be able to extract not just one bit, but multiple bits. To achieve this, we design the prover as follows: on input a witness $w$, it sequentially executes the above extraction template for each bit of the witness. That is, for every $i \in [\ell_w]$, where $\ell_w$ is the length of $w$, it additively secret shares $w_i$ into the shares $(sh_{i,1},\ldots,sh_{i,\secparam})$. It then invokes $\ell_w \cdot \secparam$ number of OT executions, where in the $(i,j)^{th}$ execution, it chooses the input $(sh_{i,j},\alpha_{i,j})$ if $b_{i,j}=0$, or the input $(\alpha_{i,j},sh_{i,j})$ if $b_{i,j}=1$, where $\alpha_{i,j},b_{i,j}$ are sampled uniformly at random. Finally, it uses a QZK protocol to prove that it behaved honestly in the earlier OT executions.
\par The proofs of quantum proof of knowledge and the QZK properties follow along the same lines as the single-bit extraction case. 

\fullversion{
\input{Overview/statot-overview}
}

%% file: Overview/statot-overview.tex
\subsubsection{Statistical Receiver-Private OT with Post-Quantum Security} 
\fullversion{
All that is left is to construct an oblivious transfer protocol that guarantees statistical indistinguishability property against malicious senders and indistinguishability property against QPT malicious receivers. We denote the protocol that we intend to construct to be $\Pi_{\srot}$. 
}
\submversion{
We first give an overview of our construction. 
}
\par The starting point of our construction is another oblivious transfer protocol, denoted by $\Pi_{\ssot}$, that has its properties flipped. That is, $\Pi_{\ssot}$ satisfies statistical indistinguishability property against malicious {\em receivers} and indistinguishability property against QPT malicious {\em senders}. The reason we start with this protocol is that we do know how to achieve this; Brakerski-D\"ottling~\cite{BD18}  constructed such a protocol from QLWE. 
\par Our approach is inspired from previous works~\cite{KKS18,GJJM20} that show how to construct statistical receiver-private OT from statistical sender-private OT. 
\par Our first attempt to construct $\Pi_{\srot}$ is the following: 
\begin{itemize}
    \item The sender of $\Pi_{\ssot}$ samples a random bit $r \xleftarrow{\$} \{0,1\}$. It takes the role of the receiver in the underlying $\Pi_{\srot}$. It then sends the first message of $\Pi_{\srot}$ with the receiver's message set to be $r$.  
    \item The receiver of $\Pi_{\ssot}$, on input choice bit $\beta$, samples another random bit $r'$. It takes the role of the sender in the underlying protocol $\Pi_{\ssot}$. It then sends the sender's message in $\Pi_{\ssot}$, where the sender's input in $\Pi_{\ssot}$ is set to be $(r',r' \oplus \beta)$. 
    \item After the end of the execution of $\Pi_{\ssot}$, the sender on input $(m_0,m_1)$, does the following: it recovers the message $\widetilde{r}$ from the underlying OT. It then sends $(\widetilde{r} \oplus m_0,\widetilde{r} \oplus r \oplus m_1)$ to the receiver.
\end{itemize}
\noindent If $\beta = 0$ then $\widetilde{r}=r'$ and so, the receiver can recover $m_0$. If $\beta=1$ then $\widetilde{r}=r' \oplus r$ and so, the receiver can recover $m_1$. 
\par The receiver privacy against computationally unbounded senders follows from the statistical sender privacy of the underlying two-round oblivious transfer protocol.
\par To prove sender privacy against QPT receivers, first let us make the previously described security notion more precise. The malicious receiver $R^*$, on input state $\ket{\Psi}$, interacts with the sender and produces an auxiliary state $\ket{\widetilde{\Psi}}$. During this interaction, the sender does not use $(m_0,m_1)$. The sender uses $(m_0,m_1)$ to compute the final round message. We define two games: in the first game, the adversary tries to distinguish $(m_0,m_1)$ versus $(m_0,m_0)$ and in the second game, the adversary tries to distinguish $(m_0,m_1)$ versus $(m_1,m_1)$. We say that oblivious transfer satisfies post-quantum computational sender privacy property if the malicious receiver cannot succeed in both the games with non-negligible advantage. 
\par A natural approach to prove that the malicious receiver cannot win both the games is to extract the bit $\beta$ from the malicious receiver; if $\beta=0$ then the receiver will not be able to succeed in the second game if $\beta=1$ then the receiver will not succeed in the first game. To ensure that we can extract the bit $\beta$ from the receiver, we additionally introduce an extraction phase to the protocol.\\

\noindent \textsc{Extraction Phase}: To design the extraction phase, we use the same technique we introduced earlier. The main difference is that instead of using statistical receiver-private OT, we instead use a statistical sender-private OT for extraction. 
\par In the extraction phase of $\Pi_{\srot}$, the sender and the receiver do the following: 
\begin{itemize}
    \item As before, the sender of $\Pi_{\srot}$, plays the role of the receiver of $\Pi_{\ssot}$ and the receiver of $\Pi_{\srot}$ plays the role of the sender of $\Pi_{\ssot}$. 
    \item The sender does the following: it samples a bit $b$ uniformly at random. It sets the $\Pi_{\ssot}$'s receiver's bit to be $b$.  
    \item The receiver, on the other hand, samples $\alpha \xleftarrow{\$} \{0,1\}$ and sets the $\Pi_{\ssot}$'s sender's input to be $(\beta,\alpha)$ with probability $\frac{1}{2}$ and $(\alpha,\beta)$ with probability $\frac{1}{2}$.
    \item At the end of the execution of $\Pi_{\ssot}$, the receiver reveals the location of $\beta$ -- i.e., it sends 0 if $(\beta,\alpha)$ was used in $\Pi_{\ssot}$ or it sends 1 if $(\alpha,\beta)$ was used.   
\end{itemize}
Note that if the location matched with $b$ then the sender can recover $\beta$, otherwise it cannot. With probability at most $\frac{1}{2}$, the sender can recover $\beta$.  We can use the same error amplification technique (via secret sharing) introduced earlier to reduce the probability of success of the malicious sender to be negligible. On the other hand, we can design an extractor that uses Watrous rewinding, as mentioned earlier to recover the bit $\beta$ with probability close to 1.\\

\noindent \textsc{Template.} Using the above ingredients, we now summarise the template to construct a statistical receiver-private oblivious transfer. 
\par The sender, on input $(m_0,m_1)$, and the receiver of $\Pi_{\srot}$ on input $\beta$, do the following: 
\begin{itemize} 

\item The sender and the receiver execute the extraction phase described above. The receiver uses its bit $\beta$ in the extraction phase. 

\item The sender and the receiver then execute $\Pi_{\ssot}$, where each party play the opposite role. The sender sets the input of the receiver in $\Pi_{\ssot}$ to be $r$, where $r \xleftarrow{\$} \{0,1\}$ and the receiver sets the input of the sender in $\Pi_{\ssot}$ to be $(r',r' \oplus \beta)$, where $r' \xleftarrow{\$} \{0,1\}$. After the end of the execution of $\Pi_{\ssot}$, the sender recovers $\widetilde{r}$. 

\item Of course, the receiver could have cheated and used a different $\beta$ in both the extraction phase and in the execution of $\Pi_{\ssot}$. To ensure that the receiver does not cheat, we force the receiver to prove that it used $\beta$ consistently. We use a computational argument system satisfying statistical zero-knowledge property for this step. 
\item Once the sender gets convinced that the receiver did not cheat, it sends $(\widetilde{r} \oplus m_0,\widetilde{r} \oplus r \oplus m_1)$ to the receiver. 

\end{itemize}
\noindent Finally, we show how to implement computational argument system satisfying statistical zero-knowledge property from QLWE. The idea is to start with a statistical NIZK computational argument system in the CRS model and then generate the CRS using a coin flipping protocol.   

%% file: Overview/concpok-overview.tex
\subsubsection{Quantum PoK in the Bounded Concurrent Setting}

Our construction of bounded concurrent quantum proof of knowledge is the same as the one described in~\Cref{sec:qpok}, except that we instantiate $\Pi$ using the bounded concurrent QZK protocol that we constructed in~\Cref{sec:BCQZK}\footnote{We emphasize that we use the specific bounded concurrent QZK protocol that we constructed earlier and we do not know how to provide a generic transformation.}. 
\par However, proving the bounded concurrent QZK protocol turns out to be even more challenging than the standalone setting. To grasp the underlying difficulties, let us revisit the proof of QZK in~\Cref{sec:qpok}. To prove the indistinguishability of the real and the ideal world, we first simulated the protocol $\Pi$. Since we are in the bounded concurrent setting, the simulator of $\Pi$ is now simultaneously simulating multiple sessions of the verifier. Then using a sequence of intermediate hybrids, we changed the inputs used in the OT executions of all the sessions one at a time. However, in the bounded concurrent setting, the OT messages can be interleaved with QZK messages. This means that the simulator of QZK could be rewinding the OT messages along with the QZK messages. This makes it difficult to invoke the security of OT.

%% file: Intro/prelims.tex
\section{Preliminaries}
\label{sec:prelims}
\fullversion{
\noindent We denote the security parameter by $\secparam$. We assume basic familiarity of cryptographic concepts. 
\par We denote (classical) computational indistiguishability of two distributions $\distr_0$ and $\distr_1$ by $\distr_0 \approx_{c,\varepsilon} \distr_1$. In the case when $\varepsilon$ is negligible, we drop $\varepsilon$ from this notation. We denote the process of an algorithm $A$ being executed on input a sample from a distribution $\distr$ by the notation $A(\distr)$.

\paragraph{Languages and Relations.} A language $\lang$ is a subset of $\{0,1\}^*$. A relation $\rel$ is a subset of $\{0,1\}^* \times \{0,1\}^*$. We use the following notation:
\begin{itemize}

\item Suppose $\rel$ is a relation. We define $\rel$ to be {\em efficiently decidable} if there exists an algorithm $A$ and fixed polynomial $p$ such that $(x,w) \in \rel$ if and only if $A(x,w)=1$ and the running time of $A$ is upper bounded by $p(|x|,|w|)$. 

\item Suppose $\rel$ is an efficiently decidable relation. We say that $\rel$ is a NP relation if $\lang(\rel)$ is a NP language, where $\lang(\rel)$ is defined as follows: $x \in \lang(R)$ if and only if there exists $w$ such that $(x,w) \in \rel$ and $|w| \leq p(|x|)$ for some fixed polynomial $p$. 

\end{itemize}
}

\begin{comment}
\subsection{Learning with Errors}
\label{sec:prelims:lwe}

\noindent In this work, we are interested in the decisional learning with errors (LWE) problem. This problem, parameterized by $n,m,q,\chi$, where $n,m,q \in \mathbb{N}$, and for a distribution $\chi$ supported over $\mathbb{Z}$ is to distinguish between the distributions $(\bfA,\bfA \bfs + \bfe)$ and $(\bfA,\bfu)$, where $\bfA \xleftarrow{\$} \mathbb{Z}_q^{m \times n},\bfs \xleftarrow{\$} \mathbb{Z}_q^{n \times 1},\bfe \xleftarrow{\$} \chi^{m \times 1}$ and $\bfu \leftarrow \mathbb{Z}_q^{m \times 1}$. Typical setting of $m$ is $n \log(q)$, but we also consider $m=\poly(n \log(q))$. 
\par We base the security of our constructions on the quantum hardness of learning with errors problem.
\end{comment}

\subsection{Notation and General Definitions}
\label{ssec:notation}

For completeness, we present some of the basic quantum definitions, for more details see \cite{nielsen2002quantum}.
\paragraph{Quantum states and channels.} Let $\cH$ be any finite Hilbert space, and let $L(\cH):=\{\cE:\cH \rightarrow \cH \}$ be the set of all linear operators from $\cH$ to itself (or endomorphism). Quantum states over $\cH$ are the positive semidefinite operators in $L(\cH)$ that have unit trace. %Quantum channels or quantum operations acting on quantum states over $\cH$ are completely positive trace preserving (CPTP) linear maps from $L(\cH)$ to $L(\cH')$ where $\cH'$ is any other finite dimensional Hilbert space.

A state over $\cH=\mathbb{C}^2$ is called a qubit. For any $n \in \mathbb{N}$, we refer to the quantum states over $\cH = (\mathbb{C}^2)^{\otimes n}$ as $n$-qubit quantum states. To perform a standard basis measurement on a qubit means projecting the qubit into $\{\ket{0},\ket{1}\}$. A quantum register is a collection of qubits. A classical register is a quantum register that is only able to store qubits in the computational basis.

A unitary quantum circuit is a sequence of unitary operations (unitary gates) acting on a fixed number of qubits. Measurements in the standard basis can be performed at the end of the unitary circuit. A (general) quantum circuit is a unitary quantum circuit with $2$ additional operations: $(1)$ a gate that adds an ancilla qubit to the system, and $(2)$ a gate that discards (trace-out) a qubit from the system. A quantum polynomial-time algorithm (QPT) is a non-uniform collection of quantum circuits $\{C_n\}_{n \in \mathbb{N}}$.

\paragraph{Quantum Computational Indistinguishability.} We define computational indistinguishability; we borrow the following definition from ~\cite{Wat09}. Roughly, the below definition states that two collections of quantum states $\{\rho_x\}$ and $\{\sigma_x\}$ are computationally indistinguishable if any quantum distinguisher, running in time polynomial in $|x|$, cannot distinguish $\rho_x$ from $\sigma_x$, where $x$ is sampled from some distribution. Moreover, the computational indistinguishability should hold even if the distinguisher has quantum advice (that might be entangled with $\rho_x$ and $\sigma_x$). 

\begin{definition}[Computational Indistinguishability of Quantum States] Let $I$ be an infinite subset $I \subset \{0,1\}^*$, let $p : \mathbb{N} \rightarrow \mathbb{N}$ be a polynomially bounded function, and let $\rho_{x}$ and $\sigma_x$ be $p(|x|)$-qubit states. We say that $\{\rho_{x}\}_{x \in I}$ and $\{\sigma_x\}_{x\in I}$ are \textbf{quantum computationally indistinguishable collections of quantum states} if for every QPT $\cE$ that outputs a single bit, any polynomially bounded  $q:\mathbb{N}\rightarrow \mathbb{N}$, and any auxiliary collection of $q(|x|)$-qubits states $\{\nu_x \}_{x \in I}$, and for all (but finitely many) $x \in I$, we have that
$$\left|\Pr\left[\cE(\rho_x\otimes \nu_x)=1\right]-\Pr\left[\cE(\sigma_x \otimes \nu_x)=1\right]\right| \leq \epsilon(|x|) $$
for some negligible function $\epsilon:\mathbb{N}\rightarrow [0,1]$. We use the following notation 
$$\rho_x \approx_{\quant,\epsilon} \sigma_x$$
and we ignore the $\epsilon$ when it is understood that it is a negligible function.
\end{definition}

\begin{comment}
\begin{definition}[Indistinguishability of channels] Let $I$ be an infinite subset $I \subset \{0,1\}^*$, let $p,q: \mathbb{N} \rightarrow \mathbb{N}$ be polynomially bounded functions, and let $\cD_x,\cF_x$
be quantum channels mapping $p(|x|)$-qubit states to $q(|x|)$-qubit states. We say that $\{\cD_x\}_{x \in I}$ and $\{\cF_x\}_{x \in I}$ are \textbf{quantum computationally indistinguishable collection of channels} if for every QPT $\cE$ that outputs a single bit, any polynomially bounded $t : \mathbb{N} \rightarrow \mathbb{N}$, any $p(|x|)+t(|x|)$-qubit quantum state $\rho$, and for all $x\in I$, we have that
$$ \left|\Pr\left[\cE\left((\cD_x\otimes \Id)(\rho)\right)=1\right]-\Pr\left[\cE\left((\cF_x\otimes \Id)(\rho)\right)=1\right]\right|\leq \epsilon(|x|) $$
for some negligible function $\epsilon:\mathbb{N}\rightarrow [0,1]$. We will use the following notation
$$ \cD_x(\cdot) \approx_{Q,\epsilon} \cF_x(\cdot)$$
and we ignore the $\epsilon$ when it is understood that it is a negligible function. 
\end{definition}
\end{comment}

\paragraph{Interactive Models.} We model an interactive protocol between a prover, $\prvr$, and a verifier, $\vrfr$, as follows. There are 2 registers $\RegP$ and $\RegV$ corresponding to the prover's and the verifier's private registers, as well as a message register, $\RegM$, which is used by both $\prvr$ and $\vrfr$ to send messages. In other words, both prover and verifier have access to the message register. We denote the size of a register $\mathsf{R}$ by $|\mathsf{R}|$ -- this is the number of bits or qubits that the register can store.  There are 3 different notions of interactive computation. 

%Our honest parties will perform classical protocols, but the adversaries will be allowed to perform quantum protocols with classical messages.

\begin{enumerate}
    \item \textbf{Classical protocol:} An interactive protocol is classical if $\RegP$, $\RegV$, and $\RegM$ are classical, and $\prvr$ and $\vrfr$ can only perform classical computation.
    \item \textbf{Quantum protocol with classical messages:} An interactive protocol is quantum with classical messages if either one of $\RegP$ or $\RegV$ is a quantum register, and $\RegM$ is classical. $\prvr$ and $\vrfr$ can perform quantum computations if their respective private register is quantum, but they can only send classical messages.
    \item \textbf{Quantum protocol:} $\RegP$, $\RegV$, and $\RegM$ are all quantum registers.  The prover performs quantum operations on $\RegP \otimes \RegM$ and the verifier performs quantum operations on $\RegV\otimes \RegM$.
\end{enumerate}
When a protocol has classical messages, we can assume that the adversarial party will also send classical messages. This is without loss of generality, because the honest party can enforce this condition by always measuring the message register in the computational basis before proceeding with its computations.

\paragraph{Notation.} We use the following notation in the rest of the paper.
\begin{itemize}
    \item $\langle P,V \rangle $ denotes the interactive protocol between the QPT algorithms $P$ and $V$. We denote the $\langle P(y_1),V(y_2) \rangle$ to be $(z_1,z_2)$, where $z_1$ is the prover's output and $z_2$ is the verifier's output. Sometimes we omit the prover's output and write this as $z \leftarrow \langle P(y_1),V(y_2) \rangle$ to indicate the output of the verifier to be $z$.
    \item $\view_V\left(\langle P(y_1),V(y_2) \rangle\right)$ denotes the view of the QPT algorithm $V$ in the protocol $\Pi$, where $y_1$ is the input of $P$ and $y_2$ is the input of $V$. In the classical case, the view includes the output of $V$ and the transcript of the conversation.
    In a quantum protocol, the view is the output on registers $\RegM \otimes \RegV$. Similarly, we can define the view of $P$ to be $\view_P\left(\langle P(y_1),V(y_2) \rangle\right)$ that includes the output on the registers $\RegP \otimes \RegM$. 
\end{itemize}

\subsection{Statistically Binding and Quantum-Concealing Commitments}
\label{sec:prelims:commit}
\noindent We employ a two-message  commitment scheme that satisfies the following two properties. 
%commitment scheme consists a classical PPT algorithm $\comm$ that takes as input security parameter $1^{\secparam}$, input message $ x$ and outputs the commitment $\bfc$. There are two properties that need to be satisfied by a commitment scheme: binding and hiding. In this work, we are interested in commitment schemes that are statistically binding and quantum-concealing (quantum analogue of computationally hiding); we define both these notions below. We adapt the definition of computational hiding to the quantum setting.  
 
%\newcommand{\msg}{\mathsf{msg}} 
%\newcommand{\receiver}{\mathsf{R}}
\begin{definition}[Statistically Binding]
A two-message commitment scheme between a committer ($\comm$) and a receiver ($\receiver$), both running in probabilistic polynomial time, is said to satisfy statistical binding property if the following holds for any adversary $\adversary$: 
$$\prob \left[ \substack{(\bfc,r_1,x_1,r_2,x_2)  \leftarrow \adversary\\ \bigwedge\\  \comm(1^{\secparam},\bfr,x_1;r_1) = \comm(1^{\secparam},\bfr,x_2;r_2) = \bfc\\ \bigwedge\\ x_1 \neq x_2}: \bfr \leftarrow \receiver(1^{\secparam}) \right] \leq \negl(\secparam),$$
for some negligible function $\negl$. 
\end{definition}

\begin{definition}[Quantum-Concealing]
A commitment scheme $\comm$ is said to be quantum concealing if the following holds. Suppose $\adversary$ be a non-uniform QPT algorithm and let $\bfr$ be the message generated by $\adversary(1^{\secparam})$. We require that $\adversary$ cannot distinguish the two distributions,  $\{\comm(1^{\secparam},\bfr,x_1)\}$ and $\{\comm(1^{\secparam},\bfr,x_2)\}$, for any two inputs $x_1,x_2$. 
\end{definition}

\begin{remark}
We only considered two message protocols in the above definition for simplicity. 
\end{remark}

\paragraph{Instantiation.} We can instantiate statistically binding and quantum-concealing commitments from post-quantum one-way functions~\cite{Naor91}. 

\begin{comment}
\paragraph{Instantiation.} A construction of perfectly binding non-interactive commitments was presented in the works of~\cite{GHKW17,LS19} assuming the hardness of learning with errors. Thus, we have the following: 

\begin{lemma}[\cite{GHKW17,LS19}]
Assuming the quantum hardness of learning with errors, there exists a construction of perfectly binding quantum-computational hiding non-interactive commitment schemes. 
\end{lemma}
\end{comment}

%\newcommand{\sender}{{\cal S}}
%\newcommand{\state}{\mathsf{st}}

\subsection{Watrous Rewinding Lemma}
\noindent We first state the following lemma due to Watrous~\cite{Wat09}.

\begin{lemma}[Watrous Rewinding Lemma] 
\label{lem:watrous}
Suppose $Q$ be a quantum circuit acting on $n+k$ qubits such that for every $n$-qubit state $\ket{\psi}$, the following holds: 
$$Q\ket{\psi}\ket{0^{\otimes k}} = \sqrt{p(\psi)}\ \ket{0} \ket{\phi_0(\psi)} + \sqrt{1-p(\psi)}\ \ket{1} \ket{\phi_1(\psi)}  $$
Let $p_0,p_1 \in (0,1)$ and $\varepsilon \in (0,1/2)$ be real numbers such that:
\begin{itemize}
    \item $|p(\psi)-p_1| \leq \varepsilon$
    \item $p_0(1-p_0) \leq p_1(1-p_1)$, and 
    \item $p_0 \leq p(\psi)$
\end{itemize}
for all $n$-qubit states. Then there exists a general quantum circuit $R$ of size $O \left( \frac{\log(1/\varepsilon) \size(Q) }{p_0(1-p_0)} \right)$ satisfying the following property: 
$$\bra{\phi_0(\psi)} \rho(\psi) \ket{\phi_0(\psi)} \geq 1 - 16 \varepsilon \frac{\log^2(1/\varepsilon)}{p_0^2(1-p_0)^2} $$
\noindent In this case, we define $R$ to be $\amplifier(Q,\varepsilon)$. If $\varepsilon$ is a negligible function in the security parameter, we omit this from the algorithm. 
\end{lemma}

\begin{comment}
\subsection{Goldreich-Levin Theorem}

\noindent We present Goldreich-Levin theorem~\cite{GL89} below. We state the version presented in~\cite{DGHMW20}.
\begin{theorem}[Goldreich-Levin Theorem]
\label{thm:gl}
There exists a PPT algorithm $\gldec$ such that for any $n,\nu$, any $y \in \{0,1\}^{n}$, and any function $\predictor:\{0,1\}^n \rightarrow \{0,1\}$ satisfying the following: 
$$\prob_{u \xleftarrow{\$} \{0,1\}^n }[\langle u,y \rangle \leftarrow \predictor \left( u \right)] \geq \frac{1}{2} + \frac{1}{\nu},$$
we have: 
$$\prob \left[ \gldec^{\predictor} \left( 1^n,1^{\nu} \right) = y \right] \geq \frac{1}{\poly(n,\nu)} $$
\end{theorem}
\end{comment}

%% file: Definitions/def.tex
\section{Concurrent Quantum ZK Proof Systems: Definitions} 
\label{sec:cqzk}
\submversion{
\noindent We denote the security parameter by $\secparam$. 
\par We denote the (classical) computational indistiguishability of the two distributions $\distr_0$ and $\distr_1$ by $\distr_0 \approx_{c,\varepsilon} \distr_1$, where $\varepsilon$ is the distinguishing advantage. In the case when $\varepsilon$ is negligible, we drop $\varepsilon$ from this notation.
\par We define two distributions $\distr_0$ and $\distr_1$ to be quantum computationally indistinguishable if they cannot be distinguished by QPT distinguishers; we define this formally in the full version. We denote this by $\distr_0 \approx_{\quant,\varepsilon} \distr_1$, where $\varepsilon$ is the distinguishing advantage. We denote the process of an algorithm $A$ being executed on input a sample from a distribution $\distr$ by the notation $A(\distr)$.

\paragraph{Languages and Relations.} A language $\lang$ is a subset of $\{0,1\}^*$. A (classical) relation $\rel$ is a subset of $\{0,1\}^* \times \{0,1\}^*$. We use the following notation:
\begin{itemize}

\item Suppose $\rel$ is a relation. We define $\rel$ to be {\em efficiently decidable} if there exists an algorithm $A$ and fixed polynomial $p$ such that $(x,w) \in \rel$ if and only if $A(x,w)=1$ and the running time of $A$ is upper bounded by $p(|x|,|w|)$. 

\item Suppose $\rel$ is an efficiently decidable relation. We say that $\rel$ is a NP relation if $\lang(\rel)$ is a NP language, where $\lang(\rel)$ is defined as follows: $x \in \lang(R)$ if and only if there exists $w$ such that $(x,w) \in \rel$ and $|w| \leq p(|x|)$ for some fixed polynomial $p$. 

\end{itemize}
%The rest of the relevant preliminaries are presented in Section~\ref{sec:prelims}. \\
%\ \\
}
\noindent In Section~\ref{sec:bczknp}, we define the notion of bounded concurrent QZK for NP. In Section~\ref{sec:bczk:qma}, we define the notion of bounded concurrent ZK for QMA. We present the definition of quantum proof of knowledge in~\Cref{sec:qpok:def}. 

\subsection{Bounded Concurrent QZK for NP}
\label{sec:bczknp}
We start by recalling the definitions of the completeness and soundness properties of a classical interactive proof system. 

\begin{definition}[Proof System]
Let $\Pi$ be an interactive protocol between a classical PPT prover $P$ and a classical PPT verifier $V$. Let $\rel(\lang)$ be the NP relation associated with $\Pi$.
\par $\Pi$ is said to satisfy {\bf completeness} if the following holds: 
\begin{itemize} 
\item {\bf Completeness}: For every $(x,w) \in \rel(\lang)$, 
$$\prob[\accept \leftarrow \langle P(x,w),V(x) \rangle] \geq 1 - \negl(\secparam),$$
for some negligible function $\negl$. 
\end{itemize}
\par $\Pi$ is said to satisfy {\bf (unconditional) soundness} if the following holds: 
\begin{itemize}
    \item {\bf Soundness}: For every  prover $P^*$ (possibly computationally unbounded), every $x \notin \rel(\lang)$, 
    $$\prob\left[ \accept \leftarrow \langle P^*(x),V(x) \rangle  \right] \leq \negl(\secparam),$$
   for some negligible function $\negl$. 
\end{itemize}
\end{definition}

\fullversion{\begin{remark}
In Section~\ref{sec:pok}, we define a stronger property called proof of knowledge property that subsumes the soundness property. 
\end{remark}}

\submversion{\begin{remark}
We will later define a stronger property called proof of knowledge property that subsumes the soundness property. 
\end{remark}}

\noindent To define (bounded) concurrent QZK, we first define $Q$-session adversarial verifiers. Roughly speaking, a $Q$-session adversarial verifier is one that invokes $Q$ instantiations of the protocol and in each instantiation, the adversarial verifier interacts with the honest prover. In particular, the adversarial verifier can interleave its messages from different instantiations.  

\begin{definition}[$Q$-session Quantum Adversary]
\label{def:classical:qadv}
Let $Q \in \mathbb{N}$. Let $\Pi$ be an interactive protocol between a (classical) PPT prover and a (classical) PPT verifier $V$ for the relation $\rel(\lang)$. Let $(x,w) \in \rel(\lang)$. We say that an adversarial non-uniform QPT verifier $V^*$ is a {\bf $Q$-session adversary} if it invokes $Q$ sessions with the prover $P(x,w)$. 
\par Moreover, we assume that the interaction of $V^*$ with $P$ is defined as follows: denote by $V^*_i$ to be the verifier algorithm used by $V^*$ in the $i^{th}$ session and denote by $P_i$ to be the $i^{th}$ invocation of $P(x,w)$ interacting with $V^*_i$. Every message sent by $V^*$ is of the form $\left( \left(1,\msg_1 \right),\ldots,\left(Q,\msg_Q \right) \right)$, where $\msg_i$ is defined as: 
$$\msg_i = \left\{ \begin{array}{ll} \na, & \text{if }V^*_i\text{ doesn't send a message},\\ (\rnd,z), & \text{if }V^*_{i}\text{ sends }z\text{ in the round }\rnd  \end{array}\right. $$
$P_i$ responds to $\msg_i$. If $\msg_i=\na$ then it sets $\msg'_i=\na$. If $V_i^*$ has sent the messages in the correct order\footnote{That is, it has sent  $(1,z_1)$ first, then $(2,z_2)$ and so on.}, then $P_i$ applies the next message function on its own private state and $\msg_i$ to obtain $z'$ and sets $\msg'_i=(t+1,z')$. Otherwise, it sets  $\msg'_i=(\bot,\bot)$. Finally, $V^*$ receives $\left( (1,\msg'_1),\ldots,(Q,\msg'_Q) \right)$. In total, $V^*$ exchanges $\ell_{\smprot} \cdot Q$ number of messages, $\ell_{\smprot}$ is the number of the messages in the protocol.
\end{definition}

\noindent While the above formulation of the adversary is not typically how concurrent adversaries are defined in the concurrency literature, we note that this formulation is without loss of generality and does capture all concurrent adversaries. 
\par We define quantum ZK for NP in the concurrent setting below. 

\begin{definition}[Concurrent Quantum ZK for NP]
An interactive protocol $\Pi$ between a (classical) PPT prover $P$ and a (classical) PPT verifier $V$ for a language $\lang \in \np$ is said to be a {\bf concurrent quantum zero-knowledge (QZK) proof system} if it satisfies completeness, unconditional soundness and the following property: 
\begin{itemize} 

\item {\em Concurrent Quantum Zero-Knowledge}:  For every sufficiently large $\secparam \in \mathbb{N}$, every polynomial $Q=Q(\secparam)$, every $Q$-session QPT adversary $V^*$ there exists a QPT simulator $\simr$ such that for every $(x,w) \in \rel(\lang)$, $\poly(\secparam)$-qubit bipartite advice state, $\rho_{AB}$, on registers $A$ and $B$, the following holds: 
$$\view_{V^*} \left\langle P(x,w),V^*(x,\rho_{AB}) \right\rangle \approx_{\quant} \simr(x,\rho_{AB})$$
where $V^*$ and $\simr$ only have access to register $A$. In other words, only the identity is performed on register $B$.
\end{itemize}
\end{definition}

\noindent In this work, we consider a weaker setting, called bounded concurrency. The number of sessions, denoted by $Q$, in which the adversarial verifier interacts with the prover is fixed ahead of time and in particular, the different complexity measures of a protocol can depend on $Q$. 

\begin{definition}[Bounded Concurrent Quantum ZK for NP]
Let $Q \in \mathbb{N}$. An interactive protocol between a (classical) probabilistic polynomial time (in $Q$) prover $P$  and a (classical) probabilistic polynomial time (in $Q$) verifier $V$ for a language $\lang \in \np$ is said to be a {\bf bounded concurrent quantum zero-knowledge (QZK) proof system} if it satisfies completeness, unconditional soundness and the following property: 
\begin{itemize} 

\item {\em Bounded Concurrent Quantum Zero-Knowledge}: For every sufficiently large $\secparam \in \mathbb{N}$, every $Q$-session concurrent QPT adversary $V^*$, there exists a QPT simulator $\simr$ such that for every $(x,w) \in \rel(\lang)$, $\poly(\secparam)$-qubit bipartite advice state, $\rho_{AB}$, on registers $A$ and $B$, the following holds: 
$$\view_{V^*} \left\langle P(x,w),V^*(x,\rho_{AB}) \right\rangle \approx_{\quant} \simr(x,\rho_{AB})$$
where $V^*$ and $\simr$ only have access to register $A$. In other words, only the identity is performed on register $B$.

\end{itemize}
\end{definition}

\subsection{Bounded Concurrent QZK for QMA}
\label{sec:bczk:qma}
We start by recalling the definitions of completeness and soundness properties of a quantum interactive proof system for promise problems. 

\begin{definition}[Interactive Quantum Proof System for QMA]
$\Pi$ is an interactive proof system between a QPT prover $P$ and a QPT verifier $V$, associated with a promise problem $A=A_{\text{yes}}\cup A_{\text{no}} \in \qma$, if the following two conditions are satisfied.
\begin{itemize} 
\item {\bf Completeness}: For all $x \in A_{\text{yes}}$, there exists a $\poly(|x|)$-qubit state $\ket{\psi}$ such that the following holds: 
$$\prob[\accept \leftarrow \langle P(x,\ket{\Psi}),V(x) \rangle] \geq 1 - \negl(|x|),$$
for some negligible function $\negl$. 
\end{itemize}
\par $\Pi$ is said to satisfy {\bf (unconditional) soundness} if the following holds: 
\begin{itemize}
    \item {\bf Soundness}: For every  prover $P^*$ (possibly computationally unbounded), every $x \in A_{\text{no}}$, the following holds: 
    $$\prob\left[ \accept \leftarrow \langle P^*(x),V(x) \rangle  \right] \leq \negl(|x|),$$
   for some negligible function $\negl$. 
\end{itemize}

\end{definition}

\noindent To define bounded concurrent QZK for QMA, we first define the notion of $Q$-sesssion adversaries. 

\begin{definition}[Q-session adversary for QMA] Let $Q \in \mathbb{N}_{\geq 1}$. Let $\Pi$ be a quantum interactive protocol between a QPT prover and a QPT verifier $V$ for a $\qma$ promise problem $A=A_{\text{yes}}\cup A_{\text{no}}$. We say that an adversarial non-uniform QPT verifier $V^*$ is a Q-session adversary if it invokes $Q$ sessions with the prover $P(x,\ket{\psi})$.
\par As in the case of concurrent verifiers for NP, we assume that the interaction of $V^*$ with $P$ is defined as follows: denote by $V^*_i$ to be the verifier algorithm used by $V^*$ in the $i^{th}$ session and denote by $P_i$ to be the $i^{th}$ invocation of $P(x,w)$ interacting with $V^*_i$. Every message sent by $V^*$ is of the form $\left( \left(1,\msg_1 \right),\ldots,\left(Q,\msg_Q \right) \right)$, where $\msg_i$ is defined as: 
$$\msg_i = \left\{ \begin{array}{ll} \na, & \text{if }V^*_i\text{ doesn't send a message},\\ (\rnd,\rho), & \text{if }V^*_{i}\text{ sends the state }\rho\text{ in the round }\rnd  \end{array}\right. $$
$P_i$ responds to $\msg_i$. If $\msg_i=\na$ then it sets $\msg'_i=\na$. If $V_i^*$ has sent the messages in the correct order, $P_i$ applies the next message function (modeled as a quantum circuit) on $\msg_i$ and its private quantum state to obtain $\rho'$ and sets $\msg'_i=(t+1,\rho')$. Otherwise, it sets $\msg'_i=(\bot,\bot)$. Finally, $V^*$ receives $\left( (1,\msg'_1),\ldots,(Q,\msg'_Q) \right)$. In total, $V^*$ exchanges $\ell_{\smprot} \cdot Q$ number of messages, where $\ell_{\smprot}$ is the number of the messages in the protocol.
\end{definition}

\begin{remark}
To invoke $Q$ different  sessions, we assume that the prover has $Q$ copies of the witness state. 
\end{remark}

\begin{remark}
We assume, without loss of generality, the prover will measure the appropriate registers to figure out the round number for each verifier. This is because the malicious verifier can always send the superposition of the ordering of messages.
\end{remark}

\noindent We define quantum ZK for QMA in the bounded concurrent setting below. 

\begin{definition}[Bounded Concurrent QZK for QMA]
Let $Q \in \mathbb{N}$. An interactive protocol $\Pi$ between a QPT prover $P$ (running in time polynomial in $Q$) and a QPT verifier $V$ (running in time polynomial in $Q$) for a $\qma$ promise problem $\cA=\cA_{\text{yes}}\cup \cA_{\text{no}}$ if it satisfies completeness, unconditional soundness and the following property:
\begin{itemize}
    \item \textbf{{\em Bounded Concurrent Quantum Zero-Knowledge:}}  For every sufficiently large $\secparam \in \mathbb{N}$, for every $Q$-session QPT adversary $V^*$, there exists a QPT simulator $\simr$ such that for every $x \in \cA_{\text{yes}}$ and any witness $\ket{\psi}$, $\poly(\secparam)$-qubit bipartite advice state, $\rho_{AB}$, on registers $A$ and $B$, the following holds: 
$$\view_{V^*} \left\langle P(x,\ket{\psi}),V^*(x,\rho_{AB}) \right\rangle \approx_{\quant} \simr(x,\rho_{AB})$$
where $V^*$ and $\simr$ only have access to register $A$. In other words, only the identity is performed on register $B$.
\end{itemize}
\end{definition}

\subsection{Quantum Proofs of  Knowledge}
\label{sec:qpok:def}
We present the definition of quantum proof of knowledge; this is the traditional notion of proof of knowledge, except that the unbounded prover could be a quantum algorithm and specifically, its intermediate states could be quantum states.

\begin{definition}[Quantum Proof of Knowledge]
\label{def:pok}
We say that an interactive proof system $(P,V)$ for a NP relation $\rel$ satisfies $(\varepsilon,\delta)$-proof of knowledge property if the following holds: 
suppose there exists a malicious (possibly computationally unbounded prover) $P^*$ such that for every $x$, and quantum state $\rho$ it holds that:
$$\prob \left[ \left(\widetilde{\rho},\decision \right) \leftarrow \langle P^*(x,\rho),V(x) \rangle \bigwedge \decision = \mathsf{accept} \right] = \varepsilon $$
Then there exists a quantum polynomial-time extractor $\extractor$, such that: 
$$\prob \left[ \left(\widetilde{\rho}', \decision, w \right) \leftarrow \extractor \left( x,\rho \right) \bigwedge \decision = \mathsf{accept} \right] = \delta $$
Moreover, we require $T(\widetilde{\rho},\widetilde{\rho}')=\negl(|x|)$, where $T(\cdot,\cdot)$ denotes the trace distance and $\negl$ is a negligible function. 
\par We drop $(\varepsilon,\delta)$ from the notation if $|\delta-\varepsilon| \leq \negl(|x|)$, for a negligible function $\negl$. 
\end{definition}

\begin{remark}[Comparison with Unruh's Proof of Knowledge~\cite{Unruh12}]
Our definition is a special case of Unruh's quantum proof of knowledge definition. Any proof system satisfying our definition is a quantum proof of knowledge system (according to Unruh's definition) with knowledge error $\kappa$, for any $\kappa$. Moreover, in Unruh's definition, the extraction probability is allowed to be polynomially related to the acceptance probability whereas in our case, the extraction probability needs to be negligibly close to the acceptance probability.
\end{remark}

\begin{definition}[Concurrent Quantum ZK PoK]
We say that a concurrent (resp., bounded) quantum ZK is a concurrent (resp., bounded) QZKPoK if it satisfies proof of knowledge property. 
\end{definition}

%\prab{add the definitions here}
%~\cite{coladangelo2020non,broadbent2019zero}

\subsection{Intermediate Tool: Quantum  Witness-Indistinguishable Proofs for NP} For our construction, we use a proof system that satisfies a property called quantum witness indistinguishability. We recall this notion below. 
\begin{definition}[Quantum Witness-Indistinguishability] 
\label{def:qwi}
An interactive protocol between a (classical) PPT prover $P$  and a (classical) PPT verifier $V$ for a language $L \in \np$ is said to be a \textbf{quantum witness-indistinguishable proof system} if in addition to completeness, unconditional soundness, the following holds: 
\begin{itemize}
    %\item $(P,V)$ is a public-coin protocol. 
    \item {\bf Quantum Witness-Indistinguishability}: For every $x \in \lang$ and $w_1,w_2$ such that $(x,w_1) \in \rel(\lang)$ and $(x,w_2) \in \rel(\lang)$, for every QPT verifier $V^*$ with $\poly(\secparam)$-qubit advice $\rho$, the following holds: 
    $$\left\{ \view_{\vrfr^*}\left(\langle P(x,w_1),V^*(x,\rho) \right) \right\} \approx_{\quant} \left\{ \view_{\vrfr^*}\left(\langle P(x,w_2),V^*(x,\rho) \right) \right\}$$
\end{itemize}
\end{definition}

\paragraph{Instantiation.} By suitably instantiating the constant round WI argument system of Blum~\cite{Blu86} with statistically binding commitments (which in turn can be based on post-quantum one-way functions~\cite{Naor91}), we achieve a 4 round quantum WI proof system for NP. Moreover, this proof system is a public-coin proof system; that is, the verifier's messages are sampled uniformly at random.

%% file: Concurrent_QZK/construction.tex
\fullversion{
\section{Bounded Concurrent QZK for NP}\label{sec:BCQZK}
%\section{Bounded Concurrent QZK for NP}\label{sec:BCQZK}

}
\submversion{
\subsection{Construction}\label{sec:BCQZK}
}

\noindent We present the construction of quantum zero-knowledge proof system for NP in the bounded concurrent setting in Figure~\ref{fig:pcoinczk}. As remarked earlier, the construction is the same as the classical bounded concurrent ZK by Pass et al.~\cite{PTW09}, whereas our proof strategy is significantly different from that of Pass et al. 
\par The relation associated with the bounded concurrent system will be denoted by $\rel(\lang)$, with $\lang$ being the associated NP language. Let $Q$ be an upper bound on the number of sessions. We use the following tools in our construction. 
\begin{itemize}
    \item Statistically-binding and quantum-concealing commitment protocol\fullversion{ (see Section~\ref{sec:prelims:commit})}, denoted by $(\comm,\receiver)$.
    \item Four round quantum witness-indistinguishable proof system $\protwi$ (Definition~\ref{def:qwi}). The relation associated with $\protwi$, denoted by $\relwi$, is defined as follows: 
    $$\relwi = \Bigg\{ \left( \left(x,\bfr_1,\bfc_1,b'_1,\ldots,\bfr_{120Q^{7}\secparam},\bfc_{120Q^{7}\secparam},b'_{120Q^7\secparam}\right)\ ;\ \left(w,r_1,\ldots,r_{120Q^7\secparam} \right) \right)\ :\  (x,w) \in \rel(\lang) \bigvee $$
    $$\ \ \ \ \ \ \ \ \ \ \ \ \ \ \ \ \ \ \ \ \ \ \ \ \  \left( \exists j_1,\ldots,j_{60Q^7\secparam + Q^4\secparam} \in [120Q^7 \secparam] \text{  s.t. } \bigwedge_{i=1}^{60Q^7\secparam + Q^4 \secparam}  \comm(1^{\secparam},\bfr_{j_i},b'_{j_i};r_{j_i})=\bfc_{j_i}  \right) \Bigg\}$$
\end{itemize}

%\subsection{Construction} 
%\noindent We describe the construction in Figure~\ref{fig:pcoinczk}.
\begin{figure}[!htb]
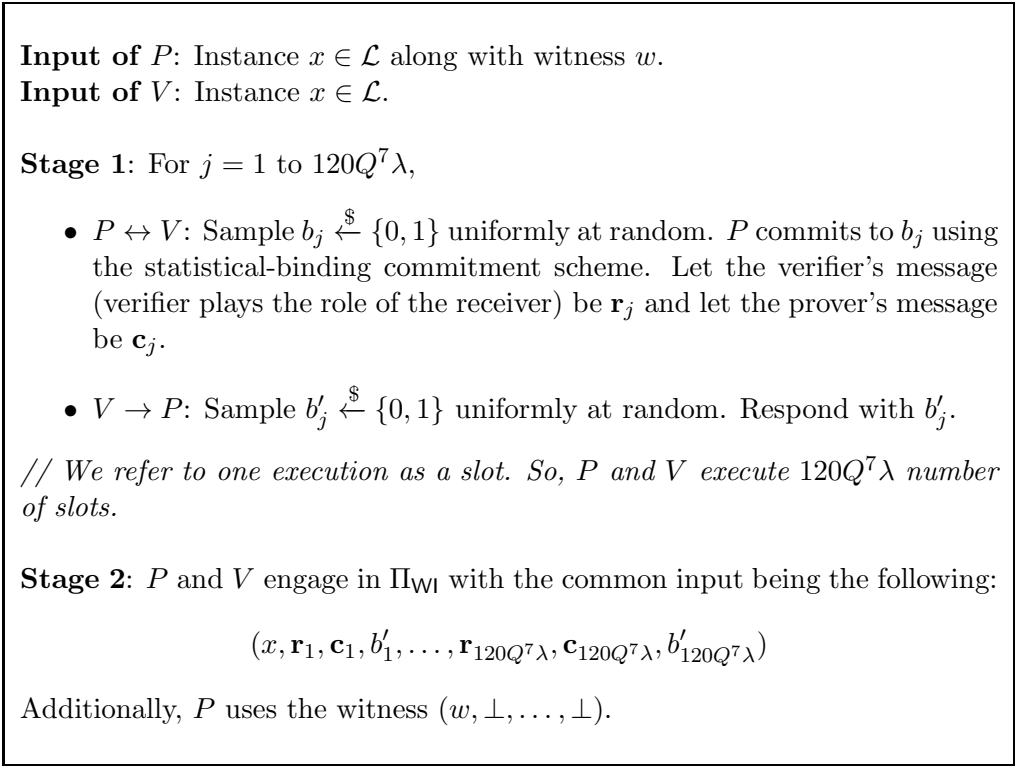

   \begin{center} 
   \begin{tabular}{|p{13cm}|}
    \hline \\
    {\bf Input of $P$}: Instance $x \in \lang$ along with witness $w$.  \\
    {\bf Input of $V$}: Instance $x \in \lang$. \\
    \ \\
    \noindent {\bf Stage 1}: For $j=1$ to $120Q^{7}\secparam$, 
    \begin{itemize}
        \item $P \leftrightarrow V$: Sample $b_j \xleftarrow{\$} \{0,1\}$ uniformly at random. $P$ commits to $b_j$ using the statistical-binding commitment scheme. Let the verifier's message (verifier plays the role of the receiver) be $\bfr_j$ and let the prover's message be  $\bfc_j$.   
        \item $V \rightarrow P$:  Sample $b'_j \xleftarrow{\$} \{0,1\}$ uniformly at random. Respond with $b'_j$. 
    \end{itemize}
    {\em // We refer to one execution as a slot. So, $P$ and $V$ execute $120Q^7\secparam$ number of slots.}\\
    \ \\
    \noindent {\bf Stage 2}: $P$ and $V$ engage in $\protwi$ with the common input being the following: $$(x,\bfr_1,\bfc_1,b'_1,\ldots,\bfr_{120Q^7\secparam},\bfc_{120Q^7\secparam},b'_{120Q^7 \secparam})$$ Additionally, $P$ uses the witness  $(w,\bot,\ldots,\bot)$. \\
    \ \\
    \hline
   \end{tabular}
    \caption{Construction of classical bounded concurrent ZK for NP.}
    \label{fig:pcoinczk}
    \end{center}
\end{figure}
%%%%%%%%%%%%%%%%%%%%%%

\ \\

\fullversion{
\noindent Observe that our construction is also a public-coin system. This follows from the fact that the instantiation of the four-round witness-indistinguishable proof system is a public-coin system. We are now ready to prove the following theorem. }

\fullversion{
\input{Concurrent_QZK/bcqzkproof}

}

%% file: Concurrent_QZK/bcqzkproof.tex
\begin{theorem}
Assuming the security of $(\comm,\receiver)$ and $\protwi$, the construction in Figure~\ref{fig:pcoinczk} is a bounded concurrent QZK proof system. 
\end{theorem}
\begin{proof}
We prove the completeness, soundness and the quantum zero-knowledge properties. 

\paragraph{Completeness.} This follows from the completeness of $\protwi$. \\

\noindent Before we prove soundness and quantum zero-knowledge, we first give the following useful definition.

\begin{definition}[Matched Slot]
We say that a slot is matched if the bit committed by $P$ equals $V$'s response. 
\end{definition}

\paragraph{Soundness.} To argue soundness, we need to argue that with probability negligibly close to 1, the number of matched slots in a transcript, associated with an instance not in the language, is less than $60Q^7 \secparam + Q^4 \secparam$. 

Let $P^*$ be the malicious prover and let $x \notin \lang$.  Denote by  $\bfc_1,\ldots,\bfc_{120Q^7\secparam}$, the commitments produced by $P^*$ in Stage 1. 
\par We first observe that  $(x,\bfr_1,\bfc_1,b'_1,\ldots,\bfr_{120Q^7\secparam},\bfc_{120Q^7\secparam},b'_{120Q^7 \secparam}) \notin \relwi$ with probability negligibly close to 1. By the statistical binding property of the underlying commitment scheme, we have that for every $j \in [60Q^7\secparam + Q^4\secparam]$, there exists a $b_j$ such that $\bfc_{j}$ (prover's message in the $j^{th}$ slot) is a commitment of $b_j$ with respect to some randomness. Let $\rv_j$ be a random variable such that $\rv_j=1$ if $b_j=b'_j$, where $b'_j$ is the bit sent by $V$. The following holds (over the randomness of the verifier):
\begin{eqnarray*}
& & \prob_{}\left[ \exists j_1,\ldots,j_{60Q^7\secparam + Q^4\secparam} \in [120Q^7 \secparam] \text{ s.t. } \bigwedge_{i=1}^{60Q^7\secparam + Q^4 \secparam}  \left( \comm(1
^{\secparam},\bfr_{j_i},b_{j_i};r_{j_i})=\bfc_{j_i} \bigwedge b_{j_i}=b'_{j_i} \right) \right] \\
& = & \prob\left[ \sum_{j=1}^{120Q^7\secparam} \rv_j \geq 60Q^7\secparam + Q^4\secparam \right] \\
& \leq & e^{-\frac{(Q^4\secparam)^{2}}{3(60Q^7 \secparam)}}\ \text{(By Chernoff Bound)} \\
& = & e^{-\frac{Q\secparam}{180} } \\
& = & \negl(\secparam)
\end{eqnarray*}
\noindent The above observation, combined with the fact that $x \notin \lang$, proves the following holds: $$(x,\bfr_1,\bfc_1,b'_1,\ldots,\bfr_{120Q^7\secparam},\bfc_{120Q^7\secparam},b'_{120Q^7 \secparam}) \notin \relwi$$ 
with probability negligibly close to 1.

%%%%%%%%%%%%

\subsection{Quantum Zero-Knowledge} 
Let the malicious QPT verifier be $V^*$. We start by describing some notation. \\

\noindent {\em Parameters}. 
\begin{itemize}
    \item $\ell_{\smprot}$ denotes the number of messages in any given protocol. 
    
    \item We divide the messages exchanged by the simulator with all the sessions into blocks. Let $L$ denote the number of blocks. We set $L=24 Q^6 \secparam$.
    
    \item $\ell_{\smslot}$ denotes the number of slots in Stage 1 of the protocol. That is, $\ell_{\smslot}=120Q^7 \secparam$. Note that every slot contains three messages. We have  $\ell_{\smprot}=3\ell_{\smslot}+\smsecmsg$. 
    \item $\ell_B$ denotes the number of messages contained inside one block. Note that $\ell_B = \frac{\ell_{\smprot} \cdot Q}{L}$. 
    \item $B_i$ denote the $i^{th}$ block.
    \item $\bfN_i$ to be number of blocks containing at least one slot of the $i^{th}$ verifier. 
\end{itemize}

\noindent {\em Registers used by the simulator}: The quantum simulator uses the following registers: 
\begin{itemize}
    \item $\randreg_t$, for $t \in [\ell_{\smprot} \cdot Q]$: it contains the input and randomness used by the simulator to compute the $t^{th}$ message in the transcript; a transcript consists of all the messages in the $Q$ sessions. 
    
    \item $\simreg_t$, for $t \in [\ell_{\smprot} \cdot Q]$: it contains the $t^{th}$ message if it is sent by the simulator. 
    
    \item $\verreg_t$, for $t \in [\ell_{\smprot} \cdot Q]$: it contains the $t^{th}$ message if it is sent by the malicious verifier $V^*$. 
    
    \item $\slotreg_i$, for $i \in [L]$: it contains the matched slots of the $i^{th}$ block. %\rolo{Should we way: "the" match slot or "a" match slot?}

    \item $\bitreg_i$, for $i \in [Q]$: this is a single-qubit register that  contains a bit that indicates whether the simulator needs to use the witness or the matched slots to compute the $i^{th}$ WI proof (where the ordering is determined based on the point of arrival of WI messages).  
    
    \item $\witreg$: it contains the NP witness. 
    
   % \item $\sfC_{i,j}$: it contains Stage 1 commitments corresponding to the $j^{th}$ slot to be sent to the $i^{th}$ verifier. 
   % \item $\sfb_{i,j}$: it contains Stage 1 messages of the $i^{th}$ verifier corresponding to the $j^{th}$ slot. 
   % \item $\sfprovsecstg_{i}$: it contains Stage 2 messages of the prover to be sent to the $i^{th}$ verifier.   
    %\item $\sfversecstg_{it}$: it contains Stage 2 messages of the $i^{th}$ verifier. 
    %\item $\abortreg_i$: it contains the abort symbol if the $i^{th}$ verifier has aborted. 
    %\item $\sfD$: it contains the decision register that indicates whether to rewind or not.  
    \item $\auxreg$: it contains the private state of the verifier. It is initialized with the auxiliary state of the verifier. 
    
      \item $\decreg$: it contains the decision register that indicates whether to rewind or not. 
      
      \item $\regX$: this is a $\poly(\secparam)$-qubit ancillary register.
      
\end{itemize}

\medskip

\newcommand{\regV}{\mathsf{V}}
\noindent{\underline{Description of $\mathsf{Sim}^{V^*}(1^\secparam, x, \ket{\Psi}):$}}
\begin{enumerate}
    \item For any $w$, let $\ket{\Psi_{0,w}}$ denote the following state:
    $$  \ket{\Psi_{0,w}} = \left( \overset{\ell_{\smprot}\cdot Q}{\underset{t=1}{\bigotimes}} \ket{0}_{\randreg_{t}}\ket{0}_{\simreg_t} \ket{0}_{\verreg_t} \right) \otimes \left( \overset{L}{\underset{j=1}{\bigotimes}} \ket{0}_{\slotreg_j}  \right) \otimes \left( \overset{Q}{\underset{i=1}{\bigotimes}} \ket{0}_{\bitreg_i}  \right) \otimes \ket{w}_{\witreg} \otimes \ket{\Psi}_{\auxreg} \otimes \ket{0}_{\decreg} \otimes \ket{0^{\otimes \poly(\secparam)}}_{\regX} $$
    Initialize the state $\ket{\Psi_{0,\bot}}$. %\prab{not sure I know what $R_i,S_i$ are. The register descriptions before this and here are inconsistent with each other.}
    \item For all $j=\{1,2,\ldots,L\}$, let $U_j^{V^*}$ be the unitary that performs the following operations ((a) and (b)) in superposition.
    
    \begin{enumerate}
     \item For all integers $t \in [(j-1)\ell_B+1, j\ell_B]$: 
    \begin{itemize}
        \item If the $t^{th}$ message is a Stage 1 message from the prover responding to the first session message of a slot,  apply the following operation in superposition over the receiver's message\footnote{We assume without loss of generality that the length of the sender's randomness in the commitment scheme is $\secparam$.}: $$\ket{\bfr}_{\verreg_t}\ket{0}_{\randreg_{t}}\ket{0}_{\simreg_t} \rightarrow \frac{1}{\sqrt{2^{\secparam+1}}}\underset{b \in \{0,1\},r \in \{0,1\}^\secparam}{\sum} \ket{\bfr}_{\verreg_t}\ket{b,r}_{\randreg_t}\ket{\comm(1^\secparam, \bfr,b;r)}_{\simreg_t},$$
        while leaving all the other registers intact. Note that we can prepare this state efficiently by first applying $H^{\otimes (\secparam + 1)}$ to the $\randreg_t$ register followed by applying $\comm$ in superposition and storing the output in the $\simreg_t$ register.
        \item If the $t^{th}$ message is a verifier's message, apply $V^*$ on the registers corresponding to the transcript of the protocol until the $t^{th}$ message (i.e. registers $\{(\simreg_i)\}_{i  \leq t},\{\verreg_i\}_{i < t}$, $\auxreg$) and on $\auxreg$ register that corresponds to the verifier's private state, and output in the register $\verreg_t$. 
        \item If the $t^{th}$ message is a Stage 2 message from the prover responding to the $i^{th}$ WI initiated by the verifier (this just means that so far, $(i-1)$ WIs from $(i-1)$ sessions have already been initiated in the transcript), let $w$ be the string in the register $\witreg$. Let $c_i$ be the bit in register $\bitreg_i$.  If $c_i = 1$, use $w$ as the witness to the WI proof. If $c_i=0$, check if at least $\frac{\ell_{\smslot}}{2}+Q^4 \secparam$ matched slots corresponding to the session whose WI message is being computed. If so, compute the WI of Stage 2 using these matched slots. Otherwise, abort and output $\bot$ on register $\simreg_t$\footnote{It may not be clear why we need this register. However, having this register would help us in the presentation of the hybrids.}.

    \end{itemize}
    \item Let $T$ contain the transcript of messages sent in block $B_j$ along with the input and randomness used by the simulator to create these messages (i.e. the string stored in the registers $\{(\randreg_t,\simreg_t,\verreg_t)\}_{i \in B_j}$), and let $\mu(T)$ denote the set of all slots that are inside $B_j$ in the transcript $T$. In superposition, perform the unitary $U'$ defined below. Let $I$ be a register containing a subset of qubits in $\regX$. We omit the subscripts of the registers associated with the transcript $T$. 
    \begin{eqnarray*}
   & &  U' \ket{T} \ket{0}_{\slotreg_j} \ket{0}_{\decreg} \ket{0^{\otimes |I|}}_{I} \\
    & \approx & \ket{T} \otimes \left(\frac{1}{\sqrt{|\mu(T)|}} \sum_{(\bfc,b') \in \mu(T)}  \ket{\bfc,b'}_{\slotreg_j}\ket{1 \oplus \mathsf{Match}(T,\bfc,b'))}_{\decreg} \ket{\phi_{\bfc,b'}}_I \right) \text{ if } \mu(T) \neq \emptyset \\
    & = & \ket{T} \ket{0}_{\slotreg_j} \ket{+}_{\dec} \ket{0^{\otimes |I|}}_I \text{ if } \mu(T)=\emptyset
    \end{eqnarray*}
    where $\mathsf{Match}(T,\bfc,b') = 1$ if $\bfc$ is a commitment to $b'$ and $0$ otherwise. $\ket{\phi_{\bfc,b'}}$ is some auxiliary state. Note that $T$, in addition to containing the transcript of messages exchanged in $B_j$, also contains the input and the randomness used by the simulator to create these messages. 
    \par By $\approx$, we mean the following: we say $\ket{\phi_0} \approx \ket{\phi_1}$ if both the states $\ket{\phi_0}$ and $\ket{\phi_1}$ are exponentially close (in trace distance) to each other. To see how we can obtain the above state, the unitary $U'$ creates uniform superpositions over $[1],[2],\ldots,[|T|]$. Then, $U'$ determines $\mu(T)$ and uses the uniform superposition over $[|\mu(T)|]$ to create a uniform superposition over $\ket{\bfc,b'}$. 
    
   % \prab{point to the overview and explain why we are doing the above operation}
    \end{enumerate}
Let $W_{j} = \amplifier \left( U_j^{V^*} \right)$; where $\amplifier$ is the circuit guaranteed by Lemma~\ref{lem:watrous}. Simulator computes $\ket{\Psi_{j,\bot}} = W_j \ket{\Psi_{j-1,\bot}}$. 
    
    \item For all $t \in \{1,...,\ell_{\smprot}\cdot Q \}$, measure all the $\simreg_t$ and $\verreg_t$ registers in the computational basis, and output the measurement outcomes along with the resulting state in the $\auxreg$ register. In other words, let $Y$ be the measurement outcome after measuring the registers corresponding to the protocol's transcript. Then, output $Y$ along with
    $$\widetilde{\rho} = \frac{\tr_{\overline{\aux}} \left[ \Pi_{Y} \ket{\Psi_{L,\bot}}\bra{\Psi_{L,\bot}} \Pi_{Y}\right]}{\tr \left[ \Pi_{Y} \ket{\Psi_{L,\bot}}\bra{\Psi_{L,\bot}} \Pi_{Y}\right]} $$
    where $\Pi_Y$ projects the registers $(\simreg_1,\verreg_1,\ldots,\simreg_{\ell_{\smprot}\cdot Q}, \verreg_{\ell_{\smprot}\cdot Q})$ onto $Y$. By $\tr_{\overline{\aux}}[\cdot]$, we mean the operation of tracing out all the registers except $\aux$. 
    \end{enumerate}

\begin{remark}
Using the description of the unitaries $U_i^{V^*}$ as above, note that for any $(x,w) \in \rel(\lang)$, if the prover and the verifier ran their protocol in superposition (and never measured), their combined output would be $U_L^{V^*} \cdots U_1^{V^*} \left( I \otimes X^{\otimes_{j \in [Q]} \bitreg_j} \right) \ket{\Psi_{0,w}}$, where $X^{\otimes_{j \in [Q]} \bitreg_j}$ is Pauli X's applied to the $\{\bitreg_i\}_{i \in [Q]}$ registers and $I$ is the identity operator applied on the rest of the registers. On the other hand, the state obtained by the simulator just before the final partial measurement is $W_L \cdots W_1 \ket{\Psi_{0,\bot}}$.
\end{remark} 

\noindent We will show that for any verifier's auxiliary state $\ket{\Psi}$, the output of this simulator is indistinguishable from the output of the verifier when interacting with the honest prover.

\begin{lemma}
For any $(x,w) \in \rel(\lang)$, and for any auxiliary $\poly(\secparam)$-qubits state\footnote{We can assume without of generality, via the process of purification, that the input state of the verifier is a pure state.} $\ket{\Psi}$, the output of $\Simu^{V^*}(1^\secparam,x, \ket{\Psi})$ is computationally indistinguishable from $\view_{V^*} \left\langle P(x,w),V^*(x,\ket{\Psi}) \right\rangle$.
\end{lemma}

\begin{proof}

We will proceed with a series of hybrids.\\

\noindent \underline{$\hybrid_0$}: The output of this hybrid is the output of the verifier when interacting with the honest prover.\\

\noindent \underline{$\hybrid_1$}: Define a hybrid simulator $\hybrid_{1}.\Simu^{V^*}(x,w,\ket{\Psi})$ that behaves like the honest prover, but performs the execution of the prover and the verifier in all the sessions in superposition. This simulator first prepares the state $U_L^{V^*} \ldots U_1^{V^*} \left( I \otimes X^{\otimes_{j \in [Q]} \bitreg_j} \right) \ket{\Psi_{0,w}}$, then, it measures the registers corresponding to the transcript (that is, $\{(\simreg_t,\verreg_t)\}_{t \in [\ell_{\smprot}]]}$) and outputs the measurement outcome along with the resulting verifier's private state.\\

\noindent The distribution of outputs in $\hybrid_0$ and $\hybrid_1$ are identical, since measurements can be deferred to the end by the {\em principle of deferred measurement}.\\

\noindent \underline{$\hybrid_{2.i}$, for $i=1$ to $L$}: Consider the following sequence of hybrid simulators, $\hybrid_{2.i}.\Simu^{V^*}(x,w)$, that behaves like $\hybrid_1.\Simu^{V^*}(x,w)$, but perform Watrous' rewinding on blocks $B_1,\ldots,B_i$. In other words, instead of performing the unitary $U_i^{V^*}$, it performs $W_i=\amplifier\left( U_i^{V^*} \right)$. This means that $\hybrid_{2.i}.\Simu^{V^*}(x,w,\ket{\Psi})$ computes:
$$U_L^{V^*}\cdots U_{i+1}^{V^*}W_i\cdots W_1 \left( I \otimes X^{\otimes_{j \in [Q]} \bitreg_j} \right) \ket{\Psi_{0,w}}$$
The final partial measurement is performed as in the previous hybrid.\\

\noindent We defer the proof of the following claim to Section~\ref{sec:quanconc}.

\begin{claim}
\label{clm:quanconc}
Assuming that $\comm$ satisfies hiding against quantum polynomial-time adversaries, the output distributions of the verifier in $\hybrid_{2.i}$ is computationally indistinguishable from the output distribution of the verifier in $\hybrid_{2.i+1}$. 
\end{claim}

\noindent \underline{$\hybrid_{3.i}$ for $i \in [Q]$}: Define a hybrid simulator $\hybrid_{3.i}.\Simu^{V^*}$ that behaves like $\hybrid_{2.L}$ except that it does not applies the initial bit flip $X$ on registers $\bitreg_k$ for all $k \leq i$. Formally, hybrid $\hybrid_{3.i}$ computes:
$$W_LW_{L-1}\ldots W_1 \left(I \otimes X^{\otimes_{j > i } \bitreg_j}\right) \ket{\Psi_{0,w}}. $$ This change means that in Stage 2 of the protocol, for the sessions that initiate the first $i$ WI protocols, the hybrid simulator $\hybrid_{3.i}.\Simu$ will use matched slots instead of the actual witness to compute the WI proof. For the rest of the sessions, the hybrid simulator still uses the witness $w$ to produce the WI proof.  \\

\noindent We defer the proof of the following claim to Section~\ref{sec:mainclaim}.

\begin{claim}\label{ind WI}
Assuming the witness-indistinguishability property of $\protwi$, the output distributions of the hybrids $\hybrid_{3.i}$ and $\hybrid_{3.i+1}$ are computationally indistinguishable. 
\end{claim}

\noindent \underline{$\hybrid_{4}$}: The output of this hybrid is the output of the simulator. \\ \ \\
\noindent The output distributions of $\hybrid_{3.L}$ and $\hybrid_4$ are identical. 
\end{proof}
\end{proof}

\subsubsection{Proof of Claim~\ref{clm:quanconc}}
\label{sec:quanconc}
\noindent We prove this in the following steps: 
\begin{enumerate} 

\item First, we reduce proving the indistinguishability of $\hybrid_{2.i}$ and $\hybrid_{2.i-1}$ to proving the following statement: the following two distributions are computationally indistinguishable.
\begin{itemize}
    \item $\distr_1$: Measure the $\{\simreg_t,\verreg_t\}_{t \leq i}$ registers at the end of execution of the block $B_i$ in $\hybrid_{2.i-1}$ and output the measurement outcome along with the residual state in the register $\auxreg$. 
    \item $\distr_2$: Measure the $\{\simreg_t,\verreg_t\}_{t \leq i}$ registers at the end of execution of the block $B_i$ in $\hybrid_{2.i}$ and output the measurement outcome along with the residual state in the register $\auxreg$. 
\end{itemize}

\item Next, we show the indistinguishability of $\distr_1$ and $\distr_2$ by using Watrous rewinding and quantum-concealing property of the commitments. 

\end{enumerate} 

\noindent Bullet 1 follows from the fact that the registers $\{\simreg_t,\verreg_t\}_{t \leq i}$ are never written upon after the execution of Block $B_i$ and hence measurement operators applied on these registers in the end commute with the unitaries applied after the execution of $B_i$. 
\par For Bullet 2, we first make some observations on the state obtained in $\hybrid_{2.i}$ after applying Watrous rewinding.

\paragraph{Applying Watrous Rewinding.} Let $\ket{\Psi_{0,w}^{i-1}} = W_{i-1}\ldots W_1 \left( I \otimes X^{\otimes_{j \in [Q]} \bitreg_j} \right) \ket{\Psi_{0,w}}$. Without loss of generality,  we can write $U_i^{V^*} \ket{\Psi_{0,w}^{i-1}}$ the following way: 

\begin{eqnarray*}
U_i^{V^*} \ket{\Psi_{0,w}^{i-1}} & = & \sqrt{q} \ket{\Phi_{i,\noslot}} \ket{+}_{\decreg} + \sqrt{(1-q) } \ket{\Phi_{i,\slot}}\\ 
\end{eqnarray*}
%\begin{eqnarray*}
%U_i^{V^*} \ket{\Psi_{0,w}^{i-1}} & = & \sqrt{q} \ket{\Phi_{i,\noslot}} \ket{+}_{\dec} + \sqrt{(1-q) p } \ket{\Phi_{i,\yes}}\ket{0}_{\dec} + \sqrt{(1-q) p'} \ket{\Phi_{i,\no}}\ket{1}_{\dec}\\ 
%\end{eqnarray*}
\noindent where:
\begin{itemize}
    \item $\ket{\Phi_{i,\noslot}}$ is a  superposition of all the transcripts containing no slot in the $i^{th}$ block $B_i$. This is defined on all the registers except the $\decreg$ register. 
    \item $\ket{\Phi_{i,\slot}}$ is a superposition of all the transcripts containing at least one slot in the $i^{th}$ block $B_i$. This is defined on all the registers. 
\end{itemize}
Furthermore, $\ket{\Phi_{i,\slot}}$ can be written as $\sqrt{p(\Phi_{i,\slot})} \ket{\Phi_{\yes}}\ket{0}_{\decreg} + \sqrt{1-p(\Phi_{i,\slot})} \ket{\Phi_{\no}}\ket{1}_{\decreg}$, for some states $\ket{\Phi_{\yes}}$,  $\ket{\Phi_{\no}}$ and some function $p(\cdot)$. We first claim the following. 

\begin{claim}
\label{clm:quanconc:inter}
Assuming quantum concealing property of $(\comm,\receiver)$, the following holds: $$\left|p(\Phi_{i,\slot}) - \frac{1}{2} \right| \leq \negl(\secparam)$$ 
\end{claim}
\begin{proof} 
By the quantum-concealing property of $\comm$, any QPT adversary $\adversary$, with auxiliary state $\ket{\Phi}$, can win the following game with probability negligibly close to $\frac{1}{2}$: given a commitment $c = \comm(b;r)$, where $b\xleftarrow{\$} \{0,1\}$ and $r\xleftarrow{\$} \{0,1\}^{\secparam}$, we say that $\adversary$ wins if it outputs $b'=b$. 
\par We execute the above experiment in superposition: 
\begin{itemize}
\item $\cA$ sends the first commitment message, $\bfr$.
\item Challenger prepares the following state (omitting the register containing $\bfr$): $$\frac{1}{\sqrt{2^{\secparam + 1}}} \underset{b\in \{0,1\}, r\in \{0,1\}^{\secparam}}{\sum} \ket{b,r}_X\ket{\comm(1^\secparam,b,\bfr;r)}_Y \ket{0}_{Z}\ket{\Phi}_{\auxreg}\ket{0}_{\decreg}$$
\item $\cA$ is computed (over the registers $Y,Z,\auxreg$) in superposition: $$ \frac{1}{\sqrt{2^{\secparam + 1}}} \underset{b\in \{0,1\}, r\in \{0,1\}^{\secparam}}{\sum} \ket{b,r}_X \ket{\comm(1^\secparam, \bfr,b;r)}_Y \ket{\cA(\comm(b;r))}_Z  \ket{\Phi'}_{\auxreg} \ket{0}_\decreg $$ 

\item The challenger computes the following:
$$ \frac{1}{\sqrt{2^{\secparam + 1}}} \underset{b\in \{0,1\}, r\in \{0,1\}^{\secparam}}{\sum} \ket{b,r}_X \ket{\comm(1^\secparam, \bfr,b;r)}_Y \ket{\cA(\comm(1^\secparam,\bfr,b;r))}_Z  \ket{\Phi'}_{\auxreg} \ket{b \oplus \adversary(\comm(b;r))}_\decreg $$
\end{itemize}

We can rewrite the above state as follows: $$ \sqrt{p'} \ket{\phi_0} \ket{0}_\decreg +\sqrt{1- p'} \ket{\phi_1} \ket{1}_\decreg $$
From the above game, it follows that $p'$ is negligibly close to $\frac{1}{2}$. Moreover, if we suitably instantiate $\adversary$ (using the verifier) and $\ket{\Phi}$, it follows that $\ket{\Phi_{\yes}}=\ket{\phi_0}$ and $\ket{\Phi_{\no}}=\ket{\phi_1}$. Thus, we have $p(\Phi_{i,\slot})$ to be negligibly close to $\frac{1}{2}$.

\end{proof}
\noindent Using above, we write $U_{i}^{V^*} \ket{\Psi_{0,w}^{i-1}}$ as follows: 
\begin{eqnarray*}
U_{i}^{V^*} \ket{\Psi_{0,w}^{i-1}} 
& = & \sqrt{p(\Phi_{i,\slot})} \ket{\Psi_{i,\Good}}\ket{0}_{\decreg} + \sqrt{1-p(\Phi_{i,\slot})} \ket{\Psi_{i,\Bad}}\ket{1}_{\decreg},
\end{eqnarray*}
where $\ket{\Psi_{i,\Good}}$ is a superposition over transcripts such that either one of the following two conditions are satisfied: (i) the slot chosen in the $i^{th}$ block is matched or, (ii) the verifier aborts and the simulator decides to not rewind. Similarly, we can define $\ket{\Psi_{i,\Bad}}$. 
\noindent Define $p_1 = \frac{1}{2}$ and $p_0 = 0.49$. We note that the following holds:
\begin{itemize}
    \item $|p(\Phi_{i,\slot}) - p_1| \leq \varepsilon$, where $\varepsilon = \nu(\secparam)$, for some negligible function $\nu(\cdot)$ and,
    \item $p_0(1-p_0) \leq p_1(1-p_1)$ and, 
    \item $p_0 \leq p(\Phi_{i,\slot})$.  
\end{itemize}
Thus, from the Watrous rewinding lemma (Lemma~\ref{lem:watrous}),  $\amplifier\left(U_i^{V^*} \right)$ outputs a circuit $W_i$, of polynomial size, such that $W_i$ on input the state $\ket{\Psi_{0,w}^{i-1}}$, outputs a state $\ket{\Psi_{0,w}^{i}}$ that is exponentially (in $\secparam$) close in trace distance to the state $\ket{\Psi_{i,\Good}}$. This means that, in hybrid $\hybrid_{2.i+1}$, the state obtained after the execution of block $B_i$ is exponentially close in trace distance to the state $\ket{\Psi_{i,\Good}}\ket{0}_{\decreg}$. 

\paragraph{Indistinguishability of $\distr_1$ and $\distr_2$.} We just argued above that the intermediate state  obtained in $\hybrid_{2.i}$ is $\ket{\Psi_{i,\Good}}\ket{0}_{\decreg}$. On the other hand, the intermediate state obtained in $\hybrid_{2.i-1}$ is  $\ket{\Psi_{0,w}^{i-1}}$ is $U_{i}^{V^*} \ket{\Psi_{0,w}^{i-1}} 
 =  \sqrt{p(\Phi_{i,\slot})} \ket{\Psi_{i,\Good}}\ket{0}_{\decreg} + \sqrt{1-p(\Phi_{i,\slot})} \ket{\Psi_{i,\Bad}}\ket{1}_{\decreg}$. We need to argue that the distribution of measurements of the registers $\{\simreg_t,\verreg_t\}$ along with the residual state $\auxreg$ register in both the cases are computationally indistinguishable. 
 
 Note that for any $\rho_0,\rho_1$ such that $\rho_0 \approx_c \rho_1$\footnote{By $\rho_0 \approx_c \rho_1$, we mean that the state sampled according to $\rho_0$ is computationally indistinguishable from the state sampled according to $\rho_1$.}, then for any $p\geq 0$ we have that $\rho_0 = p \cdot \rho_0 + (1-p) \rho \approx_c p \cdot \rho_0 + (1-p) \cdot \rho_1$.  In our case we have, $\rho_0$ is the post-measurement state on the registers $\{\simreg_t,\verreg_t\}_{t \leq i},\auxreg$  after measuring the  $\{\simreg_t,\verreg_t\}_{t \leq i}$ registers of the state $\ket{\Psi_{\Good}}$. Similarly, we define $\rho_1$ with respect to $\ket{\Psi_{\Bad}}$. In $\hybrid_{2.i-1}$, the intermediate state is a mixture of $\rho_0$ and $\rho_1$ and in $\hybrid_{2.i}$, the intermediate state is $\rho_0$. 
\par Thus, it suffices to show that with probability neglibly close to $1$, the post-measurement states $\rho_0$ and $\rho_1$ are computationally indistinguishable. This follows from the quantum-concealing property of commitment schemes and is similar to the proof of Claim~\ref{clm:quanconc:inter}; if the verifier can distinguish a matched slot versus an unmatched slot then this verifier is violating the quantum-concealing property of the commitment scheme. \\

\noindent This proves that hybrids $\hybrid_{2.i}$ and $\hybrid_{2.i+1}$ are computationally indistinguishable. 

\subsubsection{Proof of Claim~\ref{ind WI}}
\label{sec:mainclaim}
Before we prove Claim~\ref{ind WI}, we first give an auxiliary definition and some claims. 

\paragraph{Auxiliary Definition and Claims.}
\begin{definition}[Partitioning]
We define a partitioning of a protocol transcript (consisting of messages from all the sessions) $\sch$ to be $\{B_1,\ldots,B_L\}$ associated with parameter $\ell_B$ as follows: $B_1$ consists of the first $\ell_{B}$ messages of $S$, $B_2$ consists of the second $\ell_{B}$ messages of $S$ and so on. If $|S| - \ell_B \cdot (L-1) < \ell_B$ then the last block $B_L$ will just contain the remaining $|S| - \ell_B \cdot (L-1)$ messages.
\end{definition}

\noindent The following claim lower bounds the number of blocks that will contain a full slot for any given verifier.  In particular, with our chosen parameters, we can show that the number of such blocks is at least $6Q^5 \secparam$. This will turn out to be enough number of blocks for the simulator to be able to obtain more than $60Q^7 \secparam + Q^4 \secparam$ matched commitments, with probability neglibly close to 1, for every verifier before starting Stage 2.
\begin{claim}
\label{clm:noblocks}
For any transcript $\sch$ of $Q$ verifiers $V_1,\ldots,V_Q$ with partitioning $\{B_1,\ldots,B_L\}$, for every verifier $V_i$, we have $\bfN_i \geq 6Q^5 \secparam$; that is, there are at least $6Q^5 \secparam$ number of blocks containing at least one slot of $V_i$. 
\end{claim}
\begin{proof}
Fix a verifier $V_i$. Note that the number of blocks containing at least 4 messages of $V_i$ lower bounds $\bfN_i$. Denote $\mu_i$ be the number of blocks containing at least 4 messages of $V_i$.
\par Let $b_1,\ldots,b_{\mu_i}$ be the number of messages of $V_i$ in each of these $\mu_i$ blocks. Let the number of messages in the remaining $L-\mu_i$ blocks be denoted by $a_1,\ldots,a_{L-\mu_i}$. 
\par The following holds:  $\sum_{i=1}^{\mu_i} b_i + \sum_{i=1}^{L- \mu_i} a_i = \frac{2(\ell_{\smprot}-1)}{3}$. Since $\sum_{i=1}^{\mu_i} b_i \leq \ell_B \mu_i$, $\sum_{i=1}^{L-\mu_i} a_i \leq 3(L-\mu_i)$  and $\ell_B = \frac{\ell_{\smprot} \cdot Q}{L}$, we have:  
$$ \mu_i \ell_B + 3(L - \mu_i) \geq \frac{2 (\ell_{\smprot}-1)}{3} \geq \frac{\ell_\smprot}{2}$$ 
From this, we can determine $\mu_i$ to be at least $\frac{\frac{\ell_{\smprot}}{2} - 3L}{\ell_{B} - 3}$. We can now lower bound the number of blocks containing at least 4 messages as follows.

\begin{align*}
    \bfN_i \geq \mu_i &\geq \left(\frac{ \frac{\ell_{\smprot}}{2} - 3L}{\frac{\ell_{\smprot} Q}{L} - 3} \right) \\
    &\geq \frac{ \frac{\ell_{\smprot}}{2} - 3L}{\ell_{\smprot}}\cdot \frac{L} {Q} \\
    &\geq \left[1-\frac{6L}{\ell_{\smprot}} \right] \frac{L}{2Q} \\
    &\geq \left[ 1-\frac{6L}{3\ell_{\smslot}} \right]\frac{L}{2Q} \\
    &\geq \left[ 1-\frac{2}{5Q} \right]\frac{L}{2Q}\\
    & \geq \left(1 - \frac{1}{2} \right) 12\secparam Q^5\ \ \ \ \ \ (\because L=24Q^6\secparam,\ \ell_{\smslot}=120Q^7\secparam) \\
&\geq 6\secparam Q^5
\end{align*}

 %Since any block that contains at least 3 messages of $V_i$ should contain a slot of $V_i$, we have $N$ also denoting the number of blocks containg at least one slot of $V_i$.  
\end{proof}

\noindent The following claim lower bounds the expected number of slots that will be {\em rigged} by the simulator (i.e., these are the slots the simulator matches by rewinding) for any given verifier before starting Stage 2. Specifically, it bounds the number of slots that it will be able to match thanks to block rewinding. 

\begin{claim}[Matching by Rigging]
\label{clm:int:clm}
Let $\sch$ be a scheduling of $Q$ verifiers $V_1,\ldots,V_Q$. Let $\{B_1,\ldots,B_L\}$ be the partitioning associated with $\sch$. 
\par Consider the following process: for $i=1,\ldots,L$, 
\begin{itemize}
    \item Let $T_i$ be such that all the verifiers $\{V_j\}_{j \in T_i}$ have a slot in $B_i$. 
    \item Pick $j^* \xleftarrow{\$} T$. 
    \item Finally, pick a slot of $V_{j^*}$ in block $B_i$ uniformly at random.
\end{itemize}
Let $\rv_{i,j}$ be a random variable defined to be 1 if in the $j^{th}$ block, a slot of $V_i$ is picked.   Then, for any $i \in [Q]$, $\expct[\sum_{j \in [L]} \rv_{i,j}] \geq 6\secparam Q^4$.   Furthermore,  we have that
$$\prob\left[\exists i \in[Q], \sum_{j \in [L]} \rv_{i,j} \leq 3 \secparam Q^4\right] \leq \negl(\secparam)$$.

\end{claim}
\begin{proof}
Let $b_{i,j}$ be such that $b_{i,j}=1$ if the  $i^{th}$ verifier has a slot in the $j^{th}$ block, else its set to 0. Then, we have $\expct[\sum_{j \in [L]} \rv_{i,j}] \geq \sum_{j \in [L]} b_{i,j} \cdot \frac{1}{Q}$. Note that $\left|\{j: b_{i,j} \neq 0\} \right| = \bfN_i$. Thus, we have $\expct[\sum_{j \in [L]} \rv_{i,j}] \geq \frac{1}{Q} \cdot \bfN_i$. Further applying  Claim~\ref{clm:noblocks}, we have $\expct[\sum_{j \in [L]} \rv_{i,j}] \geq 6\secparam Q^4$.  To finish the proof of the claim, first notice that by Chernoff bound, we have that for any $i \in [Q]$, 
$$\prob\left[\sum_{j \in [L]} \rv_{i,j} \leq 3\secparam Q^4 \right] \leq e^{-\frac{3}{4}Q^4 \secparam}.  $$ By the union bound, we obtain that
$$\prob\left[\exists i \in[Q], \sum_{j \in [L]} \rv_{i,j} \leq 3 \secparam Q^4\right] \leq Qe^{-\frac{3}{4}Q^4 \secparam} $$
\end{proof}

\noindent While the above claim provides a lower bound on the number of rigged slots, the following claim lower bounds the number of slots matched by luck. Combining the above and the below claim, it follows that with overwhelming probability, the number of matchd slots is at least $60Q^7 \secparam + Q^4 \secparam$.  

\begin{claim}[Matching by Luck]
\label{clm:main}
Let $\sch$ be a transcript of the $Q$ verifiers $V_1, \ldots, V_Q$. For any $i \in [Q]$ let $Z_{i,1},...Z_{i,120Q^7 \secparam}$ be binary random variables such that $Z_{i,j} = 1$ iff $\comm(b'_j;r_j) = \bfc_j$ where $b'_j$ is the $j^{th}$ response of the $i^{th}$ verifier to commitment $\bfc_j$ by the prover. Let $X_{i,j}$ be as defined in Claim~\ref{clm:int:clm}. The following holds: 
$$\prob \left[\exists T_i \subseteq [L],\  \left(\forall j \in T_i, X_{i,j}=1\right) \bigwedge \left( \sum_{j \in [L] \setminus T_i} Z_{i,j} \geq 60Q^7 \secparam - 2 Q^4 \secparam \right)  \right] \geq 1 - \nu(\secparam),$$
for some negligible function $\nu(\cdot)$. 
%With probability negligibly close to $1$ the following holds: for all $i \in [Q]$, there exists $T_i \subset [L]$ s.t. $\rv_{i,j} = 1$ for all $j\in T_i$, $|T_i|=3\secparam Q^4$ and, $$\sum_{j \in [L] \setminus T_i} Z_{i,j} \geq 60Q^7 \secparam - 2 Q^4 \secparam.  $$
\end{claim}

\begin{proof}
By  the previous Claim, we have that with probability neglible close to $1$, for all $i \in [Q]$, there exists $T_i$ satisfying the desired properties, what is left is to show that $$\sum_{j \in [L] \setminus T_i} Z_{i,j} \geq 60Q^7 \secparam - 2 Q^4 \secparam$$
for all $i \in [Q]$.

For any $i \in [Q]$, we have that $\expct[\sum_{j \in [L] \setminus T_i} Z_{i,j}] = 60 Q^7 \secparam - \frac{3}{2} Q^4 \secparam$, and by Chernoff bound:
\begin{align*}
    \prob\left[\sum_{j \in [L] \setminus T_i} Z_{i,j} \leq 60Q^7\secparam - 2Q^4 \secparam \right] &= \prob\left[\sum_{j \in [L] \setminus T_i} Z_{i,j} \leq \left(60Q^7\secparam - \frac{3}{2}Q^4 \secparam\right) - \frac{1}{2}Q^4 \secparam \right]\\
    & \leq \exp{\left(-\frac{(\frac{1}{2}Q^4 \secparam)^2}{2 (60Q^7\secparam - \frac{3}{2}Q^4 \secparam)}\right)} \\
    & \leq \exp{\left(-\frac{(\frac{1}{2}Q^4 \secparam)^2}{2 (60Q^7\secparam)}\right)} \\
    &= e^{-\frac{Q \secparam}{480}}.
\end{align*}

Again, by union bound, we have that 
$$\prob\left[\exists i\in [Q], \sum_{j \in [L] \setminus T_i} Z_{i,j} \leq 60Q^7\secparam - 2Q^4 \secparam\right] \leq Q e^{-\frac{Q \secparam}{480}}.$$
\end{proof}

\noindent Combining these last two claims we conclude that the probability that there is a session $V^*_i$ for which the simulator does not have more than $60 Q^7 \secparam + Q^4 \secparam$ matched commitments is negligibly small in $\secparam$.

\newcommand{\rigged}{\mathsf{rigged}}
\newcommand{\nonrigged}{\mathsf{nonrig}}
\paragraph{Finishing Proof of Claim~\ref{ind WI}.}
We use the auxiliary claims from the previous section to complete the proof. 
\begin{comment}
We denote the final state in $\hybrid_{3.i}$ as follows: 
$$\alpha \ket{\Phi_{3.i,0}} + \beta \ket{\Phi_{3.i,1}},$$
where: 
\begin{itemize}
    \item $\ket{\Phi_{3.i,0}}$ is a superposition of all the transcripts containing $\geq 60Q^7 \secparam + Q^4 \secparam$ and, 
    \item $\ket{\Phi_{3.i,1}}$ is a superposition of rest of the transcripts.
\end{itemize}
Similarly, denote the final state in $\hybrid_{3.i+1}$ as follows: 
$$\alpha \ket{\Phi_{3.i+1,0}} + \beta \ket{\Phi_{3.i+1,1}},$$

\noindent We first claim that the distributions of $\ket{\Phi_{3.i,0}}$ and $\ket{\Phi_{3.i+1,0}}$ are computationally indistinguishable. 
\end{comment}
We prove this via the following hybrid argument. \\

\noindent \underline{$\hybrid_{3.i}^{(1)}$}: This is identical to the hybrid $\hybrid_{3.i}$. \\

%\noindent \underline{$\hybrid_{3.i}^{(2)}$}: This is the same as the previous hybrid except that the simulator before sending the $t^{th}$ round message, measures its register $\simreg_t$ and sends the result to $V^*$.
%\par The output distribitions of  $\hybrid_{3.i}^{(1)}$ and $\hybrid_{3.i}^{(2)}$ are identical.\\

\noindent \underline{$\hybrid_{3.i}^{(2)}$}: This is the same as the previous hybrid except that the simulator sets it responses, to the $i^{th}$ session, as $\bot$ if the number of matched slots for the $i^{th}$ session is $< 60Q^7 \secparam + Q^4 \secparam$.
\par From Claim~\ref{clm:main}, we have that the probability that this hybrid aborts is negligible in $\secparam$. Conditioned on this hybrid not aborting, the output distributions of $\hybrid_{3.i}^{(1)}$ and $\hybrid_{3.i}^{(2)}$ are identical.\\

\noindent \underline{$\hybrid_{3.i}^{(3)}$}:  This is identical to the hybrid $\hybrid_{3.i+1}$.

\par To argue that $\hybrid_{3.i}^{(3)}$ and $\hybrid_{3.i}^{(2)}$ are computationally indistinguishable we will use the quantum witness indistinguishable property of $\protwi$. Suppose that there is an adversary $\cA$ that distinguishes the output distributions of these two hybrids. We define the following QPT $\cB_i$ that breaks the security of $\protwi$. That is, $\cB_i$ is a QPT verifier, in the WI experiment, that can distinguish whether the prover used one witness versus another. $\cB_i$ is given as auxiliary advice a transcript (and verifier's private state) of $\hybrid_{3.i}^{(2)}$ executed until the verifier's first message of the $i^{th}$ WI execution in the transcript. In particular, conditioned on not aborting, this transcript has enough number of matching slots corresponding to the $i^{th}$ execution (in the order of arrival of messages) of WI. Then, $\cB_i$ interacts with the verifier $V^*$ as in the protocol with $P$ from then on, but forwards the verifier's messages corresponding to the $i^{th}$ WI execution to the prover of WI. The output of $\cB_i$ is the same as the output of the verifier $V^*$.
\par Firstly, from the security of WI, the output distribution of $\cB_i$ when the prover uses $w$ is computationally indistinguishable from the output distribution of $\cB_i$ when the prover uses the other witness, i.e., decommitments of matched slots.
\par If the prover used the witness $w$, then the output distribution of $\cB_i$ is computationally indistinguishable from the output of $\hybrid_{3.i}^{(2)}$. To see why, note that the only difference between $\cB_i$ and $\hybrid_{3.i}^{(2)}$ is that in $\cB_i$, all the blocks starting from the $i^{th}$ WI are not rewound. But we already showed, assuming security of commitments, that $V^*$ cannot distinguish the case when the block is being rewound versus the case when it is not. \par Furthermore, similarly, when the prover is using the decommitments of matched slots, the output distribution of $\cB_i$ is computationally indistinguishable from the output of $\hybrid_{3.i}^{(3)}$. 
\par Thus, the output distributions of $\hybrid_{3.i}^{(2)}$ and $\hybrid_{3.i}^{(3)}$ are computationally indistinguishable.

%% file: QPoK/postot.tex
\section{Post-Quantum Statistical Receiver Oblivious Transfer}
\label{sec:statOT}
We begin by presenting the definition of statistical receiver oblivious transfer with post-quantum security. We consider a natural adaption of the definition of~\cite{GJJM20} (see also~\cite{DGHMW20}), who originally defined in the classical setting. %Later, we will work with a different, albeit equivalent, definition.

The definition we provide is written this way to make it compatible with the 3-round OT definition from~\cite{GJJM20}.  The main difference is that we allow for an interactive phase instead of the sender's first message in~\cite{GJJM20}.

\subsection{Definition}
\label{sec:pqstatot}

\begin{definition}[Post-Quantum Statistical Receiver-Private Oblivious Transfer] 
\label{def:ot:definition1}
An oblivious transfer protocol, $\Pi_{OT}$, is an interactive protocol between a PPT sender and a PPT receiver $(\sender, \receiver)$, and a triplet of algorithms $(\ot_2,\ot_3, \ot_4)$ such that \\ %tuple of algorithms $(\ot_1,\ot_2,\ot_3,)$ which specifies the following protocol. \\

\noindent {\bf Interactive Phase}. $\sender$ and $\prover$ interact for $\poly(\secparam)$ rounds. The receiver's input is $\secparam$ and a bit $\beta \in \{0,1\}$. The sender's input is $\secparam$. Let $\smot_1$ be the transcript generated in this round, and let $\state_S$ and $\state_R$ be the private state of the sender and receiver (respectively) at the end of the round.\\

\noindent {\bf Receiver's Final Message}. The receiver $R$ computes $(\smot_2, \state_R') \leftarrow \ot_2(1^\secparam, \smot_1,\beta,\state_R)$ \\

\noindent {\bf Sender's Final Message}. $\sender$ with input $(m_0,m_1) \in \{0,1\}^2$ computes $(\smot_3, \state_S') \leftarrow \ot_3(1^{\secparam},\smot_2,\state_S,m_0,m_1)$, and it sends $\smot_3$ to $\receiver$. \\

\noindent {\bf Reconstruction}. The receiver computes $m' \leftarrow \ot_4(1^{\secparam},\smot_3,\state_R')$. Output $m'$. \\

\noindent {\bf Correctness.} For any $\beta \in \{0,1\}, (m_0,m_1) \in \{0,1\}^2$, we have: 
$$\prob \left[ \substack{ (\smot_1,\state_S,\state_R)\leftarrow \langle\sender(1^\secparam), \receiver(1^\secparam,\beta) \rangle\\ (\smot_2,\state_R') \leftarrow \ot_2(1^\secparam,\smot_1,\beta,\state_R)
\\ (\smot_3,\state_S') \leftarrow \ot_3(1^{\secparam},\smot_2,\state_S,m_0,m_1)\\ m' \leftarrow \ot_4(1^{\secparam},\smot_3,\state_R')}\ :\ m' = m_{\beta} \right] = 1$$

\noindent {\bf Statistical Receiver-Privacy.} For any sender $S^*$, denote $(\smot_1^{(0)}, \state_R^{(0)}) \leftarrow \langle S^*(1^\secparam), R(1^\secparam,0) \rangle$ and $(\smot_1^{(1)}, \state_R^{(1)}) \leftarrow \langle S^*(1^\secparam), R(1^\secparam,1) \rangle$. Furthermore, let $(\smot_2^{(0)},(\state_R^{(0)})') \leftarrow \ot_2(1^\secparam, \smot_1^{(0)},0,\state_R^{(0)})$ and let $(\smot_2^{(1)},(\state_R^{(1)})') \leftarrow \ot_2(1^\secparam, \smot_1^{(1)},1,\state_R^{(1)})$.

Then the statistical distance between the marginal distributons $\{(\smot_1^{(0)}, \smot_2^{(0)})\}$ and $\{(\smot_1^{(1)},\smot_2^{(1)})\}$ is a negligible function in $\secparam$.\\
%\noindent {\bf Statistical Receiver-Privacy.} For any (potentially maliciously generated) $\ot_1^*$, denote $\left( \smot_2^{(0)},\state_R^{(0)} \right) \leftarrow \ot_2(1^{\secparam},\ot_1^*,0)$ and $(\smot_2^{(1)},\state_R^{(1)}) \leftarrow \ot_2(1^{\secparam},\ot_1^*,1)$. Then the statistical distance between the marginal distributons $\{\smot_2^{(0)}\}$ and $\{\smot_2^{(1)}\}$ is a negligible function in $\secparam$.\\

\noindent {\bf Post-Quantum  Sender-Privacy.} For any non-uniform QPT distinguisher $\adversary$ and any malicious receiver $\receiver^*$, which receives as input state that is possibly entangled with the input state of $\adversary$, we define the following games. \\

\noindent {\em Interact with $\receiver^*$.} The challenger plays the role of an honest sender in the interactive phase with $\receiver^*$. %Specifically, the challenger executes $(\smot_1,\state_S) \leftarrow \ot_1(1^{\secparam})$. Then send $\smot_1$ to $\receiver^*$. 
Then the receiver $\receiver^*$ outputs 
a state in a register $\regB$ and a message $z$. The message $z$ is sent to the challenger. The register $\regB$ is given to $\adversary$. \\

\noindent \underline{{\em Game $G_0(m_0,m_1)$}}: The challenger samples $b_0 \leftarrow \{0,1\}$ at random and computes $ot_3 \leftarrow \ot_3(1^{\secparam},z,\state_S,\allowbreak m_{b_0},m_1)$. Then, $ot_3$ is sent to $\adversary$. Finally, $\adversary$ outputs two bits $b'_0$ and $b'_1$. If $b_0=b'_0$ then we say that $\adversary$ wins the game $G_0$. \\

\noindent \underline{{\em Game $G_1(m_0,m_1)$}}: The challenger samples $b_1 \xleftarrow{\$} \{0,1\}$ at random, and then computes $ot_3 \leftarrow \ot_3(1^{\secparam},z,\state_S,\allowbreak m_0,m_{b_1})$. Then, $ot_3$ is sent to $\adversary$. Finally, $\adversary$ outputs two bits $b'_0$ and $b'_1$. If $b_1=b'_1$ then we say that $\adversary$ wins the game $G_1$.  \\

\noindent We define the advantage as follows: 
$$\advantage(\adversary,\receiver^*,m_0,m_1) = \mathbb{E}_{\viewR} \left[ \min\left\{ p_0,p_1 \right\}  \right],$$
where: 
\begin{itemize}
    \item $p_0 = \left| \prob[\adversary \text{ wins } \game_0(m_0,m_1)] - \frac{1}{2} \right| $
    \item $p_1 = \left| \prob[\adversary \text{ wins } \game_1(m_0,m_1)] - \frac{1}{2} \right| $
\end{itemize}
We say that the oblivious transfer scheme is computational sender-secure if for every $m_0,m_1 \in \{0,1\}$, we have $\advantage(\adversary,\receiver^*,m_0,m_1)$ to be negligible in $\secparam$. 
\end{definition} 

\submversion{
\input{Overview/statot-overview}
}
%%%%%%%%%%%%%%%%%%%%%%%%%%%%%%%

\subsection{Main Tools}
\noindent To construct a statistical receiver-private oblivious transfer, we use two tools: (i) a statistical zero-knowledge argument system and, (ii) statistical sender-private oblivious transfer. 

\input{QPoK/statzk}
\input{QPoK/statsenderot}

\subsection{Construction}

\begin{itemize}

    \item A 2-round post-quantum statistical sender-private OT, $\Pi_{OT}=(OT_1,OT_2,OT_3)$. Without loss of generality, we assume that the length of the randomness is $\secparam$. 
    
     We say that a transcript $\tau$, consisting of messages $(\msg_1,\msg_2)$, is valid with respect to sender's randomness $r$ and its input bits $(m_0,m_1)$ if the following holds: $(\msg_2,\state) \leftarrow \ot_2(1^{\secparam},\msg_1,m_0,m_1;r)$.

    \item A statistical zero-knowledge quantum argument system, $\Pi_{\zk}$, for the NP relation $\rel(\lang_\zk)$. We described the relation $\rel(\lang_\zk)$, parametrized by security parameter $\secparam$, described below:

     $$\rel\left(\lang_{\zk} \right)=\Bigg\{\left(\left(\tau_{\ot}^*,\{\tau_{\ot}^{(i,j)},b_{i,j}\}_{i \in [\secparam+2],j \in [\secparam]} \right);\ \left(r',\beta,r_{\ot}^*, \left\{r_{\ot}^{(i,j)},sh_{i,j},\alpha_{i,j} \right\}_{i \in [\secparam+2],j \in [\secparam]}\right) \right)\ :\ $$
     $$ \left(\substack{\forall i \in [\secparam+2],j \in [\secparam],\\ \tau_{\ot}^{(i,j)} \text{ is valid w.r.t }\\
     r_{\ot}^{(i,j)} \text{ and }
     (((1-b_{i,j})sh_{i,j} + b_{i,j} \cdot \alpha_{i,j}),\  (b_{i,j}sh_{i,j} + (1-b_{i,j}) \cdot \alpha_{i,j}))\\ \bigwedge \\
     \tau_{\ot}^* \text{ is valid w.r.t } r_{\ot}^* \text{ and } (r',r'\oplus \beta) } \right) \bigwedge \left(\substack{\forall i \in [\secparam+2],\\ \oplus_{j=1}^{\secparam} sh_{i,j} = w_i} \right) \bigwedge w=(r',\beta,r_{\ot}^*) \Bigg\}$$

\end{itemize}

%\begin{comment}

\newpage
%%%%%%%%%%%%%%%%%%%%%%% FULL VERSION 
\begin{figure}[!htb]
   \begin{center}
   \begin{tabular}{|p{16cm}|}
    \hline \\
    {\bf Input of sender $S$}: $(m_0,m_1)$.  \\
    {\bf Input of receiver $R$}: $\beta$. \\
    \begin{itemize}
    \setlength\itemsep{0.7em}
        \item $R$ generates $r' \xleftarrow{\$} \{0,1\}$ uniformly at random. $R$ samples $r^*_{OT} \xleftarrow{\$} \{0,1\}^{\secparam}$. 
        \item Let $w=(r',\beta,r^*_{OT})$. For every $i \in [\secparam+2]$, $R$ generates shares $sh_{i,1},\ldots,sh_{i,\secparam}$ uniformly at random conditioned on $ \oplus_{j=1}^{\secparam} sh_{i,j} = w_i$. 
        \item For $i \in [\secparam+2]$, $R$ also generates bits $\alpha_{i,1},\ldots,\alpha_{i,\secparam}$ uniformly at random. %It also generates the bits $\alpha_{1,1},\ldots,\alpha_{1,\secparam}$ uniformly at random.
       % \par $P$ commits to $sh_{i,j}$ using $\comm$ to obtain the commitment $\bfc_{i,j}$, for every $i \in [\ell_w],j \in [\secparam]$. Call the tuple of commitments $\{\bfc_{i,j}\}_{i \in [\ell_w],j \in [\secparam]}$ to be $\overrightarrow{\bfc}$. 
        
        \item For $i \in [\secparam+2],j \in [\secparam]$, do the following: 
        \begin{itemize} 
        
        \item $S \leftrightarrow R$: $S$ and $R$ execute $\Pi_{OT}$ with $S$ playing the role of the receiver, denoted by $R'_{OT}$, in $\Pi_{OT}$ and $R$ playing the role of the sender, denoted by $S'_{OT}$, in $\Pi_{\ot}$. The input of the receiver $R'_{OT}$ in this protocol is 0, while the input of the sender $S'_{OT}$ is set to be $(sh_{i,j},\alpha_{i,j})$ if $b_{i,j}=0$, otherwise it is set to be $(\alpha_{i,j},sh_{i,j})$ if $b_{i,j}=1$, where the bit $b_{i,j}$ is sampled uniformly at random by $S'_{OT}$. We call this execution as $(i,j)^{th}$ execution of $\Pi_{OT}$.
        \par Call the resulting transcript of the protocol to be $\tau_{OT}^{(i,j)}$ and let $r_{OT}^{(i,j)}$ be the randomness used by the sender $S'_{OT}$ in $\Pi_{OT}$. 
        
        \item $R \rightarrow S$: $R$ sends $b_{i,j}$ to $S$. 
        
        %\item $P \leftrightarrow V$: $P$ and $V$ execute $\Pi_{\zk}^{(1)}$ with $P$ playing the role of the prover of $\Pi_{\zk}^{(1)}$ and $V$ playing the role of the verifier of $\Pi_{\zk}^{(1)}$. The input of the prover in $\Pi_{\zk}^{(1)}$ is the instance $\left(\tau_{\ot}^{(i,j)},i,j \right)$ and the associated witness is $\left( r_{\ot}^{(i,j)},sh_{i,j},b_{i,j} \right)$. The input of the verifier is the instance $\left(\tau_{\ot}^{(i,j)},i,j \right)$. If the verifier in $\Pi_{\zk}^{(1)}$ rejects, then $V$ rejects.
        
        \end{itemize}
        
    \item $S$ samples $r \xleftarrow{\$} \{0,1\}$. 
        
    \item $S \leftrightarrow R$: $S$ and $R$ execute $\Pi_{OT}$ with $S$ playing the role of the receiver, denoted by $R'_{OT}$, in $\Pi_{OT}$ and $R$ playing the role of the sender, denoted by $S'_{OT}$, in $\Pi_{\ot}$. The input of the receiver $R'_{OT}$ is $r$ and the input of the sender is $(r',r' \oplus \beta)$. Let $\widetilde{r}$ be the bit recovered by $R'_{OT}$ at the end of the protocol. 
    \par We call this execution the main execution of $\Pi_{OT}$. Call the resulting transcript of the protocol to be $\tau_{OT}^{*}$ and let $r^*_{OT}$ be the randomness used by the sender $S'_{OT}$ in $\Pi_{OT}$. 
        
    \item $S \leftrightarrow R$: $S$ and $R$ execute $\Pi_{\zk}$ with $R$ playing the role of the prover $P$ of $\Pi_{\zk}$ and $S$ playing the role of the verifier $V$ of $\Pi_{\zk}$. The instance is $\left(\tau_{OT}^*,\left\{ \tau_{OT}^{(i,j)}, b_{i,j} \right\}_{i \in [\ell_w],j \in [\secparam]} \right)$ and the witness is $\left( r',\beta,r^*_{OT}, \left\{r_{OT}^{(i,j)},sh_{i,j}, \alpha_{i,j} \right\}_{i \in [\ell_w],j \in [\secparam]} \right)$. If the verifier in $\Pi_{\zk}$ rejects, then $S$ aborts.
        
    \item $S$ sends $\left(\widetilde{r} \oplus m_0,\widetilde{r} \oplus r \oplus m_1 \right)$. 
        
    \end{itemize}
    \\ 
    \hline
   \end{tabular}
    \caption{}
    \label{fig:statot}
    \end{center}
\end{figure}

%\end{comment}

\noindent We show that the construction in~\Cref{fig:statot} is a post-quantum statistical receiver oblivious transfer protocol.
\paragraph{Correctness.} The correctness of our protocol follows from the correctness of $\Pi_{OT}$ and the completeness of $\Pi_{\zk}$. 

\paragraph{Statistical Receiver Privacy.} Let $S^*$ be a computationally unbounded sender. Instead of proving that the sender cannot distinguish receiver's bit to be 0 versus receiver's bit to be 1 with non-negligible probability, we instead prove the following: suppose receiver chooses its bit uniformly at random then the probability that the malicious sender can output $\beta$ with probability negligibly close to $\frac{1}{2}$. We prove this via a hybrid argument. In the first hybrid, the receiver behaves honestly and uses the receiver's bit to be $\beta$, where $\beta$ is chosen uniformly at random. We define a sequence of hybrids and show computational indistinguishability of every pair of consecutive hybrids. In the final hybrid, the receiver's bit will be information-theoretically hidden in the messages exchanged with $S^*$, which will prove the statistical receiver privacy property. \\

\noindent \underline{$\hybrid_1$}: In this hybrid, the receiver uses the bit $\beta$. 
\par Let $\varepsilon$ be the probability that $S^*$ outputs $\beta$. \\

\noindent \underline{$\hybrid_2$}: Let $\simr_{\zk}$ be the simulator associated with $\Pi_{\zk}$. Instead of $R$ playing the role of the prover in $\Pi_{\zk}$, it executes $\simr_{\zk}$.
\par From the statistical zero-knowledge property of $\Pi_{\zk}$, the output distributions of $S^*$ in the hybrids $\hybrid_1$ and $\hybrid_2$ are statistically close. The probability that $S^*$ outputs $\beta$ is negligibly close to $\varepsilon$ in this hybrid. \\

\noindent \underline{$\hybrid_{3.(i,j)}$, for $i \in [\secparam+2]$, $j \in [\secparam]$}: In the $(i,j)^{th}$ execution of $\Pi_{OT}$, perform inefficient extraction to extract $b'_{i,j}$ from $R'_{OT}$. Recall that $S^*$ plays the role of $R'_{OT}$. If $b'_{i,j}=b_{i,j}$ then set the input of the sender $S'_{OT}$ to be $(sh_{i,j},sh_{i,j})$ and if $b'_{i,j}\neq b_{i,j}$ then set the input of the sender $S'_{OT}$ to be $(\alpha_{i,j},\alpha_{i,j})$.   
\par The statistical indistinguishability of this hybrid and the previous hybrid follows from the statistical sender-privacy property of $\Pi_{OT}$. The probability that $S^*$ outputs $\beta$ is negligibly close to $\varepsilon$ in this hybrid. \\

\noindent \underline{$\hybrid_{4}$}: If there exists $i \in [\secparam+2]$, such that for every $j \in [\secparam]$, $b_{i,j}=b'_{i,j}$ then abort. 
\par The probability that this hybrid aborts is at most $\frac{\secparam+2}{2^{\secparam}}$. Conditioned on this hybrid not aborting, this hybrid is identical to the previous one. The probability that $S^*$ outputs $\beta$ is negligibly close to $\varepsilon$ in this hybrid. \\

\noindent \underline{$\hybrid_5$}: In the main execution of $\Pi_{OT}$, perform inefficient extraction to extract $r$ from $\Pi_{OT}$. If $r=0$, then set the input of the sender $S'_{OT}$ to be $(r',r')$ and if $r=1$, then set the input of the sender to be $(r' \oplus \beta,r' \oplus \beta)$. 
\par The statistical indistinguishability of $\hybrid_{4}$ and $\hybrid_5$ follows from the statistical sender-privacy property of $\Pi_{OT}$. The probability that $S^*$ outputs $\beta$ is negligibly close to $\varepsilon$ in this hybrid. \\

\noindent Note that in $\hybrid_5$, the receiver's bit $\beta$ is information-theoretically hidden in the messages exchanged with $S^*$. Thus, the probability that $S^*$ guesses $\beta$ in $\hybrid_5$ is $\frac{1}{2}$. This proves that the probability that the receiver outputs $\beta$ in $\hybrid_1$ is negligibly close to $\frac{1}{2}$. 

\paragraph{Post-Quantum Sender Privacy.} Let $R^*$ be a QPT receiver and $\adversary$ be a QPT adversary such that the following holds for some $m_0 \in \{0,1\},m_1 \in \{0,1\}$,
$$\mathbb{E}_{\viewR} \left[ \min\left\{ p_0,p_1 \right\}  \right] \geq \nu(\secparam),$$
where: 
\begin{itemize}
    \item $\viewR$ is the view of $R^*$. 
    \item $p_0 = \left| \prob[\adversary \text{ wins } \game_0(m_0,m_1)] - \frac{1}{2} \right| $
    \item $p_1 = \left| \prob[\adversary \text{ wins } \game_1(m_0,m_1)] - \frac{1}{2} \right| $
\end{itemize}

\noindent for some non-negligible fucntion $\nu(\secparam)$, where $G_0,G_1$ are defined with respect to $R^*$ and $\adversary$ as in~\Cref{sec:pqstatot}. We define $p_0^{({\bf i})}$ to be the absolute difference of the probability that $\adversary$ wins in the game $G_0$ and $\frac{1}{2}$ in the hybrid $\hybrid_{{\bf i}}$. Similarly, we define $p_1^{({\bf i})}$.     \\

\noindent Consider the following hybrids. \\

\noindent \underline{$\hybrid_1$}: This hybrid corresponds to the real execution of the protocol. 
\par By our initial assumption, we have $\mathbb{E}_{\viewR} \left[ \min\left\{ p_0,p_1 \right\}  \right] \geq \nu(\secparam)$.\\

\noindent \underline{$\hybrid_2$}: In this hybrid, defer the measurements of the receiver until the end. %of the interactive phase. In other words, the receiver is ran in superposition in all the $(i,j)^{th}$ OT executions.
\par The output distributions of $\hybrid_1$ and $\hybrid_2$ are identical. Thus, $\mathbb{E}_{\viewR} \left[ \min\left\{ p_0^{(2)},p_1^{(2)} \right\}  \right] \geq \nu(\secparam)$. \\

\noindent \underline{$\hybrid_{3.(i,j)}$ for every $i \in [\secparam+2],j \in [\secparam]$}: Instead of computing $S$, perform the following hybrid extractor as follows. \\
We first give a description of the registers used in the system. 

\begin{itemize}
    \item $\randreg_{i,j}$ for $i \in [\secparam+2],j \in [\secparam]$: this consists of the sender $S$ -- recall that $S$ is taking the role of the receiver $R'$ of the $(i,j)^{th}$ execution of the OT -- randomness used by the extractor in the $(i,j)^{th}$ executions of $\Pi_{\ot}$.  
    \item $\bitreg_{i,j}$, for $i \in [\secparam+2],j \in [\secparam]$: this is a single-qubit register that contains a bit that is used by the extractor in the $(i,j)^{th}$ execution of the OT protocol.
    \item $\decreg$: it contains the decision register that indicates whether to rewind or not.
    \item $\auxreg$: this is initialized with the auxiliary state of the receiver.
    \item $\transreg_{i,j}$, for $i \in [\secparam+2],j \in [\secparam]$: it contains the transcripts of the $(i,j)^{th}$ executions of the OT protocol.
    \item $\transreg^*$: it contains the transcript of the protocol $\Pi_{\zk}$.
    \item $\regX$: this is a $\poly(\secparam)$-qubit ancillary register. 
\end{itemize}
\noindent {\em Description of $\hybrid_{3.(i,j)}.\extractor(x,\ket{\Psi})$}: The state of the extractor is initialized as follows: 

    $$\left( \overset{}{\underset{i \in [\ell_w],\\ j \in [\secparam]}{\bigotimes}} \ket{0}_{\bitreg_{i,j}} \ket{0}_{\randreg_{i,j}}\ket{0}_{\transreg_{i,j}} \right) \otimes \ket{0}_{\transreg^*} \otimes  \ket{\Psi}_{\auxreg} \otimes \ket{0}_{\decreg} \otimes \ket{0^{\otimes \poly(\secparam)}}_{\regX} $$

\begin{itemize}
    \item For $i' \in [\secparam+2],j' \in [\secparam]$ such that $(i',j') \geq (i,j)$, perform the following operations in superposition:

    \begin{itemize}
        \item Let $\ket{\widetilde{\Psi}}$ be the state of the system at the beginning of the $(i,j)^{th}$ execution. 
        \item Prepare the following state\footnote{We will assume, without loss of generality, that the length  of the random strings used in the OT protocol be $\secparam$.}: 
        $$ \ket{0}_{\bitreg_{i,j}}\ket{0}_{\randreg_{i,j}} \rightarrow \frac{1}{\sqrt{2^{\secparam+1}}} \sum_{\beta_{i,j} \in \{0,1\},s_{\ot}^{(i,j)} \in \{0,1\}^{\secparam}} \ket{\beta_{i,j}}_{\bitreg_{i,j}} \ket{s_{\ot}^{(i,j)}}_{\randreg_{i,j}} $$
        
        \item It then performs the $(i,j)^{th}$ execution of $\Pi_{\ot}$ along with the $R^*$'s message immediately after the $(i,j)^{th}$ execution of $\Pi_{\ot}$ in superposition. The resulting transcript is stored in the register $\transreg_{i,j}$. We denote the unitary that performs this step to be $U_{i,j}^{(1)}$.

        \item After $R^*$ sends the message immediately after the $(i,j)^{th}$ execution of $\Pi_{\ot}$, apply the  unitary $U_{i,j}^{(2)}$ defined as follows: 
        \begin{eqnarray*}
        & & U_{i,j}^{(2)} \ket{\beta_{i,j}}_{\bitreg_{i,j}} \ket{s_{\ot}^{(i,j)}}_{\randreg_{i,j}} \ket{\tau_{\ot}^{(i,j)},b_{i,j}}_{\transreg_{i,j}} \ket{0}_{\decreg}\\
        & = & \left\{ \begin{array}{ll} \ket{\beta_{i,j}}_{\bitreg_{i,j}} \ket{s_{\ot}^{(i,j)}}_{\randreg_{i,j}}\ket{\tau_{\ot}^{(i,j)},b_{i,j}}_{\transreg_{i,j}} \ket{{\sf Match}_{i,j}}_{\decreg}\ \text{if }{\sf acc}_{i,j}=1, \\ \ \\ \ket{\beta_{i,j}}_{\bitreg_{i,j}} \ket{s_{\ot}^{(i,j)}}_{\randreg_{i,j}}\ket{\tau_{\ot}^{(i,j)}, b_{i,j}}_{\transreg_{i,j}} \ket{+}_{\decreg},\ \text{ otherwise}\end{array} \right.
        \end{eqnarray*}
        Here, ${\sf Match}_{i,j}=0$ if and only if $\beta_{i,j}=b_{i,j}$, where $b_{i,j}$ is the bit sent by $R^*$ after the $(i,j)^{th}$ execution of the OT protocol. Moreover, ${\sf acc}_{i,j}=1$ only if $R^*$ has not aborted in $(i,j)^{th}$ OT execution. 
        \par Let $W_{i,j} = \amplifier\left(U_{i,j}^{(2)} U_{i,j}^{(1)} \right)$, where $\amplifier$ is defined in Lemma~\ref{lem:watrous}. Perform $W_{i,j} \ket{\widetilde{\Psi}}$ to obtain $\ket{\Psi_{i,j}}$.
    \end{itemize}
    \item For $i' \in [\secparam+2],j' \in [\secparam]$, such that $(i',j')<(i,j)$ perform the $(i',j')^{th}$ execution of $\Pi_{OT}$ as in the previous hybrid. 
    \item Perform the main execution of $\Pi_{OT}$ in superposition.
    \item Perform the execution of $\Pi_{\zk}$ in superposition. 
    \item Measure all the registers except $\auxreg$. Perform the OT reconstruction on input the measured transcript $\tau_{\ot}^{i,j}$, for $i \in [\secparam+2], j \in [\secparam]$, measured randomness $s_{\ot}^{i,j}$ and receiver's bit $b_{i,j}$. Call the resulting reconstruction output to be $\widetilde{sh_{i,j}}$. Let $\widetilde{u}_i = \oplus_{j=1}^{\ell_w} \widetilde{sh_{i,j}}$. Let $(r',\beta,r^*_{OT})$ be the concatenation of all the $\widetilde{u}_i$ bits. If either the $S'_{OT}$'s inputs to the main execution of $\Pi_{OT}$ is not $(r',r' \oplus \beta)$ or if $S'_{OT}$'s randomness is not $r^*_{OT}$ then abort. Otherwise output the state in $\auxreg$ along with $w$.  
\end{itemize}
\noindent From the post-quantum computational receiver privacy of $\Pi_{OT}$, it holds that $\hybrid_{3.(i,j)}$ and the previous hybrid are computationally indistinguishable. Thus, the following holds:
\par $\mathbb{E}_{\viewR} \left[ \min\left\{ p_0^{(3.(i,j))},p_1^{(3.(i,j))} \right\}  \right] \geq \nu_{3.(i,j)}(\secparam)$, where $\nu_{3.(i,j)}$ is a non-negligible function. \\

\noindent \underline{$\hybrid_4$}: In the main execution of $\Pi_{OT}$, the input of $R'_{OT}$ is set to be 0. Recall that in the previous hybrids, the input of $R'_{OT}$ was $r$.  
\par From the post-quantum computational receiver privacy of $\Pi_{OT}$, it holds that the hybrids $\hybrid_4$ and $\hybrid_{3.(\secparam+2,\secparam)}$ are computationally indistinguishable. Thus, the following holds:
\par $\mathbb{E}_{\viewR} \left[ \min\left\{ p_0^{(4)},p_1^{(4)} \right\}  \right] \geq \nu_{4}(\secparam)$, where $\nu_{4}$ is a non-negligible function.\\ 

\noindent \underline{$\hybrid_5$}: Sample $u \xleftarrow{\$} \{0,1\}$. If $\beta=1$, set the last message to be $(u,r' \oplus m_1)$. Else if $\beta=0$, set the last message to be $(r' \oplus m_0,u)$. 
\par From the computational soundness of $\Pi_{\zk}$, it follows that $\beta$ extracted from all the $(i,j)^{th}$, for $i \in [\secparam+2],j \in [\secparam]$, executions of $\Pi_{OT}$ is the same as the $\beta$ used by the receiver in the main execution of $\Pi_{OT}$ with probability negligibly close to 1. This further implies that the bit reconstructed by $R'_{OT}$ in the main execution is $\widetilde{r} = r' \oplus (r \cdot \beta)$. Thus, the last message sent by $S$ can be rewritten as follows: $(\widetilde{r} \oplus m_0,\widetilde{r} \oplus r \oplus m_1)=(r' \oplus (r \cdot \beta) \oplus m_0,r' \oplus (r \cdot \beta) \oplus r \oplus m_1)$. If $\beta=1$, we have the message sent by $S$ to be $(r' \oplus r\oplus m_0,r' \oplus m_1)$. If $\beta=0$, we have the message sent by $S$ to be $(r' \oplus m_0,r' \oplus r \oplus m_1)$. We now use the fact that $r$ is information-theoretically hidden from the receiver $R^*$ to show that the hybrids $\hybrid_4$ and $\hybrid_5$ are computationally indistinguishable. Thus, the following holds: 
\par  $\mathbb{E}_{\viewR} \left[ \min\left\{ p_0^{(5)},p_1^{(5)} \right\}  \right] \geq \nu_{5}(\secparam)$, where $\nu_{5}$ is a non-negligible function.\\

\noindent But one of the two sender's inputs are information-theoretically hidden from the malicious receiver; in one of the two games $G_0$ or $G_1$, the adversary can win only with negligible probability. This contradicts the fact that $\mathbb{E}_{\viewR} \left[ \min\left\{ p_0^{(5)},p_1^{(5)} \right\}  \right]$ is non-negligible. Thus, our construction satisfies post-quantum sender privacy.

%% file: QPoK/statzk.tex
\subsubsection{Statistical Zero-Knowledge Quantum Argument System}
\label{sec:statzk}

\begin{definition}[Statistical ZK Quantum Argument System]
Let $\Pi$ be an interactive protocol between a classical PPT prover $P$ and a classical PPT verifier $V$. Let $\rel(\lang)$ be the NP relation associated with $\Pi$.
\par $\Pi$ is said to satisfy {\bf completeness} if the following holds: 
\begin{itemize} 
\item {\bf Completeness}: For every $(x,w) \in \rel(\lang)$, 
$$\prob[\accept \leftarrow \langle P(x,w),V(x) \rangle] \geq 1 - \negl(\secparam),$$
for some negligible function $\negl$. 
\end{itemize}
\par $\Pi$ is said to satisfy {\bf (quantum computational) soundness} if the following holds: 
\begin{itemize}
    \item {\bf (Quantum Computational) Soundness}: For every  QPT adversary $P^*$, every $x \notin \rel(\lang)$, every $\poly(\secparam)$-qubit advice $\rho$,
    $$\prob\left[ \accept \leftarrow \langle P^*(x,\rho),V(x) \rangle  \right] \leq \negl(\secparam),$$
   for some negligible function $\negl$. 
\end{itemize}
\par $\Pi$ is said to satisfy {\bf statistical zero-knowledge} if the following holds: 
\begin{itemize}
    \item {\bf Statistical Zero-Knowledge}:For every sufficiently large $\secparam \in \mathbb{N}$, every computationally unbounded adversary $V^*$, there exists a QPT simulator $\simr$ such that for every $(x,w) \in \rel(\lang)$, the state output by $V^*$ is close in trace distance to the state output by $\simr$. 
\end{itemize}
\end{definition}

\newcommand{\nizk}{\mathsf{nizk}}
\newcommand{\bfa}{\mathbf{a}}
\newcommand{\crs}{\mathsf{CRS}}

To construct a statistical ZK quantum argument system, we will use a non-interactive (statistical) ZK protocol for NP (NIZK).  We define NIZK below.

\begin{definition}[Non-interactive statistical ZK argument system] A NIZK argument system, $\Pi_{\nizk}$, for an NP relation $\rel(\lang)$ in the common random string model is a triplet of PPT algorithms:
\begin{itemize}
    \item $\mathsf{Gen}(1^\secparam)$: outputs a public random string $\crs$. The output distribution of $\crs$ is generated according to the uniform distribution. 
    \item $P(\crs,x,w)$:  On instance $(x;w) \in \rel(\lang)$ outputs a proof $\pi$.
    \item $V(\crs,x,\pi)$: Outputs $1$ if it accepts the proof, and $0$ if it rejects.
\end{itemize}
$\Pi_{\nizk}$ is said to satisfy \textbf{completeness} if the following holds:
\begin{itemize}
    \item {\bf Completeness:} For all $(x,w) \in \rel(\lang)$,
    $$\Pr \left[\substack{\crs \leftarrow \mathsf{Gen}(1^\secparam) \\ 
    \pi \leftarrow P(\crs,x,w)} : V(\crs,x,\pi)=1 \right] = 1 $$
\end{itemize}
$\Pi_{\nizk}$ is said to satisfy \textbf{(quantum computational) soundness}  if the following holds:
\begin{itemize}
    \item {\bf (Quantum Computational) Soundness:} For all $x \notin \lang$, for any QPT $P^*$ and auxiliary $\poly(\secparam)$-qubits state $\rho$,
    $$\Pr \left[\substack{\crs \leftarrow \mathsf{Gen}(1^\secparam) \\ 
    \pi \leftarrow P^*(\crs,x,\rho)} : V(\crs,x,\pi)=1 \right] = \negl(\secparam) $$
\end{itemize}
$\Pi_{\nizk}$ is said to satisfy \textbf{statistical zero-knowledge}  if the following holds:
\begin{itemize}
    \item {\bf Statistical Zero-Knowledge:} There exists a QPT simulator $\Simu$ such that for all $(x,w) \in \rel(\lang)$, the following two distributions are statistically close: 
    \begin{enumerate}
        \item Sample $\crs \leftarrow \mathsf{Gen}(1^\secparam)$, sample $\pi \leftarrow P(\crs,x,w)$.  Output $(\crs,\pi)$.
        \item Sample $(\crs^*, \pi^*) \leftarrow \Simu(1^\secparam,x)$.  Output $(\crs^*,\pi^*)$.
    \end{enumerate}
\end{itemize}
\end{definition}

\paragraph{Instantiation.} The work of~\cite{PS19} shows how to construct statistical NIZK arguments for NP in the LWE. We note that the same construction and proof can be ported to the quantum setting to demonstrate a construction of statistical NIZK quantum argument system for NP from QLWE. 
A discussion on the quantum security of~\cite{PS19} can be found in~\cite{CVZ20}.

\subsubsection{Construction}

In order to construct a statistical ZK quantum argument system for an NP relation $\rel(\lang)$, we use the following ingredients:

\begin{itemize}
    \item A quantum zero-knowledge protocol $\Pi_{\zk}$ for the NP relation $\rel(\lang_\zk)$. We described the relation $\rel(\lang_\zk)$, parametrized by security parameter $\secparam$, described below:
    \[
    \rel(\lang_\zk) = \{\left((crs, \bfc, b); (a,\bfr)\right) : \substack{crs=a \oplus b \\ \bigwedge \\
    \bfc = \comm(1^\secparam, a; \bfr)}   \}
    \]
    \item A perfectly binding and quantum computationally hiding commitment scheme, $\comm$, where the length of randomness is $\secparam$.
    \item A non-interactive statistical zero-knowledge argument system $\Pi_{\nizk}$ for $\rel(\lang)$, where the length of the CRS is $q(\secparam)$. 
\end{itemize}

%%%%%%%%%%%%%%%%%%% FULL VERSION
We present a construction in~\Cref{fig:statzk}. 
\begin{figure}[!htb]
   \begin{center}
   \begin{tabular}{|p{12cm}|}
    \hline \\
    {\bf Instance}: $x$.  \\
    {\bf Witness}: $w$. \\
    \begin{itemize}
        \item $V$ samples $\bfa \xleftarrow{\$} \{0,1\}^{\secparam}$. 
        \item For every $i \in [q(\secparam)]$, $P$ and $V$ perform the following operations: 
        \begin{itemize}
            \item $V \rightarrow P$: $V$ computes $\bfc_i \leftarrow \comm(1^{\secparam},a_i;\bfr_i)$, where $\bfr_i \xleftarrow{\$} \{0,1\}^{\secparam}$ and $a_i$ is the $i^{th}$ bit of $\bfa$. $V$ sends $\bfc_i$ to $P$. 
            \item $P \rightarrow V$: $P$ sends $b_i$ to $V$, where $b_i \xleftarrow{\$} \{0,1\}$. 
            \item $V \rightarrow P$: $V$ sends $crs_i = a_i \oplus b_i$ to $P$. 
            \item $V \leftrightarrow P$: $P$ and $V$ will execute $\Pi_{\zk}$, with $P$ playing the role of verifier in $\Pi_{\zk}$ and $V$ plays the role of the prover. The instance is $\left(crs_i,\bfc_i,b_i \right)$ and the witness is $\left( a_i,\bfr_i \right)$. 
        \end{itemize}
        
        \item Set $\crs=\left(crs_1,\ldots,crs_{q(\secparam)} \right)$. 
        
        \item $P \rightarrow V$: $P$ computes a proof $\pi$ on input common random string $\crs$,  instance $x$ and witness $w$ using $\Pi_{\nizk}$. It sends $\pi$ to $V$.
        \item $V$ computes the verifier of $\Pi_{\nizk}$ on input $(\crs,x,\pi)$. It outputs the decision bit of the verifier of $\Pi_{\nizk}$. 
    \end{itemize}
    \\ 
    \hline
   \end{tabular}
    \caption{Statistical ZK Quantum Argument System}
    \label{fig:statzk}
    \end{center}
\end{figure}

\paragraph{Completeness.} The completeness follows from the completeness of $\Pi_{\zk}$ and $\Pi_{\nizk}$. 

\paragraph{Quantum Computational Soundness.} Let $x \notin \lang$. Suppose $P^*$ be a QPT prover that on input $(x,\rho)$, for some $\poly(\secparam)$-qubit advice $\rho$, convinces the verifier $V^*$ to accept $x$ with probability $\varepsilon$. We prove that $\varepsilon$ is negligible using a hybrid argument. \\

\noindent \underline{$\hybrid_1$}: This corresponds to the execution of $P^*$ and $V$. The probability that $V$ accepts is $\varepsilon$. \\  

\noindent \underline{$\hybrid_{2.i}$ for $i \in [q(\secparam)]$}: We consider a hybrid verifier $\hybrid_{2.i}.V$ that executes the simulator $\simr$ in the $i^{th}$ execution of $\Pi_{\zk}$, instead of running the real prover. Except this change, the hybrid verifier  $\hybrid_{2.i}.V$ behaves the same as  $\hybrid_{2.i-1}.V$ if $i > 1$ or as $V$ if $i=1$. 
\par From the (computational) quantum zero-knowledge property of $\Pi_{\zk}$, the probability that $\hybrid_{2.i}.V$ accepts is negligibly close to $\varepsilon$.\\ 

\noindent \underline{$\hybrid_{3.i}$, for $i \in [q(\secparam)]$}: We consider a hybrid verifier $\hybrid_{3.i}.V$ that computes the $i^{th}$ commitment $\bfc_i$ as follows: $\bfc_i \leftarrow \comm(1^{\secparam},0)$. Except this change, the hybrid verifier $\hybrid_{3.i}.V$ behaves the same as $\hybrid_{3.i-1}.V$ if $i> 1$ or as $\hybrid_{2.q(\secparam)}.V$ if $i=1$.  
\par From the quantum-concealing property of $\comm$, the probability that $\hybrid_{3.i}.V$ accepts is negligibly close to $\varepsilon$. \\

\noindent \underline{$\hybrid_4$}: We consider a hybrid verifier $\hybrid_4.V$, which is essentially the same as $\hybrid_{3.q(\secparam)}.V$, except that it generates $\crs \xleftarrow{\$} \{0,1\}^{q(\secparam)}$ and sends $\crs$ to $P$.
\par Since the hybrids $\hybrid_{3.q(\secparam)}.V$ and $\hybrid_{4}.V$ are identical, the probability that $\hybrid_{4}.V$ accepts is negligibly close to $\varepsilon$.\\

\noindent From the computational soundness of $\Pi_{\nizk}$, the probability that $\hybrid_4.V$ accepts is negligible. Thus, $\varepsilon$ is negligible.  

\newcommand{\bfR}{\mathbf{R}}
\newcommand{\bfC}{\mathbf{C}}
\newcommand{\bfB}{\mathbf{B}}
\newcommand{\bfT}{\mathbf{IZ}}
\newcommand{\bfNZ}{\mathbf{NZ}}
\newcommand{\bfcrs}{\mathbf{C}}
\newcommand{\simcrsgen}{\mathsf{SimCRSGen}}
\newcommand{\match}{\mathbf{Match}}
\paragraph{Statistical Zero-Knowledge.} Let $V^*$ be a computationally unbounded verifier and let $\ket{\Psi}$ be the initial state of $V^*$. Before we describe the simulator we first define the following registers. For $i=1,\ldots,q(\secparam)$:
\begin{itemize}
    \item $\bfB_i$: it contains the bit sent by the simulator in the $i^{th}$ iteration.
    \item $\bfR_i$: it contains the receiver's commitment and the $i^{th}$ bit of $\crs$ sent during the $i^{th}$ iteration. 
    \item $\bfT_i$: it contains the messages of zero-knowledge exchanged during the $i^{th}$ iteration.
    \item $\decreg$: it contains the decision bit. 
    \item $\auxreg$: it contains the auxiliary state of the verifier. 
    \item $\bfNZ$: it contains the final NIZK proof sent by the simulator. 
    \item $\bfcrs$: it contains the common reference string. 
    \item $\bfX$: this is a $\poly(\secparam)$-qubit ancillary register. 
\end{itemize}

\noindent {\em Description of Simulator}:
\begin{itemize}
    \item The simulator prepares the following state: 
    $$\ket{\Psi_1} = \left( \overset{q(\secparam)}{\underset{i=1}{\otimes}} \ket{0}_{\bfR_i} \ket{0}_{\bfB_i} \ket{0}_{\bfT_i} \right) \otimes \ket{0}_{\bfNZ} \ket{0}_{\bfX} \ket{0}_{\bfcrs} \ket{\Psi}_{\auxreg} \ket{0}_{\dec} $$
    \item It runs the NIZK simulator, $(\widehat{\crs}, \widehat{\pi})\leftarrow \Pi_{\nizk}.\Simu(1^\secparam,x)$, to compute $\widehat{\crs}$. It stores $\widehat{\crs}$ in the register $\bfcrs$, and it stores $\widehat{\pi}$ in $\bfNZ$.
    \item For all $i=1,\ldots,q(\secparam)$, let $U_i$ be the unitary that performs the following operations in superposition. 
    \begin{itemize}
        \item It first applies $V^*$ on the registers $\{\bfB_j,\bfR_j,\bfT_j,\auxreg\}_{j < i},\bfR_{i}$. 
        \item It then maps $\ket{0}_{\bfB_i}$ to $\ket{+}_{\bfB_i}$. 
        \item It then applies $V^*$ on the registers $\{\bfB_j,\bfR_j,\bfT_j,\auxreg\}_{j < i},\bfR_i,\bfB_i$. 
        \item It then performs the $i^{th}$ iteration of $\Pi_{\zk}$ with $V^*$ in superposition. The transcript is stored in $\bfT_i$. 
        \item It then applies the following unitary $\widehat{U}$: 
        $$\widehat{U} \ket{b_i}_{\bfB_i} \ket{\bfc_i\ crs_i}_{\bfR_i} \ket{\tau_i}_{\bfT_i}\ket{\widehat{\crs}}_{\bfcrs}\ket{0}_{\decreg} = \left\{ \begin{array}{ll} \ket{b_i}_{\bfB_i} \ket{\bfc_i\ crs_i}_{\bfR_i} \ket{\tau_i}_{\bfT_i}\ket{\widehat{\crs}}_{\bfcrs}\ket{1 \oplus \theta_i}_{\decreg}, & \text{if }\tau_i\text{ is valid},  \\ \ket{b_i}_{\bfB_i} \ket{\bfc_i\ crs_i}_{\bfR_i} \ket{\tau_i}_{\bfT_i} \ket{\widehat{\crs}}_{\bfcrs} \ket{+}_{\decreg}, & \text{othewise}  \end{array} \right.  $$
    \end{itemize}
    We define $\theta_i=1$ if the $i^{th}$ bit of $\widehat{\crs}$ is the same as $crs_i$, where $crs_i$ is the $i^{th}$ bit of $\crs$ computed by $V^*$ in the register $\bfR_i$. Furthermore, we define $\tau_i$ to be valid if the verifier in the $i^{th}$ execution of $\Pi_{\zk}$ accepts. 
    
    \item Let $W_{i} = \amplifier \left( U_i \right)$; where $\amplifier$ is the circuit guaranteed by Lemma~\ref{lem:watrous}. Simulator computes $\ket{\Psi_{i}} = W_i \ket{\Psi_{i-1}}$.
    
    \item Finally, it uses $\widehat{\pi}$ stored in $\bfNZ$ as the proof for the NIZK step. 
    
    \item Measure all the registers except for $\auxreg$.
    
\end{itemize}

\noindent We now prove the statistical indistinguishability of the real and the ideal world using a hybrid argument. Consider the following hybrids.\\

\noindent \underline{$\hybrid_1$}: This corresponds to the real execution between $P$ and $V^*$.\\ 

\noindent \underline{$\hybrid_{2.i}$ for $i \in [q(\secparam)]$}: We define a hybrid prover as follows: sample $\crs \leftarrow \mathsf{Gen}(1^{\secparam})$. Note that $\crs$ is generated according to the uniform distribution. Prepare the state $\ket{\Psi_1}$ as given in the description of the simulator. Apply $W_{i} \cdots W_1 \ket{\Psi_1}$ to obtain $\ket{\Psi_i}$. That is, perform Watrous rewinding for the first $i$ iterations of the OT protocol, similarly to the way the simulator does, but using $\crs$,  instead of $\widehat{\crs}$. Then, the hybrid prover uses the real prover to interact with $V^*$, that receives as input  $\ket{\Psi_i}$, to perform the operations for the rest of the protocol. 
\par We now show that the output distributions of the hybrids $\hybrid_{2.i-1}$ and $\hybrid_{2.i}$ are computationally indistinguishable. In order to show this, we use a similar argument that was used in the proof of Claim~\ref{clm:quanconc}. It suffices to argue that the following distributions are statistically close: 
\begin{itemize}
    \item $\distr_1$: Measure the registers $\{\bfR_i,\bfT_i\}_{i \in [q(\secparam)]},\bfNZ$ after the $i^{th}$ iteration in $\hybrid_{2.i-1}$ and output the measurement outcome along with the residual state in $\auxreg$. 
    \item $\distr_2$: Measure the registers  $\{\bfR_i,\bfT_i\}_{i \in [q(\secparam)]},\bfNZ$ after the $i^{th}$ iteration in $\hybrid_{2.i}$ and output the measurement outcome along with the residual state in $\auxreg$.
\end{itemize}
We prove this in two steps: first, we apply Watrous rewinding and analyze the state obtained after the $i^{th}$ iteration in $\hybrid_{2.i}$. In the next step, we use this to argue the indistinguishability of $\distr_1$ and $\distr_2$. 

\begin{comment}
In order to show this, it suffices to show the following:
$$\mathsf{Tr}_{\overline{\left\{ \{ \bfR_j,\bfT_j \}_{j \leq i},\auxreg   \right\}}}\left[ W_i \cdots W_1 \ket{\Psi_1}\bra{\Psi_1} W_1^{\dagger} \cdots W_i^{\dagger} \right] \approx_c \mathsf{Tr}_{\overline{\left\{ \{ \bfR_j,\bfT_j \}_{j \leq i},\auxreg    \right\}}} \left[ U_i W_{i-1} \cdots W_1 \ket{\Psi_1}\bra{\Psi_1} W_1^{\dagger} \cdots W_{i-1}^{\dagger} U_i \right]$$
$\mathsf{Tr}_\overline{X}$ indicates the operation of tracing out all the registers except $X$. The reason why the above statement suffices is because $\mathsf{Tr}_{\overline{\left\{ \{ \bfR_j,\bfT_j \}_{j \leq i},\auxreg   \right\}}}\left[ W_i \cdots W_1 \ket{\Psi_1}\bra{\Psi_1} W_1^{\dagger} \cdots W_i^{\dagger} \right]$ indicates the state of the verifier after the $i^{th}$ iteration in  $\hybrid_{2.i}$ and $\mathsf{Tr}_{\left\{ \{ \bfR_j,\bfT_j \}_{j \leq i},\auxreg    \right\}} \left[ U_i W_{i-1} \cdots W_1 \ket{\Psi_1}\bra{\Psi_1} W_1^{\dagger} \cdots W_{i-1}^{\dagger} U_i \right]$ indicates the state of the verifier after the $(i-1)^{th}$ iteration in $\hybrid_{2.{i-1}}$. 
\end{comment}

\paragraph{Applying Watrous Rewinding.}  Suppose $\ket{\Psi_{i-1}}= W_{i-1} \cdots W_1 \ket{\Psi_1}$. Consider the following: 
\begin{eqnarray*}
U_i \ket{\Psi_{i-1}} & = & \sqrt{q}\left( \sqrt{p} \ket{\phi_{{\sf good}}} \ket{0}_{\decreg} + \sqrt{1-p} \ket{\phi_{{\sf bad}}} \ket{1}_{\decreg} \right) + \sqrt{1 - q} \ket{\phi_{{\sf invalid}}} \ket{+}_{\decreg},\\
& = & \sqrt{p}\left( \sqrt{q} \ket{\phi_{{\sf good}}} + \sqrt{1-q} \ket{\phi_{{\sf invalid}}} \right) \ket{0}_{\decreg} + \sqrt{1-p} \left( \sqrt{q} \ket{\phi_{{\sf bad}}} + \sqrt{1-q} \ket{\phi_{{\sf invalid}}} \right) \ket{1}_{\decreg},
\end{eqnarray*}
where: 
\begin{itemize}
    \item $q$ is the probability with which $V^*$ convinces $P$ in the $i^{th}$ execution of $\Pi_{\zk}$,
    \item $\left| p-\frac{1}{2} \right| \leq \negl(\secparam)$: this follows from a similar argument as in the proof of Claim~\ref{clm:quanconc:inter},
    \item $\ket{\phi_{{\sf invalid}}}$ (defined on all the registers except the $\decreg$ register) is a superposition over the messages containing the $i^{th}$ iteration $\Pi_{\zk}$ transcripts that are not accepted by the verifier of $\Pi_{\zk}$,
    \item $\ket{\phi_{{\sf good}}}$ (defined on all the registers except the $\decreg$ register) is a superposition over the messages of the $i^{th}$ iteration containing $crs_i=\crs_i$  and, 
    \item $\ket{\phi_{{\sf bad}}}$ (defined on all the registers except the $\decreg$ register) is a superposition over the messages of the $i^{th}$ iteration containing $crs_i \neq \crs_i$. 
\end{itemize}
\noindent Once we apply~\Cref{lem:watrous}, the resulting state will be $W_i \ket{\Psi_{i-1}}=\left( \sqrt{q} \ket{\phi_{{\sf good}}} + \sqrt{1-q} \ket{\phi_{{\sf invalid}}} \right) \ket{0}_{\decreg}$ with probability negligibly close to 1.

\paragraph{Indistinguishability of $\distr_1$ and $\distr_2$.} As in the proof of Claim~\ref{clm:quanconc}, it suffices to argue that the distribution of measurements of $\{\bfR_i,\bfT_i\}_{i \in [q(\secparam)]},\bfNZ$ in $\ket{\phi_{{\sf good}}}$ along with the residual state in $\auxreg$ is computationally indistinguishable from the distribution of measurements of $\{\bfR_i,\bfT_i\}_{i \in [q(\secparam)]},\bfNZ$ in $\ket{\phi_{{\sf bad}}}$ along with the residual state in $\auxreg$. This follows from the perfect binding property of $\comm$ and the statistical soundness property of $\Pi_{\zk}$ using a similar argument used in Claim~\ref{clm:quanconc:inter}: if the verifier is not computed in superposition then the verifier cannot distinguish whether $crs_i = \crs_i$ or whether  $crs_i \neq \crs_i$. Moreover, this is true even if the verifier is computed in superposition and measuring the transcript registers in the end.   
\\

\noindent \underline{$\hybrid_3$}: Execute the simulator on input the state $\ket{\Psi}$.
\par From the statistical zero-knowledge property of $\Pi_{\nizk}$, it follows that the state output by $V^*$ in the hybrid $\hybrid_{2.q(\secparam)}$ is close in trace distance to the state output by $V^*$ in the hybrid $\hybrid_{3}$. 

%% file: QPoK/statsenderot.tex
\subsubsection{Post-Quantum Statistical Sender-Private OT}
\noindent The tool we use in this construction is a two-round oblivious transfer protocol that has computational security against senders and statistical security against receivers. We define this tool below. We instantiate this primitive with the QLWE-based construction in~\cite{BD18}.

\begin{definition}[Post-Quantum Statistical Sender-Private OT]
A two-round oblivious transfer is a tuple of algorithms $(\ot_1,\ot_2,\ot_3)$ which specifies the following protocol. \\

\noindent {\bf Round 1.} The receiver $\receiver$, on input security parameter $\secparam$, bit $\beta$, computes $(\smot_1,\state_R) \leftarrow \ot_1(1^{\secparam},\beta)$ and sends $\smot_1$ to the sender $\sender$.\\

\noindent {\bf Round 2.} The sender $\sender$, on input $\smot_1$ and message bits $(m_0,m_1)$, computes $\smot_2 \leftarrow \ot_2(1^{\secparam},\ot_1,(m_0,\allowbreak m_1))$. It sends $\smot_2$ to the receiver $R$. \\

\noindent {\bf Reconstruction.} The receiver computes $m' \leftarrow  \ot_3(1^{\secparam},\smot_1,\smot_2,\state_R)$. \\

\noindent {\bf Correctness.} For any $\beta \in \{0,1\}, (m_0,m_1) \in \{0,1\}^2$, we have: 
$$\prob \left[ \substack{ (\smot_1,\state_R) \leftarrow \ot_1(1^{\secparam},\beta)\\ \smot_2 \leftarrow \ot_2(1^{\secparam},\smot_1,(m_0,m_1))\\ m' \leftarrow \ot_3(1^{\secparam},\smot_1,\smot_2,\state_R)}\ :\ m' = m_{\beta} \right] = 1$$

\noindent {\bf Post-Quantum Receiver-Privacy.} The following holds: 
$$\{\ot_1(1^{\secparam},0)\} \approx_{c,Q} \{\ot_1(1^{\secparam},1)\}$$

\noindent {\bf Statistical Sender-Privacy.} There exists a computationally unbounded extractor such that for every the first round message $\smot_1$, it outputs a bit $b \in \{0,1\}$ such that the following holds for every $(m_0,m_1) \in \{0,1\}$:
$${\sf SD}(\ot_2(1^{\secparam},\smot_1,(m_0,m_1)),\ot_2(1^{\secparam},\smot_1,(m_b,m_b))) \leq \negl(\secparam),$$
where ${\sf SD}$ denotes statistical distance and $\negl$ is a negligible fucntion. 
\end{definition}

%% file: QPoK/pok.tex
%%%%%%%%%%%%%%%%%%%%%%%%%%%%%%%%%%%%%%%%%%%%%%%%%%%%%%%%%%%%%%%%%%%%%%%%%%%%%%%%%%%%%%%%%%%%%%%%%%%%%%%%%%%%%%%%%%%%%%%%%

\fullversion{
\section{Quantum Proofs of Knowledge for Bounded Concurrent QZK}
\label{sec:pok}

In this section, we construct a bounded concurrent QZK satisfying quantum proof of knowledge property assuming post-quantum statistical receiver-private oblivious transfer. 
\par We first start with a simpler case: we present a construction of quantum proof of knowledge for a standalone quantum ZK proof system for NP. We then show how to extend this construction to the bounded concurrent QZK setting.
}
\submversion{
\section{Bounded Concurrent QZKPoK: Proofs}
\label{sec:pok}
}
%\par We base our QZKPoK on a variant of QLWE, called QLWE with reduction-friendly cloning security. In the next subsection we introduce this assumption. 

\fullversion{

\input{QPoK/qpok-cons}
}

\submversion{
\ \\
We prove the completeness, quantum proof of knowledge and QZK properties of the construction presented in Section~\ref{sec:qzkpok}.
\ \\
}

\paragraph{Completeness.} The completeness follows by the completeness of $\Pi_{\zk}$. 

\paragraph{Quantum Proof of Knowledge.} Let $P^*$ be a malicious prover, that on input $(x,\rho)$, can convince $V$ to accept $x$ with non-negligible probability $\varepsilon$. Before we construct a QPT extractor $\extractor$, we first give a description of the registers used in the system. 

\begin{itemize}
    \item $\randreg_{i,j}$ for $i \in [\ell_w],j \in [\secparam]$: this consists of the receiver randomness used by the extractor in the $(i,j)^{th}$ executions of $\Pi_{\ot}$.  
    \item $\bitreg_{i,j}$, for $i \in [\ell_w],j \in [\secparam]$: this is a single-qubit register that contains a bit that is used by the extractor in the $(i,j)^{th}$ execution of the OT protocol.
    \item $\decreg$: it contains the decision register that indicates whether to rewind or not.
    \item $\auxreg$: this is initialized with the auxiliary state of the prover.
    \item $\transreg_{i,j}$, for $i \in [\ell_w],j \in [\secparam]$: it contains the transcripts of the $(i,j)^{th}$ executions of the OT protocol.
    \item $\transreg^*$: it contains the transcript of the protocol $\Pi_{\zk}$.
    \item $\regX$: this is a $\poly(\secparam)$-qubit ancillary register. 
\end{itemize}

\noindent \underline{Description of $\extractor(x,\ket{\Psi})$}: The state of the extractor is initialized as follows: 

    $$\left( \overset{}{\underset{i \in [\ell_w],\\ j \in [\secparam]}{\bigotimes}} \ket{0}_{\bitreg_{i,j}} \ket{0}_{\randreg_{i,j}}\ket{0}_{\transreg_{i,j}} \right) \otimes \ket{0}_{\transreg^*} \otimes  \ket{\Psi}_{\auxreg} \otimes \ket{0}_{\decreg} \otimes \ket{0^{\otimes \poly(\secparam)}}_{\regX} $$

\begin{itemize}
    \item For all $i \in [\ell_w],j \in [\secparam]$, perform the following operations in superposition: 
    \begin{itemize}
        \item Let $\ket{\widetilde{\Psi}}$ be the state of the system at the beginning of the $(i,j)^{th}$ execution. 
        \item Prepare the following state\footnote{We will assume, without loss of generality, that the length  of the random strings used in the OT protocol be $\secparam$.}: 
        $$ \ket{0}_{\bitreg_{i,j}}\ket{0}_{\randreg_{i,j}} \rightarrow \frac{1}{\sqrt{2^{\secparam+1}}} \sum_{\beta_{i,j} \in \{0,1\},s_{i,j}^{\ot} \in \{0,1\}^{\secparam}} \ket{\beta_{i,j}}_{\bitreg_{i,j}} \ket{s_{i,j}^{\ot}}_{\randreg_{i,j}} $$
        
        \item It then performs the $(i,j)^{th}$ execution of $\Pi_{\ot}$ along with the $P^*$'s message immediately after the $(i,j)^{th}$ execution of $\Pi_{\ot}$ in superposition. The resulting transcript is stored in the register $\transreg_{i,j}$. We denote the unitary that performs this step to be $U_{i,j}^{(1)}$.

        \item After $P^*$ sends the message immediately after the $(i,j)^{th}$ execution of $\Pi_{\ot}$, apply the  unitary $U_{i,j}^{(2)}$ defined as follows: 
        \begin{eqnarray*}
        & & U_{i,j}^{(2)} \ket{\beta_{i,j}}_{\bitreg_{i,j}} \ket{s_{\ot}^{(i,j)}}_{\randreg_{i,j}} \ket{\tau_{\ot}^{(i,j)},b_{i,j}}_{\transreg_{i,j}} \ket{0}_{\decreg}\\
        & = & \left\{ \begin{array}{ll} \ket{\beta_{i,j}}_{\bitreg_{i,j}} \ket{s_{\ot}^{(i,j)}}_{\randreg_{i,j}}\ket{\tau_{\ot}^{(i,j)},b_{i,j}}_{\transreg_{i,j}} \ket{{\sf Match}_{i,j}}_{\decreg}\ \text{if }{\sf acc}_{i,j}=1, \\ \ \\ \ket{\beta_{i,j}}_{\bitreg_{i,j}} \ket{s_{\ot}^{(i,j)}}_{\randreg_{i,j}}\ket{\tau_{\ot}^{(i,j)}, b_{i,j}}_{\transreg_{i,j}} \ket{+}_{\decreg},\ \text{ otherwise}\end{array} \right.
        \end{eqnarray*}
        Here, ${\sf Match}_{i,j}=0$ if and only if $\beta_{i,j}=b_{i,j}$, where $b_{i,j}$ is the bit sent by $P^*$ after the $(i,j)^{th}$ execution of the OT protocol. Moreover, ${\sf acc}_{i,j}=1$ only if $P^*$ has not aborted in $(i,j)^{th}$ OT execution. 
        \par Let $W_{i,j} = \amplifier\left(U_{i,j}^{(2)} U_{i,j}^{(1)} \right)$, where $\amplifier$ is defined in Lemma~\ref{lem:watrous}. Perform $W_{i,j} \ket{\widetilde{\Psi}}$ to obtain $\ket{\Psi_{i,j}}$.
    \end{itemize}
    \item Perform the execution of $\Pi_{\zk}$ in superposition. 
    \item Measure all the registers except $\auxreg$. Perform the OT reconstruction on input the measured transcript $\tau_{i,j}^{\ot}$, measured randomness $s_{i,j}^{\ot}$ and receiver's bit $b_{i,j}$. Call the resulting reconstruction output to be $\widetilde{sh_{i,j}}$. Let $\widetilde{w}_i = \oplus_{j=1}^{\ell_w} \widetilde{sh_{i,j}}$. Let $w$ be the concatenation of the bits $\widetilde{w}_1,\ldots,\widetilde{w}_{\ell_w}$. If $w$ is not a witness for $x$, abort. Otherwise output the state in $\auxreg$ along with $w$.  
\end{itemize}

\noindent We now argue that our protocol satisfies the proof of knowledge property. We assume that there is some total ordering defined on $(i,j)$, for $i \in [\ell_w]$ and $j \in [\secparam]$. Without loss of generality, we assume that $(1,1)$ is the least element in this total ordering. \\

\noindent \underline{$\hybrid_1$}: In this hybrid, $P^*$ interacts with the honest verifier $V$. The verifier $V$ accepts the proof with probability $\varepsilon$. \\

\noindent \underline{$\hybrid_{2.(i,j)}$, for $i \in [\ell_w],j \in [\secparam]$}: We define a hybrid verifier $\hybrid.V_{i,j}$ as follows. Let $\ket{\Phi}$ be the initial state of the system. Compute $\ket{\Psi_{i,j}} = \underset{(i',j') \leq (i,j)}{\prod} W_{i',j'}  \ket{\Phi}$. From here on, the rest of the iterations of $\Pi_{\ot}$ are computed by interacting with $P^*$ interacting honestly as specified in the real execution. The protocol $\Pi_{\zk}$ is computed by interacting with $P^*$ honestly as in the real execution. Finally, $\hybrid.V_{i,j}$  outputs its decision.
\par The probability that $\hybrid.V_{i,j}$ accepts is negligibly close to $\varepsilon$. Moreover, from the statistical security against senders, it follows that the state output by $P^*$ in this hybrid is close in trace distance to the state output by $P^*$ in the previous hybrid. We omit the proof since it essentially follows the same line of argument used in~\Cref{sec:quanconc}.    \\

\noindent \underline{$\hybrid_3$}: We define a hybrid verifier $\hybrid.V_3$ as follows.  Let $\ket{\Phi}$ be the initial state of the system. Compute $\ket{\Psi_{i,j}} = \underset{(i \in [\ell_w],j \in [\secparam])}{\prod} W_{i,j}  \ket{\Phi}$. The protocol  $\Pi_{\zk}$ is computed by interacting with $P^*$ honestly as in the real execution. Finally, $\hybrid.V_{3}$  outputs its decision. 
\par The probability that $\hybrid.V_{3}$ accepts is negligibly close to $\varepsilon$. This follows from the fact that $\hybrid.V_3$ is identical to $\hybrid.V_{i^*,j^*}$, where $(i^*,j^*)$ is the highest element in the total ordering. 
\par Moreover, the state output by $P^*$ in this hybrid is the same as the state output by $P^*$ in the previous hybrid. \\

\noindent \underline{$\hybrid_4$}: Define a hybrid verifier $\hybrid.V_4$ as follows: it executes the hybrid verifier $\hybrid.V_3$ until the step just before it outputs its decision. Let $\widetilde{sh_{i,j}}$ be the share output by the reconstruction algorithm of the receiver of $\Pi_{\ot}$. Let $\widetilde{w}_i = \oplus_{j=1}^{\ell_w} \widetilde{sh_{i,j}}$. Let $w$ be the concatenation of the bits $\widetilde{w}_1,\ldots,\widetilde{w}_{\ell_w}$. If $w$ is not a witness for $x$, abort. Otherwise, output the decision of $\hybrid.V_3$. 
\par The probability that $\hybrid.V_4$ accepts is negligibly close to $\varepsilon$. To see this, note that it is sufficient to argue that $|p_{3}-p_4| \leq \negl(\secparam)$, where $p_3$ is the probability with which $\hybrid.V_3$ aborts and $p_4$ is the probability with which $\hybrid.V_4$ aborts. This fact follows from the soundness of $\Pi_{\zk}$.  Moreover, the output state of the prover in $\hybrid_3$ is the same as the output state of the prover in $\hybrid_4$. \\

\noindent Note that the probability that the extractor $\extractor$ outputs a valid witness $w$ is the same as the probability that the hybrid verifier $\hybrid.V_4$ accepts. Moreover, the state output by $P^*$ when interacting with $\extractor$ is exactly the same as the state output by $P^*$ in $\hybrid_4$.  

\subsubsection{(Standalone) Quantum Zero-Knowledge} 
\label{subsec:qzk}
\noindent Suppose $(x,w) \in \rel(\lang)$. Suppose $V^*$ is a QPT verifier, that on input $(x,\ket{\Psi})$, interacts with the honest prover $P(x,w)$. We construct a simulator $\simr$ that takes as input $(x,\ket{\Psi})$ such that the output distribution of the simulator is computationally indistinguishable from the output distribution of $V^*$. \\
\ \\
\noindent \underline{Description of $\simr(x,\ket{\Psi})$}: 
\begin{itemize}
    \item For every $i \in [\ell_w]$, $\simr$ samples $sh_{i,1},\ldots,sh_{i,\secparam}$ uniformly at random. 
   % \item For every $i \in [\ell_w]$, $\simr$ samples $\alpha_{i,1},\ldots,\alpha_{i,\secparam}$ uniformly at random.
    \item For $i \in [\ell_w],j \in [\secparam]$, do the following: 
    \begin{itemize}
        \item $\simr$ and $V^*$ execute $\Pi_{\ot}$. The verifier $V^*$ takes the role of the receiver of $\Pi_{\ot}$ and $\simr$ takes the role of the sender. The input of the sender is $(sh_{i,j},sh_{i,j})$. 
        %\item Let the state of the verifier, at this point of this protocol, be $\ket{\Psi_{i,j}}$. Let $\simr_{\zk}$ be the $\Pi_{\zk}$ simulator associated with the $\Pi_{\zk}^{(1)}$ verifier $V^*_{i,j}$, where $V^*_{i,j}$ be the code used by $V^*$ in the $(i,j)^{th}$ execution of the protocol $\Pi_{\zk}$. Compute $\simr_{\zk}^{(1)}$ on input the state $\ket{\Psi}_{i,j}$ and the instance $\left(\tau_{\ot}^{(i,j)},i,j \right)$.   
        \item  $\simr$ samples a random bit $b_{i,j}$ and sends to $V^*$.
    \end{itemize}
    \item Let the state of the verifier, at this point of this protocol, be $\widetilde{\ket{\Psi}}$. Let $\simr_{\zk}$ be the $\Pi_{\zk}$ simulator associated with the $\Pi_{\zk}$ verifier $\widetilde{V^*}$, where $\widetilde{V^*}$ is the code used by $V^*$ in the execution of the protocol $\Pi_{\zk}$. Compute $\simr_{\zk}$ on input the state $\widetilde{\ket{\Psi}}$ and the instance $\left(x,\left\{ \tau_{\ot}^{(i,j)}, b_{i,j} \right\} \right)$. 
    \item Output the transcript of the protocol along with the private state of the verifier $V^*$. 
\end{itemize}

\noindent We now prove that the state output by $V^*$ when interacting with the honest prover $P(x,w)$ is computationally indistinguishable from the state output by $\simr(x,\ket{\Psi})$. Consider the following hybrids. As before we consider a total ordering on $(i,j)$, for $i \in [\ell_w],j \in [\secparam]$.  \\

\noindent \underline{$\hybrid_1$}: In this hybrid, $P$ and $V^*$ interact with each other. The output of this hybrid is the output of $V^*$. \\

%\noindent \underline{$\hybrid_{2.(i,j)}$ for $i \in [\ell_w],j \in [\secparam]$}: We define a hybrid prover $\hybrid_{2.(i,j)}.P$ that behaves as follows: it simulates the $(i',j')^{th}$ execution of $\Pi_{\zk}^{(1)}$, for all $(i',j') \leq (i,j)$, using the simulator $\simr_{\zk}^{(1)}$ as given in the description of $\simr$. The rest of the protocol is the same as in the previous hybrid. 
%\par The computational indistinguishability of this hybrid and the previous one follows from the zero-knowledge property of $\Pi_{\zk}^{(1)}$. \\

\noindent \underline{$\hybrid_2$}: We define another hybrid prover $\hybrid_2.P$ that behaves as follows: it simulates the protocol $\Pi_{\zk}$ using the simulator $\simr_{\zk}$ as given in the description of $\simr$. The rest of the protocol is the same as in the hybrid $\hybrid_{1}$. 
\par The computational indistinguishability of $\hybrid_1$ and $\hybrid_{2}$ follows from the quantum zero-knowledge property of $\Pi_{\zk}$.\\

\noindent \underline{$\hybrid_3$}: The output of this hybrid is the output of $\simr(x,\ket{\Psi})$.

\begin{claim}
Assuming the post-quantum sender-privacy of $\Pi_{\ot}$, the output of $\hybrid_2$ is computationally indistinguishable from the output of $\hybrid_3$.  
\end{claim}
\begin{proof}
Let $\adversary$ be the distinguisher distinguishing $\hybrid_2$ and $\hybrid_3$. We are going to prove that $\adversary$ can only distinguish with negligible probability. Consider the intermediate hybrids. \\

\noindent \underline{$\hybrid_{2.1}$}: This is identical to $\hybrid_2$. \\

\noindent We now consider a series of hybrids that are defined with respect to $\adversary$. \\
 
\noindent \underline{$\hybrid_{2.2.(i^*,j^*)}^{\adversary}$ for $i^*\in[\ell_w], j^*\in [\secparam-1]$}: We say that a hybrid prover $\hybrid_{2.2.(i^*,j^*)}.P$, uses $\left( \left\{ \widehat{b}_{i,j} \right\}_{(i,j) \leq (i^*,j^*)} \right)$, if the following holds: it executes the prover as in $\hybrid_{2.1}$, except, for $(i,j) \leq (i^*,j^*)$, it uses the input $(sh_{i,j},sh_{i,j})$ if $\widehat{b}_{i,j}\neq b_{i,j}$ or uses the input $(\alpha_{i,j},\alpha_{i,j})$ if $\widehat{b}_{i,j}=b_{i,j}$. 
\par Now, execute the above hybrid prover  $\hybrid_{2.2.(i^*,j^*)}.P$ by adaptively choosing $\left( \left\{ \widehat{b}_{i,j} \right\}_{(i,j) \leq (i^*,j^*)} \right)$ (as a function of the current state of the verifier and $\adversary$) such that the output distributions of $\hybrid_{2.2.(i^*,j^*)}.P$ and $\hybrid_{2.2.(i^*,j^*)-1}.P$ cannot be distinguished by $\adversary$. If such a set of bits cannot be adaptively chosen then abort. Otherwise, this hybrid prover interacts with the verifier and the output of this hybrid is set to be the output of the verifier.  

\begin{claim}
\label{clm:int:advspecific1}
The hybrid $\hybrid_{2.2.(i^*,j^*)}^{\adversary}$ aborts with negligible probability. 
\end{claim}
\begin{proof}
We prove this by induction. 

\paragraph{Base Case: $(i^*,j^*)=(1,1)$.} We prove that $\hybrid_{2.2.(1.1)}^{\adversary}$ aborts with negligible probability. From the post-quantum sender privacy property of $\Pi_{\ot}$ (Definition~\ref{def:ot:definition1}), it follows that upon fixing the view of the verifier until the last message of execution of $(1,1)^{th}$ OT protocol, there exists a bit $\widehat{b}$, with probability negligibly close to 1, such that the adversary cannot win the Game $G_{\widehat{b}}\left( m_0^{(1,1)},m_1^{(1,1)} \right)$ (specified in Definition~\ref{def:ot:definition1}) where, $\left(m_0^{(1,1)},m_1^{(1,1)} \right)=(sh_{1,1},\alpha_{1,1})$ is $b_{1,1} = 0$ and  $\left(m_0^{(1,1)},m_1^{(1,1)} \right)= \left(\alpha_{1,1}, sh_{1,1} \right)$ is $b_{1,1} = 1$. 

\paragraph{Induction Hypothesis.} Suppose this statement is true for all $(i,j) < (i^*,j^*)$. We prove this statement to be true even for $(i^*,j^*)$ using proof by contradiction.
\par Suppose $\hybrid_{2.2.(i^*,j^*)}$ aborts with non-negligible probability then we design a QPT adversary ${\cal B}=({\cal B}_1,{\cal B}_2)$, that receives as input non-uniform quantum advice, and breaks the post-quantum sender privacy property of $\Pi_{\ot}$. 
\par We first define the non-uniform advice as follows: it computes the interaction between the hybrid prover $\hybrid_{2.2.(i^*,j^*)-1}.\prover$ and the verifier $V^*$, until the $((i^*,j^*)-1)^{th}$ execution of OT. It outputs the private state of $\hybrid_{2.2.(i^*,j^*)-1}.\prover$ and the private state of $V^*$. Call this state $\ket{\Psi_{adv}}$. \\

\noindent ${\cal B}_1$, upon receiving $\ket{\Psi_{adv}}$, takes the role of the receiver and interacts with the external challenger until the receiver's last message of the $(i^*,j^*)^{th}$ execution of $\Pi_{\ot}$. ${\cal B}_1$ uses the code of $V^*$ to interact with the external challenger. The external challenger on the other receives as input $m_0^{(i^*,j^*)}=sh_{i^*,j^*}$ and $m
_1^{(i^*,j^*)}=\alpha_{i^*,j^*}$ if $b_{i^*,j^*}=0$ or $m_1^{(i^*,j^*)}=sh_{i^*,j^*}$ and $m
_0^{(i^*,j^*)}=\alpha_{i^*,j^*}$ if $b_{i^*,j^*}=1$, from ${\cal B}_1$, where $sh_{i^*,j^*},\alpha_{i^*,j^*},b_{i^*,j^*}$ are generated as in $\hybrid_{2.1}$. It then outputs the state $\ket{\Psi_{1}}$ of $V^*$ obtained after the receiver sends the last message in the $(i^*,j^*)^{th}$ execution of $\Pi_{\ot}$.\\

\noindent ${\cal B}_2$, upon receiving $\ket{\Psi_{1}}$, computes the rest of the executions of $\Pi_{\ot}$ and $\Pi_{\zk}$ by emulating the interaction between the hybrid prover $\hybrid_{2.2.(i^*,j^*)}.P$ and the verifier $V^*$. It then inputs the final state of $V^*$ to $\adversary$. The output of ${\cal B}_2$ is set to be the output of $\adversary$. \\

\noindent Our initial assumption was that the $\hybrid_{2.2.(i^*,j^*)}$ aborts with non-negligible probability. This means that the adversary $\adversary$ can distinguish with non-negligible probability (over  the view of the verifier until the $(i^*,j^*)^{th}$ OT execution) both the games -- that is, distinguishing $(m_0^{(i^*, j^*)},m_1^{(i^*,j^*)})$ from $(m_1^{(i^*, j^*)},m_1^{(i^*,j^*)})$ (Game 0) as well as distinguishing $(m_0^{(i^*, j^*)},m_1^{(i^*,j^*)})$ from $(m_0^{(i^*, j^*)},m_0^{(i^*,j^*)})$ (Game 1)  -- with probability significantly greater than $\frac{1}{2}$. This in turn means that ${\cal B}$ can break the post-quantum sender privacy property of $\Pi_{\ot}$ with non-negligible probability. Thus, we arrived at a contradiction.    
\end{proof}

\noindent \underline{$\hybrid_{2.3}^{\adversary}$}: This hybrid is the same as $\hybrid_{2.2.(\ell_w,\secparam-1)}$, except that the hybrid prover will abort if for all $i\in[\ell_w]$ and all $j \in [\secparam-1]$, it holds that $b_{i,j} \neq \widehat{b_{i,j}}$.\\

\begin{claim}
The hybrids $\hybrid_{2.2.(\ell_w,\secparam-1)}^{\adversary}$ and $\hybrid_{2.3}^{\adversary}$ can be distinguished by $\adversary$ with only negligible probability.  
\end{claim}
\begin{proof}
 To prove this, we consider an alternate hybrid prover in $\hybrid_{2.2.(\ell_w,\secparam-1)}^{\adversary}$ which samples, for any $i$, $b_{i,j} \xleftarrow{\$} \{0,1\}$ at the end of first $(\secparam-1)$ iterations of $\Pi_{\ot}$. It then sets the input to the $\secparam^{th}$ iteration of $\Pi_{\ot}$ to be $\left(\oplus_{j=1}^{\secparam-1} m_{b_{i,j}}^{(i,j)},u \right)$ with probability $\frac{1}{2}$ or $\left(u,\oplus_{j=1}^{\secparam-1} m_{b_{i,j}}^{(i,j)} \right)$ with probability $\frac{1}{2}$, where $u \xleftarrow{\$} \{0,1\}$ and $\{m_{0}^{(i,j)},m_1^{(i,j)}\}_{j \in [\secparam-1]}$ are the inputs used in the first $\secparam-1$ iterations of $\Pi_{\ot}$. Note that the output distribution of  $\hybrid_{2.2.(\ell_w,\secparam-1)}^{\adversary}$ remains the same even with this change.  
 \par Since the $b_{i,j}$'s, for $j \leq \secparam-1$, are sampled after the $\widehat{b}_{i,j}$'s are decided, the probability that $\widehat{b}_{i,j}\neq b_{i,j}$ is $\frac{1}{2}$ for any $i \in [\ell_w],j \in [\secparam-1]$. Thus, the probability that $\left( \forall i \in [\ell_w],j \in [\secparam-1], b_{i,j} \neq \widehat{b}_{i,j} \right)$ is $\leq \frac{\ell_w}{2^{\secparam-1}}$. Conditioned on this bad event, the output distributions of $\hybrid_{2.2.(\ell_w,\secparam-1)}^{\adversary}$ and $\hybrid_{2.3}$ are identical. Thus, $\adversary$ cannot distinguish the hybrids $\hybrid_{2.2.(\ell_w,\secparam-1)}^{\adversary}$ and $\hybrid_{2.3}^{\adversary}$.\\
 \end{proof}

\noindent \underline{$\hybrid_{2.4.i^*}^{\adversary}$ for all $i \in [\ell_w]$}: This hybrid is the same as $\hybrid_{2.3}^{\adversary}$ except that the hybrid prover $\hybrid_{2.4.i^*}.P$ is additionally parameterized by $\left( \left\{ \widehat{b}_{i,\secparam} \right\}_{i \leq i^*} \right)$. The only change from the previous hybrid is that the hybrid prover, for $i \leq i^*$, use the input $(sh_{i,\secparam},sh_{i,\secparam})$ if $\widehat{b}_{i,\secparam} \neq b_{i,\secparam}$ or use $(\alpha_{i,\secparam},\alpha_{i,\secparam})$ if $\widehat{b}_{i,\secparam}=b_{i,\secparam}$. 
\par Now, consider a hybrid prover  $\hybrid_{2.4.i^*}.P$, parameterized by $\left( \left\{ \widehat{b}_{i,j} \right\}_{(i \leq i^*) \vee (j \leq \secparam-1)} \right)$, where $\left( \left\{ \widehat{b}_{i,j} \right\}_{(i,j) \leq (i^*,j^*)} \right)$, is defined to be such that the output distributions of $\hybrid_{2.i^*}.P$ and $\hybrid_{2.4.i^*-1}.P$ cannot be distinguished by $\adversary$. If such a hybrid prover does not exist, then abort. Otherwise, this hybrid prover interacts with the verifier and the output of this hybrid is set to be the output of the verifier. 

\begin{claim}
The hybrid $\hybrid_{2.4.i^*}$ aborts with negligible probability. 
\end{claim}

\noindent We omit the proof of the above claim since it uses the same inductive argument as the proof of Claim~\ref{clm:int:advspecific1}. \\

\noindent \underline{$\hybrid_{2.5}$}: This hybrid is the same as $\hybrid_3$, i.e. the output of the simulator. 
\par Conditioned on $\hybrid_{2.4.\ell_w}$ not aborting, the output distributions of  $\hybrid_{2.4.\ell_w}$ and $\hybrid_{2.5}$ are the same. This follows from the fact that if $\hybrid_{2.4.\ell_w}$ does not abort then the distribution of the inputs used in all the OT executions in the hybrids $\hybrid_{2.4.\ell_w}$ and $\hybrid_{2.5}$ are the same. Thus, $\adversary$ can distinguish $\hybrid_{2.4.\ell_w}$ and $\hybrid_{2.5}$ only with negligible probability. \\

\noindent From the above hybrids, it follows that $\adversary$ can distinguish the hybrids $\hybrid_{2.1}$ and $\hybrid_{2.5}$ with only negligible probability. 

\end{proof}

\subsection{Extending to Bounded Concurrent QZK Setting}
\label{sec:bcqzkpok}

We show how to adopt the construction in Section~\ref{sec:qzkpok} to the bounded concurrent setting. 

\subsubsection{Construction}

The construction of bounded concurrent quantum proof of knowledge system is the same as Figure~\ref{fig:cqzkpok}, except that we instantiate $\Pi_{\zk}$ with a modified version of the bounded concurrent QZK for NP construction in Section~\ref{sec:BCQZK}. 

\paragraph{Modified Bounded Concurrent QZK for NP Construction.} We modify the construction in Section~\ref{sec:BCQZK} as follows: Let $M$ be the round complexity of the statistical receiver private OT protocol and let $M = \secparam^c$ for some constant $c$, where $\secparam$ is the security parameter used in the OT protocol. Let $\secparam'$ denote a different security parameter used in $\Pi_{\zk}$ such that $\secparam' - M \geq \secparam$. We set the threshold of matched slots needed in the WI protocol from~\Cref{sec:BCQZK}, to instead be, $60Q^7 \secparam' + Q^4 \secparam' - M$, provided we set $\secparam' \gg M$. 
\par The completeness and soundness proofs of this modified construction are the same as the ones in Section~\ref{sec:BCQZK}. Even the quantum zero-knowledge property is the same as before. However, we will need the simulator to satisfy a stronger property defined next. 

\paragraph{Strong QZK Simulator.} We explain the differences between the strong QZK simulator and the simulator $\Simu$ defined in~\Cref{sec:BCQZK}. The strong simulator proceeds as follows:
\begin{enumerate}
    \item It simulates block-by-block similarly to $\Simu$, but instead of using $\ket{+}_{\decreg}$ only in the decision bit of the registers that aborted, it can choose to use $\ket{+}_{\decreg}$ on other transcripts as well. This decision is based on an efficiently computable function $f$. For example, on a transcript $t$, it can set $\decreg$ to $\ket{+}_{\decreg}$ conditioned on $f(t)=1$.
    \item After rewinding a block, it measures the transcript of that block instead of waiting until the end to measure. Furthermore, it keeps tracks of the total number of blocks on which the measurement outcomes correspond to a transcript in which it used $\ket{+}_{\decreg}$.
    \item If at any point, the number of block measurement outcomes that correspond to $\ket{+}_{\decreg}$ transcripts is greater than $M$, it aborts.
\end{enumerate}

Conditioned on the strong simulator not aborting in Step 3, its output is computationally indistinguishable from the output of $\Simu$.

%\noindent At the same time, in $\Pi_{zk}$, we replace the threshold of matched slots needed in the WI proofs from $60Q^7 \secparam' + Q^4 \secparam' - M$, provided we set $\secparam' \gg M$.  With this changes, it can be shown that for every session $\Simu'$ can rigged at least $3Q^4 \secparam' - M$ with high probability, while the number of slots it guesses by luck is still greater than $60Q^7 \secparam' - 2Q^4 \secparam'$. With these changes, then conditioned on $\Simu'$ not aborting in Step 3, its output is computationally indistinguishable from the output of $\Simu$.

%Meanwhile, by choosing $\secparam' \gg M$, the soundness still holds because a malicious prover will still not be able to guess enough slots to pass the threshold.

%\par We will set $M$ to be the round complexity of the statistical receiver private OT protocol, so we have $M = \secparam^c$ for some constant $c$. Instantiating the protocol in Figure~\ref{fig:cqzkpok} with $\Pi_{\zk}$, and using $\Simu'_M$ as the zero-knowledge simulator, we will show that the resulting protocol is a bounded concurrent quantum zero-knowledge proof of knowledge.\\

%\par The proof of completeness and quantum proof of knowledge remains the same. 

\paragraph{Arguing Bounded Concurrent Quantum Zero-Knowledge for Figure~\ref{fig:cqzkpok}.} In the proof of bounded concurrent quantum zero-knowledge, we now need to handle $Q$-session verifiers, where $Q$ is the number of sessions associated with the protocol.
\par The description of the simulator is the same as in the description of the simulator in~\Cref{subsec:qzk} except that we now execute the bounded concurrent strong simulator described above for $\Pi_{\zk}$ instead of the standalone ZK simulator.
\par We describe the hybrids below. Our description of hybrids and the proofs of indistinguishability between the hybrids closely follows the structure of the proof in~\Cref{subsec:qzk} and hence we only highlight the main differences.\\

\noindent \underline{$\hybrid_1$}: Same as $\hybrid_1$ in~\Cref{subsec:qzk}.\\ 

\noindent \underline{$\hybrid_2$}: We define another hybrid prover $\hybrid_2.P$ that behaves as follows: it simulates the protocol $\Pi_{\zk}$ using the bounded concurrent simulator $\simr_{\zk}$ as given in the description of $\simr$. The rest of the protocol is the same as in the hybrid $\hybrid_{1}$. 
\par The computational indistinguishability of $\hybrid_1$ and $\hybrid_{2}$ follows from the bounded concurrent quantum zero-knowledge property of $\Pi_{\zk}$.\\

\noindent \underline{$\hybrid_3$}: The output of this hybrid is the output of the simulator. 

\begin{claim}
Assuming the post-quantum sender-privacy of $\Pi_{\ot}$, the output of $\hybrid_2$ is computationally indistinguishable from the output of $\hybrid_3$.  
\end{claim}
\begin{proof}
Let $\adversary$ be the distinguisher distinguishing $\hybrid_2$ and $\hybrid_3$. We are going to prove that $\adversary$ can only distinguish with negligible probability. Consider the intermediate hybrids. \\

\noindent \underline{$\hybrid_{2.1}$}: This is identical to $\hybrid_2$. \\

%\noindent We now consider a series of hybrids that are defined with respect to $\adversary$. These hybrids are similar to the ones described in~\Cref{sec:qzkpok} except that now we have to be careful on what OT to change. In particular, the hybrid provers will change OTs in order of arrival.
%To explain further, for all $k\in[Q]$, we want to replace the $(i,j)$ OT execution for all $i \in [\ell_w]$ and $j \in [\secparam]$. There will be a total $Q\cdot \ell_w \cdot (\secparam)$ such OTs that we want to replace. Note that this is similar to how, in the proof of bounded concurrent ZK~\Cref{sec:BCQZK}, the hybrids replace the WI witnesses in order of arrival too, although for different reasons.\\

\noindent We now consider a sequence of hybrids. Each hybrid in this sequence is parameterized by the number of OT executions and the number of sessions. We first replace the inputs of all the OTs corresponding to one session before we move on to the next session. That is, each hybrid is of the form $\hybrid_{2.2.i.j.k}$. We first start with $i=1,j=1,k=1$. We then iterate over $j=1,\ldots,\secparam-1$, and then we increment $i$. We keep doing this, until we reach the hybrid $\hybrid_{2.2.\ell_w.\secparam-1.k}$. The next hybrid is $\hybrid_{2.3.k}$. After this, we have the hybrid,  $\hybrid_{2.4.i.k}$, where $i=1$ and $k=1$. We then iterate over $i=1,\ldots,\ell_w$. Immediately after the hybrid $\hybrid_{2.4.\ell.k}$, we have the hybrid $\hybrid_{2.5.k}$. At this point, all the OTs corresponding to the first session have been replaced. Immediately after this hybrid, we then move on to the hybrid $\hybrid_{2.2.i.j.k}$, where $i=1,j=1,k=2$. We then continue as above, until we reach the hybrid $\hybrid_{2.5.Q}$. The hybrid that follows $\hybrid_{2.5.Q}$ is the hybrid $\hybrid_{2.6}$. \\
 
\noindent \underline{$\hybrid_{2.2.i.j.k}^{\adversary}$ for $i \in[\ell_w], j \in [\secparam-1], k \in [Q]$}: We say that a prover, uses $\left( \left\{ \widehat{b}_{j} \right\}_{j\leq i} \right)$ in a particular transcript, if the following holds: in superposition, execute the prover as in $\hybrid_{2.1}$, except that in the first $j\leq i$ OT executions to end in the transcript being generated, it uses the input $(sh_{j},sh_{j})$ if $\widehat{b}_{j}\neq b_{j}$ or uses the input $(\alpha_{j},\alpha_{j})$ if $\widehat{b}_{j}=b_{j}$.

We define a hybrid prover $\hybrid_{2.2.i.j.k}.P$ as follows:
\begin{itemize}
    \item For $k'< k$, it chooses the input to the $(i,j)^{th}$ execution of the $k'$-session to be $(\alpha_{i,j},\alpha_{i,j})$, where $\alpha_{i,j}$ is sampled uniformly at random. 
    \item For $k' > k$, it chooses the inputs to the OT executions as done by the prover in $\hybrid_{2.1}$.  
    \item For $k' = k$, the hybrid prover, {\em in superposition}, adaptively uses $\left( \left\{ \widehat{b}_{(i',j')} \right\}_{(i',j') \leq (i,j)} \right)$ such that the output distributions of $\hybrid_{2.2.(i,j).k}.P$ and $\hybrid_{2.2.(i,j)-1.k}.P$ (if $(i,j)=(1,1)$ then the hybrid $\hybrid_{2.2.(i,j).k}.P$ is defined to be $\hybrid_{2.4.k-1}.P$) cannot be distinguished by $\adversary$. That is, since the whole protocol is being executed in superposition, as a function of each term in the superposition, the bits $\left( \left\{ \widehat{b}_{(i',j')} \right\}_{(i',j') \leq (i,j)} \right)$ are adaptively determined and stored in a separate register to be used by the hybrid prover. If the entire sequence of bits cannot be determined, then store $\bot$ in the same register. At the end of the protocol, we measure this register. If the outcome is $\bot$ then abort, otherwise, measure the registers storing the transcript, trace out all the registers except the register containing the auxiliary state of the verifier and output the measured transcript along with the residual auxiliary state of the verifier.
    \end{itemize}

%\par Now, execute the above hybrid prover $\hybrid_{2.2.i}.P$ {\em in superposition} adaptively uses $\left( \left\{ \widehat{b}_{j} \right\}_{j \leq i} \right)$  such that the output distributions of $\hybrid_{2.2.i}.P$ and $\hybrid_{2.2.i-1}.P$ (if $i=1$ then the hybrid $\hybrid_{2.2.1}.P$ is defined to be $\hybrid_{2.1}.P$) cannot be distinguished by $\adversary$. That is, since the whole protocol is being executed in superposition, as a function of each term in the superposition, the bits $\left( \left\{ \widehat{b}_{j} \right\}_{j \leq i} \right)$ are adaptively determined and stored in a separate register to be used by the hybrid prover. If the entire sequence of bits cannot be determined, then store $\bot$ in the same register. At the end of the protocol, we measure this register. If the outcome is $\bot$ then abort, otherwise, measure the registers storing the transcript, trace out all the registers except the register containing the auxiliary state of the verifier and output the measured transcript along with the residual auxiliary state of the verifier. 

\begin{claim}
\label{clm:int:advspecific}
The hybrid $\hybrid_{2.2.i.j.k}^{\adversary}$ aborts with negligible probability. 
\end{claim}
\begin{proof}
We prove this by induction. 

\paragraph{Base Case: $(i,j)=(1,1)$.} We prove that $\hybrid_{2.2.1.1.k}^{\adversary}$ aborts with negligible probability. Suppose not. We demonstrate a reduction that breaks the sender privacy property of OT with non-negligible probability. The goal of the reduction is to win in both the games with non-negligible probability: in the first game, it needs to distinguish the case when the challenger uses the input $(m_0,m_1)$, where $(m_0,m_1)=(sh_{1,1},\alpha_{1,1})$ with probability $\frac{1}{2}$ and $(m_0,m_1)=(\alpha_{1,1},sh_{1,1})$ or when it uses the input $(sh_{1,1},sh_{1,1})$. In the second game, it needs to distinguish the case when the challenger uses the input $(m_0,m_1)$, where $(m_0,m_1)$ is defined as above, versus the case when it uses the input $(\alpha_{1,1},\alpha_{1,1})$. 
\par We describe a reduction that does the following: just like the simulator of the bounded concurrent QZK, it divides the entire protocol transcript into blocks $B_1,\ldots,B_{L}$, where $L$ is as defined in $\Pi_{\zk}$. For every block $B_i$, it does the following:  it executes the simulator in superposition. If it encounters a message of $(1,1)^{th}$ OT, it stops simulating the rest of the block. It then puts $\ket{+}$ state in the decision register. Otherwise, it continues the simulation until the end of the block. It then performs Watrous rewinding. At the end, it measures the trascript. There are two cases:
\begin{itemize}
    \item If the block $B_i$ has completed it's execution and if no message in the $(1,1)^{th}$ OT execution has been encountered so far, then continue to the block $B_{i+1}$.  
    \item If a message in the $(1,1)^{th}$ OT execution has been encountered then forward to the challenger of the OT protocol. Use the response by the challenger to continue the execution of $B_i$, albeit by interacting $V^*$, rather running $V^*$ in superposition. Once this is completed, move on to the block $B_{i+1}$. 
\end{itemize}

\noindent Finally, after all the blocks are executed, the transcript along with the final private state of the verifier is input to $\adversary$. Note that there will be at most $M$ (the round complexity of the statistical receiver private OT protocol) blocks where the simulator's decision to rewind gets changed to $\ket{+}$ instead of using the rewinding decision that it was using previously. This is why we change the parameters of the zero-knowledge simulator to make sure that there are still enough blocks being appropriately rewound as needed for the simulation to execute correctly.
\par If the challenger uses the input $(m_0,m_1)$ then it corresponds to the hybrid $\hybrid_{2.4.k-1}$ (if $k=1$, then $\hybrid_{2.4.k-1}$ is the hybrid $\hybrid_{2.1}$) and if the challenger uses the input $(sh_{1,1},sh_{1,1})$ or the input $(\alpha_{1,1},\alpha_{1,1})$ then it corresponds to the hybrid $\hybrid_{2.2.1.1.k}$. 
\par If $\adversary$ can distinguish the hybrids $\hybrid_{2.2.1.1.k}$ and $\hybrid_{2.4.k-1}$  with non-negligible probability then the reduction can break the sender privacy property with non-negligible probability. 

\paragraph{Induction Hypothesis.} Suppose this statement is true for all $(i',j') < (i,j)$. We then show this to be true even for $(i,j)$. 
\par Suppose this is not true. We then design a reduction that violates the sender privacy of OT with non-negligible probability. The reduction essentially is defined along the same lines as the base case, except that the first $((i,j)-1)$ OT executions of the $k^{th}$ verifier are generated as non-uniform advice. That is, the advice generation algorithm executes the protocol in superposition and in each term of the superposition, it halts after the final $((i,j)-1)^{th}$ execution of the $k^{th}$. Until this point, the inputs to the $(i',j')^{th}$ OT execution, for $(i',j') < (i,j)$, is set to be either $(sh_{(i',j')},sh_{(i',j')})$ or $(\alpha_{(i',j')},\alpha_{(i',j')})$, depending on the distinguishing probability of $\adversary$. Finally, the advice generation measures the transcript and outputs the transcript along with the residual state. 
\par The reduction then performs block-by-block execution of the protocol and interacts with the challenger as in the base case. The final state of the verifier along with the transcript of the entire protocol is input to $\adversary$. 
\par As in the base case, if $\adversary$ can distinguish the two hybrids with non-negligible probability then the reduction can also violate the sender privacy property with non-negligible probability.   
 
\end{proof}

\noindent \underline{$\hybrid_{2.3.k}^{\adversary}$ for $k \in [Q]$}: This hybrid is the same as $\hybrid_{2.2.\ell_w,\secparam-1.k}$, except that the hybrid prover will abort if for all $i\in[\ell_w]$ and all $j \in [\secparam-1]$, it holds that $b_{i,j} \neq \widehat{b_{i,j}}$.\\

\begin{claim}
The hybrids $\hybrid_{2.2.\ell_w,\secparam-1.k}^{\adversary}$ and $\hybrid_{2.3.k}^{\adversary}$ can be distinguished by $\adversary$ with only negligible probability.  
\end{claim}
\begin{proof}
 To prove this, we consider an alternate hybrid prover in $\hybrid_{2.2.\ell_w,\secparam-1.k}^{\adversary}$ which samples, for any $i$, $b_{i,j} \xleftarrow{\$} \{0,1\}$ at the end of first $(\secparam-1)$ iterations of $\Pi_{\ot}$. It then sets the input to the $\secparam^{th}$ iteration of $\Pi_{\ot}$ to be $\left(\oplus_{j=1}^{\secparam-1} m_{b_{i,j}}^{(i,j)},u \right)$ with probability $\frac{1}{2}$ or $\left(u,\oplus_{j=1}^{\secparam-1} m_{b_{i,j}}^{(i,j)} \right)$ with probability $\frac{1}{2}$, where $u \xleftarrow{\$} \{0,1\}$ and $\{m_{0}^{(i,j)},m_1^{(i,j)}\}_{j \in [\secparam-1]}$ are the inputs used in the first $\secparam-1$ iterations of $\Pi_{\ot}$. Note that the output distribution of  $\hybrid_{2.2.(\ell_w,\secparam-1)}^{\adversary}$ remains the same even with this change.  
 \par Since the $b_{i,j}$'s, for $j \leq \secparam-1$, are sampled after the $\widehat{b}_{i,j}$'s are decided, the probability that $\widehat{b}_{i,j}\neq b_{i,j}$ is $\frac{1}{2}$ for any $i \in [\ell_w],j \in [\secparam-1]$. Thus, the probability that $\left( \forall i \in [\ell_w],j \in [\secparam-1], b_{i,j} \neq \widehat{b}_{i,j} \right)$ is $\leq \frac{\ell_w}{2^{\secparam-1}}$. Conditioned on this bad event, the output distributions of $\hybrid_{2.2.\ell_w.\secparam-1.k}^{\adversary}$ and $\hybrid_{2.3.k}^{\adversary}$ are identical. Thus, $\adversary$ cannot distinguish the hybrids $\hybrid_{2.2.\ell_w.\secparam-1.k}^{\adversary}$ and $\hybrid_{2.3.k}^{\adversary}$.\\
 \end{proof}

\noindent \underline{$\hybrid_{2.4.i^*.k}^{\adversary}$ for all $i^* \in [\ell_w],k \in [Q]$}: This hybrid is the same as $\hybrid_{2.3.k}^{\adversary}$ except that the hybrid prover $\hybrid_{2.4.i^*.k}.P$ is additionally parameterized by $\left( \left\{ \widehat{b}_{i,\secparam} \right\}_{i \leq i^*} \right)$. The only change from the previous hybrid is that the hybrid prover, for $i \leq i^*$, use the input $(sh_{i,\secparam},sh_{i,\secparam})$ if $\widehat{b}_{i,\secparam} \neq b_{i,\secparam}$ or use $(\alpha_{i,\secparam},\alpha_{i,\secparam})$ if $\widehat{b}_{i,\secparam}=b_{i,\secparam}$. 
\par Now, consider a hybrid prover  $\hybrid_{2.4.i^*.k}.P$, parameterized by $\left( \left\{ \widehat{b}_{i,j} \right\}_{(i \leq i^*) \vee (j \leq \secparam-1)} \right)$, where $\left( \left\{ \widehat{b}_{i,j} \right\}_{(i,j) \leq (i^*,j^*)} \right)$, is defined to be such that the output distributions of $\hybrid_{2.4.i^*.k}.P$ and $\hybrid_{2.4.i^*-1.k}.P$ cannot be distinguished by $\adversary$. If such a hybrid prover does not exist, then abort. Otherwise, this hybrid prover interacts with the verifier and the output of this hybrid is set to be the output of the verifier. 

\begin{claim}
The hybrid $\hybrid_{2.4.i^*.k}$ aborts with negligible probability. 
\end{claim}

\noindent We omit the proof of the above claim since it uses the same inductive argument as the proof of Claim~\ref{clm:int:advspecific}. \\

\noindent \underline{$\hybrid_{2.5.k}$ for $k \in [Q]$}: We define a hybrid prover that does the following:
\begin{itemize}
    \item For $k' \leq k$, it chooses the input to the $(i,j)^{th}$ execution to be $(\alpha_{i,j},\alpha_{i,j})$, where $\alpha_{i,j}$ is sampled uniformly at random. 
    \item For $k' > k$, it chooses the inputs to the OT executions as done by the prover in $\hybrid_{2.1}$.\\
    \end{itemize}
    \par Conditioned on $\hybrid_{2.4.\ell_w.k}$ not aborting, the output distributions of  $\hybrid_{2.4.\ell_w.k}$ and $\hybrid_{2.5.k}$ are the same. This follows from the fact that if $\hybrid_{2.4.\ell_w}$ does not abort then the distribution of the inputs used in all the OT executions in the hybrids $\hybrid_{2.4.\ell_w.k}$ and $\hybrid_{2.5.k}$ are the same. Thus, $\adversary$ can disitnguish $\hybrid_{2.4.\ell_w.k}$ and $\hybrid_{2.5.k}$ only with negligible probability.\\

\noindent \underline{$\hybrid_{2.6}$}: This hybrid is the same as $\hybrid_3$, i.e. the output of the simulator. 
\par The output distributions of $\hybrid_{2.5.Q}$ and $\hybrid_{2.6}$ are identical.  \\

\noindent From the above hybrids, it follows that $\adversary$ can distinguish the hybrids $\hybrid_{2.1}$ and $\hybrid_{2.6}$ with only negligible probability.

\end{proof}

%% file: QPoK/qpok-cons.tex
\subsection{Construction of (Standalone) QZKPoK}
\label{sec:qzkpok}
\noindent We construct a (standalone)   QZKPoK $(P,V)$ for an NP relation $\rel(\lang)$. The following tools are used in our construction: 
\begin{itemize}
    \item A post-quantum statistical receiver-private oblivious transfer protocol, $\Pi_{\ot}=(\sender,\receiver)$\fullversion{   (Section~\ref{sec:statOT})} satisfying perfect correctness property. 
    
    \par We say that a transcript $\tau$ is valid with respect to sender's randomness $r$ and its input bits $(m_0,m_1)$ if $\tau$ can be generated with a sender that uses $r$ as randomness for the protocol and uses $(m_0,m_1)$ as inputs.
    %\item A perfectly binding, quantum concealing commitment scheme $\comm$ (Section~\ref{sec:prelims:commit}). 
    % \item A bounded concurrent QZK proof system $\Pi_{\zk}^{(1)}$ for $\rel(\lang_{\zk}^{(1)})$. We describe the relation $\rel(\lang_{\zk}^{(1)})$, parameterized by security parameter $\secparam$, below. 
   %  $$\rel \left(\lang_{\zk}^{(1)} \right)=\left\{\left((\tau_{\ot}^{i,j},i,j);\ \left(r_{\ot}^{(i,j)},sh_{i,j},b_{i,j} \right) \right)\ :\ \substack{\tau_{\ot}^{(i,j)} \text{ is valid w.r.t }\\ r_{\ot}^{(i,j)} \text{ and }((1-b_{i,j})sh_{i,j} + b_{i,j} \cdot \bot,\ b_{i,j}sh_{i,j} + (1-b_{i,j}) \cdot \bot)} \right\}\footnote{We define $0 \cdot \bot = 0$ and $1 \cdot \bot = \bot$.}$$
     
  %   \noindent In other words, the instances in the above relation consists of transcripts $\tau_{\ot}^{(i,j)}$ that are valid with respect to sender's randomness $r_{\ot}^{(i,j)}$ and input $(sh_{i,j},\bot)$ if $b_{i,j}=0$, or $(\bot,sh_{i,j})$ if $b_{i,j}=1$.
     
     \item A (standalone) QZK proof system $\Pi_{\zk}$ for $\rel(\lang_{\zk})$. We describe the relation $\rel(\lang_{\zk})$, parameterized by security parameter $\secparam$, below.
     
     $$\rel\left(\lang_{\zk} \right)=\Bigg\{\left(\left(x,\{\tau_{\ot}^{(i,j)},b_{i,j}\}_{i \in [\ell_w],j \in [\secparam]} \right);\ \left(w,\left\{r_{\ot}^{(i,j)},sh_{i,j},\alpha_{i,j} \right\}_{i \in [\ell_w],j \in [\secparam]}\right) \right)\ :\ $$
     $$ \left(\substack{\forall i \in [\ell_w],j \in [\secparam],\\ \tau_{\ot}^{(i,j)} \text{ is valid w.r.t }\\
     r_{\ot}^{(i,j)} \text{ and }
     (((1-b_{i,j})sh_{i,j} + b_{i,j} \cdot \alpha_{i,j}),\  (b_{i,j}sh_{i,j} + (1-b_{i,j}) \cdot \alpha_{i,j}))} \right) \bigwedge \left(\substack{\forall i \in [\ell_w],\\ \oplus_{j=1}^{\secparam} sh_{i,j} = w_i} \right) \bigwedge (x,w) \in \rel(\lang) \Bigg\}$$
     In other words, the relation checks if the shares $\{sh_{i,j}\}$ used in all the OT executions so far are defined to be such that the XOR of the shares $sh_{i,1},\ldots,sh_{i,\secparam}$ yields the bit $w_i$. Moreover, the relation also checks if $w_1 \cdots w_{\ell_w}$ is the witness to the instance $x$. 
\end{itemize}

\noindent We describe the construction in Figure~\ref{fig:cqzkpok}. 

\begin{figure}[!htb]
   \begin{center}
   \begin{tabular}{|p{16cm}|}
    \hline \\
    {\bf Input of $P$}: Instance $x \in \lang$ along with witness $w$. The length of $w$ is denoted to be $\ell_w$.  \\
    {\bf Input of $V$}: Instance $x \in \lang$. \\
    \begin{itemize}\setlength\itemsep{1em}
        \item For every $i \in [\ell_w]$, $P$ samples the shares $sh_{i,1},\ldots,sh_{i,\secparam}$ uniformly at random conditioned on $ \oplus_{j=1}^{\secparam} sh_{i,j} = w_i$, where $w_i$ is the $i^{th}$ bit of $w$.
        \item For every $i \in [\ell_w]$, $P$ samples the bits $\alpha_{i,1},\ldots,\alpha_{i,\secparam}$ uniformly at random. 
       % \par $P$ commits to $sh_{i,j}$ using $\comm$ to obtain the commitment $\bfc_{i,j}$, for every $i \in [\ell_w],j \in [\secparam]$. Call the tuple of commitments $\{\bfc_{i,j}\}_{i \in [\ell_w],j \in [\secparam]}$ to be $\overrightarrow{\bfc}$. 
        
        \item For $i \in [\ell_w],j \in [\secparam]$, do the following: 
        \begin{itemize} 
        
         \item $P \leftrightarrow V$: $P$ and $V$ execute $\Pi_{\ot}$ with $V$ playing the role of the receiver in $\Pi_{\ot}$ and $P$ playing the role of the sender in $\Pi_{\ot}$. The input of the receiver in this protocol is 0, while the input of the sender is set to be $(sh_{i,j},\alpha_{i,j})$ if $b_{i,j}=0$, otherwise it is set to be $(\alpha_{i,j},sh_{i,j})$ if $b_{i,j}=1$, where the bit $b_{i,j}$ is sampled uniformly at random. 
        \par Call the resulting transcript of the protocol to be $\tau_{\ot}^{(i,j)}$ and let $r_{\ot}^{(i,j)}$ be the randomness used by the sender in $\ot$. 
        \item $P \rightarrow V$: $P$ sends $b_{i,j}$ to $V$. 
        
        %\item $P \leftrightarrow V$: $P$ and $V$ execute $\Pi_{\zk}^{(1)}$ with $P$ playing the role of the prover of $\Pi_{\zk}^{(1)}$ and $V$ playing the role of the verifier of $\Pi_{\zk}^{(1)}$. The input of the prover in $\Pi_{\zk}^{(1)}$ is the instance $\left(\tau_{\ot}^{(i,j)},i,j \right)$ and the associated witness is $\left( r_{\ot}^{(i,j)},sh_{i,j},b_{i,j} \right)$. The input of the verifier is the instance $\left(\tau_{\ot}^{(i,j)},i,j \right)$. If the verifier in $\Pi_{\zk}^{(1)}$ rejects, then $V$ rejects.
        
        \end{itemize}
        
    \item $P \leftrightarrow V$: $P$ and $V$ execute $\Pi_{\zk}$ with $P$ playing the role of the prover of $\Pi_{\zk}$ and $V$ playing the role of the verifier of $\Pi_{\zk}$. The instance is $\left(x,\left\{ \tau_{\ot}^{(i,j)}, b_{i,j} \right\}_{i \in [\ell_w],j \in [\secparam]} \right)$ and the witness is $\left( w,\left\{r_{\ot}^{(i,j)},sh_{i,j}, \alpha_{i,j} \right\}_{i \in [\ell_w],j \in [\secparam]} \right)$. If the verifier in $\Pi_{\zk}$ rejects, then $V$ rejects.
        
    \end{itemize}
    \\ 
    \hline
   \end{tabular}
    \caption{Construction of (standalone) QZKPoK for NP.}
    \label{fig:cqzkpok}
    \end{center}
\end{figure}

%% file: QMA/newQMA.tex
\section{Bounded Concurrent QZK for QMA}
\label{sec:bcqzk:qma}
We show a construction of bounded concurrent QZK for QMA. Our starting point is the QZK protocol for QMA from \cite{BJSW16}, which constructs QZK for QMA from QZK for NP, a commitment scheme and a coin-flipping protocol. We first simplify the protocol of~\cite{BJSW16} as follows: their protocol requires security of the coin-flipping protocol to hold against malicious adversaries whereas we only require the security to hold against adversaries who don't deviate from the protocol specification. Once we simplify this step, the resulting protocol will satisfy the property that the QZK simulator only rewinds during the execution of the underlying simulator simulating the QZK protocol for NP. This modification makes it easier for us to extend this protocol to the bounded concurrent setting. We simply instantiate the underlying QZK for NP protocol with its bounded concurrent version. 
%We make the following observation about the proof of quantum zero-knowledge in~\cite{BJSW16}. Only two ingredients in their construction require rewinding by the QZK simulator: the coin-flipping protocol and the QZK for NP protocol. The rest of their simulation is straightline. \prab{this is a little confusing --} Furthermore, we can combine the coin-flipping step with the QZK for NP step, leaving us with a straightline simulation that only requires rewinding in one place (at the QZK for NP step). As we will show, it then suffices to instantiate the QZK for NP with the bounded concurrent QZK for NP protocol in Section~\ref{sec:BCQZK}. %First we construct a bounded concurrent coin-flipping protocol. Then we will show how to use this to construct a bounded concurrent QZK for QMA. 

\subsection{Bounded Concurrent QZK for QMA}
We first recall the QZK for QMA construction from \cite{BJSW16}. Their protocol is specifically designed for the QMA promise problem called $k$-local Clifford Hamiltonian, which they showed to be QMA-complete for $k = 5$. We restate it here for completeness.

\begin{definition}
[$k$-local Clifford Hamiltonian Problem~\cite{BJSW16}] For all $i \in [m]$, let $H_i = C_i \ket{0^{\otimes k}} \bra{0^{\otimes k}}C^{\dagger}_i$ be a Hamiltonian term on $k$-qubits where $C_i$ is a Clifford circuit. 
\begin{itemize}
    \item Input: $H_1,H_2, \ldots, H_m$ and strings $1^p$, $1^q$ where $p$ and $q$ are positive integers satisfying $2^p > q$.
    \item Yes instances ($A_{yes}$):  There exists an $n$-qubit state such that $\tr[\rho \sum_i H_i]  \leq 2^{-p}$
    \item No instances ($A_{no}$): For every $n$-qubit state $\rho$, the following holds:  $\tr[\rho \sum_i H_i] \geq \frac{1}{q}$
    
\end{itemize}
\end{definition}

\paragraph{BJSW Encoding.} A key idea behind the construction from \cite{BJSW16} is for the prover to encode its witness, $\ket{\psi}$, using a secret-key quantum authentication code (that also serves as an encryption) that satisfies the following key properties needed in the protocol. For any state $\ket{\psi}$, denote the encoding of $\ket{\psi}$ under the secret-key $s$ by $\E_s(\ket{\psi})$. 

\begin{enumerate}

\item {\em Homomorphic evaluation of Cliffords.} Given $\E_s(\ket{\psi})$, and given any Clifford circuit $C$, it is possible to compute $\E_{s'}(C\ket{\psi})$ efficiently. Moreover, $s'$ can be determined efficiently by knowing $C$ and $s$. %\prab{come up $s'$}\\

\item {\em Homomorphic measurements of arbitrary Clifford basis.} For any Clifford circuit $C$ and any state $\ket{\psi}$, a computational basis measurement on $C\ket{\psi}$ can be recovered from a computational basis measurement on $\E_{s'}(C\ket{\psi})$ along with $C$ and $s$. Formally, there is a classically efficiently computable function $g$ such that if $y$ is sampled from the distribution induced by measuring the state $\E_{s'}(C\ket{\psi})$ in the computational basis, then $g(s,C,y)$ is sampled from the distribution induced by measuring the state $C\ket{\psi}$ in the computational basis.

\item {\em Authentication of measurement outcomes.} For any $s$ and any clifford $C$, there is a set $\cS_{s,C}$ such that for any state $\ket{\psi}$, and any computational basis measurement outcome $y$ performed on $\E_{s'}(C \ket{\psi})$, it holds that $y \in \cS_{s,C}$. Furthermore, for any $y$, given $s$ and $C$, it can be efficiently checked whether $y \in \cS_{s,C}$.

\item {\em Simulatability of authenticated states}: there exists an efficient QPT algorithm $B$ such that for any adversary $\adversary$, every $x \in A_{\yes}$ along with witness $\ket{\psi}$, $\poly(\secparam)$-qubit advice $\rho$, the following holds: the probability that $\cP(s,C_{r^*}^{\dagger},\adversary(\E_s(\ket{\psi})))$ outputs 1 is negligibly close to the probability that $\cP(s,C_{r^*}^{\dagger},\adversary(B(x,s,r^*))))$ outputs 1, where $\cP$ is defined below. %We will assume without loss of generality that $B$ outputs a pure state. 
$$\cP(s,C^\dagger,y)=\begin{cases}
1 \text{ if } g(s,C^\dagger,y) \neq 0 \\
0 \text{ otherwise}\end{cases}$$

In both the events, $s$ and $r^*$ are chosen uniformly at random. 

\end{enumerate}

 \noindent The QMA verifier of the $k$-local Clifford Hamiltonian problem measures terms of the form $C \ket{0^{\otimes k}} \bra{0^{\otimes k}}C^\dagger$ where $C$ is a Clifford circuit on a witness $\ket{\psi}$. Specifically, a verifier will first apply $C^\dagger$ and then measure in the computational basis. If the outcome of the measurement is the $0$ string, it rejects.  Otherwise, it accepts.  In the zero-knowledge case, the witness will be encoded, $\E_s(\ket{\psi})$, but the verifier can still compute $\E_{s}(C^\dagger \ket{\psi})$ and measure to obtain some string $y$.  Then, the prover can prove to the verifier (in NP) that $y$ corresponds to a non-zero outcome on a measurement of $C^\dagger \ket{\psi}$ instead using the predicate $\cP$.

\par We follow the approach of BJSW~\cite{BJSW16}, except that we instantiate the coin-flipping protocol in a specific way in order to get concurrency when instantiating the underlying QZK for NP with our bounded concurrent construction. 

\paragraph{Construction.} We use the following ingredients in our construction: 
\begin{itemize}
  \item Statistical-binding and quantum-concealing commitment scheme, $(\comm,\receiver)$\fullversion{ (Section~\ref{sec:prelims:commit})}.
   % \item Quantum-secure bounded concurrent QZKPoK system, denoted by $\Pi_{\np}^{(1)}$,  for the following language. 
   % $$\lang^{(1)}=\left\{ \left((\bfr,\bfc,s) ;\ (\ell) \right)\ :\  \comm(1^{\secparam},\bfr,s;\ell) = \bfc \right\}$$
   % \prab{need to get rid of this}
  
  %\item Post-quantum bounded-concurrent coin flipping protocol (Section~\ref{sec:coinflip}).
    %Let $Q$ be the maximum number of sessions associated with this protocol.
  
  \item Bounded concurrent QZK proof system, denoted by $\Pi_{\np}$, for the following language (Section~\ref{sec:BCQZK}). 
    $$\lang=\left\{ \left((\bfr,\bfc,\bfr',\bfc',r^*,y,b)\ ;\ (s,\ell,a,\ell') \right)\ :\ \substack{ \cP(s,C_{r^*}^{\dagger},y)=1 \\ \bigwedge \\ \comm(1^{\secparam},\bfr,s;\ell) = \bfc \\
    \bigwedge \\
    \comm(1^{\secparam},\bfr',a;\ell') = \bfc'
    \\  \bigwedge \\
    a\oplus b = r^*} \right\}$$
    Let $Q$ be the maximum number of sessions associated with the protocol. 
\end{itemize}

\noindent We describe the construction of bounded concurrent QZK for QMA (with bound $Q$) in Figure~\ref{fig:qmazk}. We prove the following.

%\paragraph{BJSW Protocol.} \ \\

\begin{figure}[!htb]
   \begin{tabular}{|p{16cm}|}
   \hline \\
\par \textbf{Instance:} A $k$-local Clifford Hamiltonian, $H = \overset{M}{\underset{r=1}{\sum}} C_r \ket{0^{\otimes k}} \bra{0^{\otimes k}} C_r^{\dagger}$.
\par \textbf{Witness:} $\ket{\psi}$
\begin{itemize}\setlength\itemsep{1em}
    \item $P \leftrightarrow V$: Prover $P$  samples a secret-key $s \xleftarrow{\$} \{0,1\}^{\poly(k,M)}$, and commits to $s$ using the commitment protocol $(\comm,\receiver)$. Let $\bfr$ be the first message of the receiver (sent by $V$) and $\bfc$ be the commitment. 
    \par {\em // We call this commitment, the secret-key commitment.}%\prab{what is secret-key commitment?} \prab{should we mention this is statistical binding commitment?}
    \item $P \rightarrow V$: $P$ sends $\E_{s}(\ket{\psi})$.
        \item $P \leftrightarrow V$: Prover samples a random string $a \xleftarrow{\$} \{0,1\}^{\log(M)}$, and commits to $a$ using the commitment protocol $(\comm, \receiver)$. Let $\bfr'$ be the first message of the receiver and $\bfc'$ be the commitment.\par {\em // We call this commitment, the coin-flipping commitment.}
    \item $V \rightarrow P$: Verifier samples a random string $b \xleftarrow{\$} \{0,1\}^{\log(M)}$.  Verifier sends $b$ to the prover.
    \item $P \rightarrow V$:  Prover sends $r^* \coloneqq a\oplus b$ to the verifier.
    
    \item Verifier computes $\eval \left(C_{r^*}^{\dagger}, \E_{s}(\ket{\Psi})\right) \rightarrow \E_{s}(C_{r^*}^{\dagger} \ket{\psi})$ and measures in the computational basis. Let $y$ denote the measurement outcome. Verifier sends $y$ to the prover.
    \item Prover checks that $y\in \cS_{s,C_{r^*}^{\dagger}}$ and that $\cP(s,C_{r^*}^{\dagger},y)=1$. If not, it aborts. 
    \item Prover and verifier engage in a QZK protocol for $\np$, $\Pi_{\np}$, for the 
    statement $(\bfr,\bfc, \bfr',\bfc', r^*,y,b)$ and the witness $(s,\ell, a,\ell')$. 
  %  following statement:
   % $$\exists (s,\ell) \text{ s.t }  \left(\cP(s,C_{r^*}^{\dagger},y)=1\right) \bigwedge \left( c=\comm(s;\ell)\right) $$
\end{itemize}
\\ 
\hline
\end{tabular}
\label{fig:qmazk}
\caption{Bounded-Concurrent QZK for QMA}
\end{figure}

\begin{theorem}
Assuming that $\Pi_{\np}$ satisfies the definition of bounded concurrent QZK for NP, the protocol given in Figure~\ref{fig:qmazk} is a bounded concurrent QZK protocol for $\qma$ with soundness $\frac{1}{\poly}$.
\end{theorem}

\begin{remark}
The soundness of the above protocol can be amplified by sequential repetition. In this case, the prover needs as many copies of the witness as the number of repetitions.
\end{remark}

\begin{proof}[Proof Sketch]

Completeness follows from~\cite{BJSW16}. 
%The only difference in our scheme and BJSW is that we combine the coin-flipping with the QZK for NP. This means that completeness follows from~\cite{BJSW16}. \prab{change this}

\paragraph{Soundness.} Once we argue that $r^*$ produced in the protocol is uniformly distributed, even when the verifier is interacting with the malicious prover, we can then invoke the soundness of ~\cite{BJSW16} to prove the soundness of our protocol. 
\par Suppose the verifier accepts the $\Pi_{\np}$ proof produced during the execution of the above protocol. From the soundness of $\Pi_{\np}$, we have that $r^* = a\oplus b$ where $a$ is the string that the prover initially committed to in $\bfc'$. By the statistical binding security of the commitment, and the fact that $b$ is chosen at random after $a$ has been committed to, we have that $r^*$ is sampled uniformly from $[M]$.

%We sketch the proof of bounded concurrent zero-knowledge below.

\paragraph{Bounded-Concurrent Quantum Zero-Knowledge.} Suppose $x \in A_{\yes}$. Suppose $V^*$ is a non-uniform malicious QPT  $Q$-session verifier. Then we construct a QPT simulator $\simr$ as follows. \\

\noindent \underline{Description of $\simr$}: it starts with the registers $\bfX_{zk},\bfX_{anc},\bfM,\auxreg$. The register $\bfX_{zk}$ is used by the simulator of the bounded concurrent QZK protocol, $\bfX_{anc}$ is an ancillary register, $\bfM$ is used to store the messages exchanged between the simulator and the verifier and finally, the register $\auxreg$ is used for storing the private state of the verifier. Initialize the registers $\bfX_{zk},\bfM$ with all zeroes. Initialize the register $\bfX_{anc}$ with $(\bigotimes_{j=1}^Q\ket{s_j}\bra{s_j}) \otimes (\bigotimes_{j=1}^Q \ket{r^*_j}\bra{r^*_j}) \otimes (\bigotimes_{j=1}^Q \rho_j) \otimes \ket{0^{\otimes \poly}}\bra{0^{\otimes \poly}}$,  where $s_i,r_i^*$ are generated uniformly at random and $\rho_j \leftarrow B(x,s_j,r^*_j)$ is defined in bullet 4 under BJSW encoding. 
\par $\simr$ applies the following unitary for $Q$ times on the above registers. This unitary is defined as follows: it parses the message $((1,\msg_1),\ldots,(Q,\msg_Q))$ in the register $\bfM$. For every round of conversation, it does the following: if it is $V^*$'s turn to talk, it applies $V^*$ on $\auxreg$ and $\bfM$. Otherwise,  
\begin{itemize}
    \item Let $S_1$ be the set of indices such that for every $i \in S_1$, $\msg_i$ is a message in the protocol $\Pi_{\np}$. Finally, let $S_2 = [Q] \backslash S_1$.  
    \item It copies $((1,\msg_1),\ldots,(Q,\msg_Q))$ into $\bfX_{zk}$ (using many CNOT operations) and for every $i \notin S_1$, replaces $\msg_i$ with N/A.  We note that $\msg_i$ is a quantum state (for instance, it could be a superposition over different messages). 
    
        \item For every $i \in S_2$, if  $\msg_i$ is the first prover's message of the $i^{th}$ session, then set $\msg'_i$ to be $\ket{\bfc_i}\bra{\bfc_i} \otimes \rho_i$, where $\bfc_i$ is the secret-key commitment of $0$.  If $\msg_i$ corresponds to the coin-flipping commitment, then set $\msg_i'$ to be $\ket{\bfc_i'}\bra{\bfc_i'}$ where $\bfc_i'$ is a commitment to $0$.  
    \item It applies the simulator of $\Pi_{\np}$ on $\bfX_{zk}$ to obtain $ ((1,\msg'_{1,zk}),\ldots(Q,\msg'_{Q,zk}))$. The $i^{th}$ session simulator of $\Pi_{\np}$ takes as input $(\bfr_i,\bfc_i,\bfr_i',\bfc_i',r^*_i,y_i, b_i)$, where $r^*_i$ was generated in the beginning and $\bfr_i,\bfc_i,\bfr_i',\bfc_i',y_i, b_i$ are generated as specified in the protocol.

    \item Determine $((1,\msg'_1),\ldots,(Q,\msg'_Q))$ as follows. Set $\msg'_i = \msg_{i,\zk}$, if $i \in S_1$. Output of this round is $((1,\msg'_1),\ldots,(Q,\msg'_Q))$. 
\end{itemize}

\noindent We claim that the output distribution of $\simr$ (ideal world) is computationally indistinguishable from the output distribution of $V^*$ when interacting with the prover (real world). \\

\noindent \underline{$\hybrid_1$}: This corresponds to the real world. \\

\noindent \underline{$\hybrid_2$}: This is the same as $\hybrid_1$ except that the verifier $V^*$ is run in superposition and the transcript is measured at the end. 
\par The output distributions of $\hybrid_1$ and $\hybrid_2$ are identical. \\

\noindent \underline{$\hybrid_{3}$}: Simulate the zero-knowledge protocol $\Pi_{\np}$ simultaneously for all the sessions. Other than this, the rest of the hybrid is the same as before.  
\par The output distributions of $\hybrid_2$ and $\hybrid_3$ are computationally indistinguishable from the bounded concurrent QZK property of $\Pi_{\np}$. \\

\noindent \underline{$\hybrid_{4.i}$ for $i \in [Q]$}: For every $j \leq i$, the coin-flipping commitment in the $j^{th}$ session is a commitment to 0 instead of $a_i$. For all $j>i$, the commitment is computed as in the previous hybrid.
\par The output distributions of $\hybrid_{4.i-1}$ (or $\hybrid_3$ if $i=1$) and $\hybrid_{4.i}$ are computationally indistinguishable from the quantum concealing property of $(\comm,\receiver)$. \\

\noindent \underline{$\hybrid_{5.i}$ for $i \in [Q]$}: For every $j \leq i$, the secret-key commitment in the $j^{th}$ session is a commitment to 0. For all $j > i$, the commitment is computed as in the previous hybrid. 
\par The output distributions of $\hybrid_{5.i-1}$  (or $\hybrid_{4.Q}$ if $i=1$) and $\hybrid_{5.i}$ are computationally indistinguishable from the quantum concealing property of $(\comm,\receiver)$. \\

\noindent \underline{$\hybrid_{6.i}$ for $i \in [Q]$}: For every $j \leq i$, the encoding of the state is computed instead using $B(x,s_i,r^*_i)$, where $s_i,r^*_i$ is generated uniformly at random.
\par The output distributions of $\hybrid_{6.i-1}$ and $\hybrid_{6.i}$ are statistically indistinguishable from simulatability of authenticated states property of BJSW encoding (bullet 4). This follows from the following fact: conditioned on the prover not aborting, the output distributions of the two worlds are identical. Moreover, the property of simulatability of authenticated states shows that the probability of the prover aborting in the previous hybrid is negligibly close to the probability of the prover aborting in this hybrid. \\

\noindent $\underline{\hybrid_7}$: This corresponds to the ideal world. 
\par The output distributions of $\hybrid_{6.Q}$ and $\hybrid_7$ are identical.

\end{proof}

\paragraph{Proof of Quantum Knowledge with better witness quality.} We can define an anologous notion of proof of knowledge in the context of interactive protocols for QMA. This notion is called proof of {\em quantum} knowledge. See~\cite{CVZ20} for a definition of this notion. Coladangelo, Vidick and Zhang~\cite{CVZ20} show how to achieve quantum proof of quantum knowledge generically using quantum proof of classical knowledge. Their protocol builds upon~\cite{BJSW16} to achieve their goal. We can adopt their idea to achieve proof of quantum knowledge property for a bounded concurrent QZK for QMA system. In Figure~\ref{fig:qmazk}, include a quantum proof of classical knowledge system for NP (for instance, the one we constructed in Section~\ref{sec:qzkpok}) just after the prover sends encoding of the witness state $\ket{\Psi}$, encoded using the key $s$. Using the quantum proof of classical knowledge system, the prover convinces the verifier of its knowledge of the $s$. The rest of the protocol is the same as Figure~\ref{fig:qmazk}. To see why this satisfies proof of quantum knowledge, note that an extractor can extract $s$ with probability negligibly close to the acceptance probability and using $s$, can recover the witness $\ket{\Psi}$. 
\par For the first time, we get proof of quantum knowledge (even in the standalone setting) with $(1 - \negl)$-quality if the acceptance probability is negligibly close to 1, where the quality denotes the closeness to the witness state. Previous proof of quantum knowledge~\cite{BG19,CVZ20} achieved only $1 - \frac{1}{\poly}$ qualtiy; this is because these works use Unruh's quantum proof of classical knowledge technique~\cite{Unruh12} and the extraction probability in Unruh is not negligibly close to the acceptance probability.

%% file: Submission/submvertechnsections.tex
\section{Bounded Concurrent QZK for NP}\label{sec:bcqzknp}
\noindent We first give an overview of bounded concurrent QZK for NP.

\input{Overview/bcqzk-overview}
\input{Concurrent QZK/construction}
\noindent We present the proofs of completeness, soundness and quantum zero-knowledge in the full version. 

\section{Quantum Proofs of Knowledge}
\label{sec:qpok}
\noindent We first present a construction of standalone quantum proof of knowledge for NP. We extend this construction to the bounded concurrent setting in~\Cref{sec:bcqzk}.

\input{Overview/stdpok-overview}
\input{QPoK/qpok-cons}
\ \\
We present the proofs of completeness, quantum proof of knowledge and quantum zero-knowledge in the full version.

\input{Overview/concpok-overview}
\ \\
\noindent We present the formal details in Section~\ref{sec:bcqzk}.

\input{QMA/newQMA}